\documentclass[12pt]{article}

\usepackage[labelfont=bf]{caption}
\usepackage[utf8]{inputenc} 
\usepackage[T1]{fontenc}    
\usepackage{hyperref}       
\usepackage{url}            
\usepackage{booktabs}       
\usepackage{amsfonts}       
\usepackage{amsmath, bbm, amsthm} 
\usepackage{nicefrac}       
\usepackage{microtype}      
\usepackage{lipsum}		
\usepackage{graphicx}
\usepackage{caption}
\usepackage{subcaption}
\usepackage{tikz}
\usepackage{natbib}
\usepackage{doi}
\usepackage{enumerate}
\usepackage{todo}
\usepackage{romannum}
\usepackage{xr}
\usepackage{xcolor}
\usepackage{placeins}

\makeatletter
\newcommand*{\addFileDependency}[1]{
\typeout{(#1)}
\@addtofilelist{#1}
\IfFileExists{#1}{}{\typeout{No file #1.}}
}\makeatother

\theoremstyle{plain}
\newtheorem{theorem}{Theorem}

\newtheorem{proposition}{Proposition}
\newtheorem{lemma}{Lemma}
\newtheorem{corollary}{Corollary}

\theoremstyle{definition}
\newtheorem{definition}{Definition}
\newtheorem{example}{Example}

\theoremstyle{remark}
\newtheorem*{remark}{Remark}

\newcommand{\blind}{1}
\newcommand{\R}{\mathbb R} 
\newcommand{\N}{\mathbb N} 
\renewcommand{\P}{\mathbb P} 
\newcommand{\E}{\mathbb E} 
\renewcommand{\S}{\mathbb S}
 
\newcommand{\disteq}{\buildrel d \over =}
 
\newcommand{\aeeq}{\buildrel a.e. \over =} 
 
\newcommand{\Cov}{\mathrm{Cov}}

\title{Principal component analysis for max-stable distributions}

\date{}

\hypersetup{
pdftitle={Principal Component Analysis for max-stable distributions},
pdfauthor={F. Reinbott, A. Janssen},
pdfkeywords={},
}

\begin{document}
\pagenumbering{arabic}

\def\spacingset#1{\renewcommand{\baselinestretch}%
{#1}\small\normalsize} \spacingset{1}


\if1\blind
{
  \title{\bf Principal Component Analysis for max-stable distributions}
  \author{Felix Reinbott \thanks{
    The authors gratefully acknowledge financial support from \textit{Deutsche Forschungsgemeinschaft (DFG, German Research Foundation) - 314838170, GRK 2297 MathCoRe. }}\hspace{.2cm}\\
    Institute of Mathematical Stochastics, Department of Mathematics,\\ Otto von Guericke University Magdeburg\\
    and \\
    Anja Jan{\ss}en\\
    Institute of Mathematical Stochastics, Department of Mathematics,\\ Otto von Guericke University Magdeburg
    }

  \maketitle
} \fi

\if0\blind
{
  \bigskip
  \bigskip
  \bigskip
  \begin{center}
    {\LARGE\bf Principal Component Analysis for max-stable distributions}
\end{center}
  \medskip
} \fi

\noindent%
\vfill

\newpage
\spacingset{1.9} 

\begin{abstract}
	Principal component analysis (PCA) is one of the most popular dimension reduction techniques in statistics
	and is especially powerful when a multivariate distribution is concentrated near a lower-dimensional subspace. 
	Multivariate extreme value distributions have turned out to provide challenges for the application of PCA since their constraint support impedes the detection of lower-dimensional structures and heavy-tails can imply that second moments do not exist, thereby preventing the application of classical variance-based techniques for PCA. 
	We adapt PCA to max-stable distributions using a regression setting and employ max-linear maps to project the random vector to a lower-dimensional space while preserving max-stability. 
	We also provide a characterization of those distributions which allow for a perfect reconstruction from the lower-dimensional representation. Finally, we demonstrate how an optimal projection matrix can be consistently estimated and show viability in practice with a simulation study and application to a benchmark dataset.
	
\end{abstract}

\section{Introduction}

Multivariate, potentially high-dimensional data frequently emerges in various applications, 
and it is of fundamental interest to extract key features and patterns from datasets. 
For this task, a large variety of unsupervised learning algorithms have been developed and applied to different settings, 
for a general overview see for example~\cite{Bishop_pr} or \cite{HastieTibshirani_eosl}. 
One of the most frequently used techniques is principal component analysis (PCA), 
performing especially well if the data concentrates around a low dimensional linear subspace, see ~\cite{Hotelling_aocsv} and \cite{Jolliffe_pca}. 

One way to motivate PCA is via the representation of multivariate normal distributions. Recall that for a $d$-dimensional random variable $X$ following a multivariate normal distribution with mean $(0, \ldots, 0)$ and covariance matrix $\Sigma \in \mathbb{R}^{d \times d}$, there exist an orthogonal matrix $Q \in \mathbb{R}^{d \times d}$ and a diagonal matrix $D  \in [0,\infty)^{d \times d}$ such that
 $$ X \overset{d}{=} Q \cdot D \cdot Z,$$
where $Z$ is a $d$-dimensional vector of i.i.d. standard normal random variables and $\overset{d}{=}$ stands for equality in distribution. Both $Q$ and $D$ can be deduced from the singular value decomposition of $\Sigma$ and can be chosen in a way such that the diagonal entries of $D^2$ are the eigenvalues of $\Sigma$ in descending order. The columns of $Q$ are called the principal components of $X$. Now, $X$ lives on a lower-dimensional subspace (i.e.\ its covariance matrix does not have full rank), if and only if some diagonal entries of $D$ are 0. More precisely, if only $p<d$ diagonal entries of $D$ are positive, the random variable $X$ can be represented with the help of the lower dimensional random variable $(Z_1, \ldots, Z_p)$ consisting of only $p$ i.i.d.\ standard normal random variables. Furthermore, in order to approximate the distribution of $X$ by a lower-dimensional representation in terms of its first principal components, one can replace the smallest diagonal elements of $D$ by zeros. If applied to a data set instead of a known normal distribution, one first estimates the covariance matrix $\Sigma$ and then performs a singular value decomposition of the result. The method of PCA thus provides a lower-dimensional (approximate) representation of a data set, where the chosen principal components can be interpreted as latent factors, and it therefore offers a tool for complexity reduction, i.e. both data compression and pattern detection. 

In this article, we are interested in complexity reduction for data of extremes, and care has to be taken when one wants to transfer methods which were designed for the bulk of a data set to its maxima or high-level exceedances. Many classic unsupervised learning procedures have been adapted specifically for this task, see, e.g., \cite{Chautru_drmve, Chiapino_fcev, JanssenWan_kmce, DreesSabourin_pcame, Jalalzai_fcsier, Aghbalou_tir, Roettger_tpmve, CooleyThibaud_ddhde}, and also ~\cite{Engelke_ssme} for a recent overview article. 
To ensure that an approach is valid for extreme events possibly beyond the observed data, 
most methods rely on the framework of extreme value theory, which provides additional structure 
and allows for stability in applications and interpretability of results. 

Since the normal distribution is not suitable to describe tails of distributions, the aforementioned motivation of PCA fails when designing a similar approach for extremes. In particular, we want to provide methods for heavy-tailed distributions for which no covariance matrix exists. Furthermore (and this will prove to be a subtle, yet crucial point) extreme value distributions typically live on a restricted set like $[0,\infty)^d$ which impedes an orthogonal decomposition of its support. Different approaches to solve these two problems have been given by~\cite{DreesSabourin_pcame, CooleyThibaud_ddhde, JiangCooley_pcae, Avella_kpcamve}, 
using a suitable transformation of the data and then identifying the linear subspace closest to the data. {Approaches to PCA for extremes have also been proposed for functional data~\cite{Kokoszka_pcarvf, Kokozka_pcaivfd, Clemencon_rvhsfpca, Kokoszka_epft}, focusing on the sample covariance operator.}

We propose a new approach especially suited to the max-stability property of multivariate extreme value distributions that is based on a well-known equivalent formulation of PCA via a regression problem and shares some similarities with~\cite{Breiman_aa}. It allows a transfer from normal to more general distributions and involves in our case techniques from tropical linear algebra, combined with a suitable metric to measure how well the lower dimensional encoding can reconstruct the original data. {Note that by using this tropical setting, we are using an approach to PCA unrelated to Hilbert spaces, in contrast to the aforementioned references.} Tropical linear algebra has proven to be powerful in extremes, see~\cite{Amendola_cimln} and \cite{GissiblKlueppelberg_rmlm},
and tropical variants of PCA also have gained attention for phylogenetic trees, see ~\cite{Yoshida_tpcapt} and \cite{Page_tpcaspt}. We present several theoretical properties of our approach which motivate its application and aide in interpretation of results. Finally, we present a non-parametric estimator of the optimal projection matrix and show its consistency. 

This paper is structured as follows: 
Section~\ref{sec:background} recalls the necessary concepts from extreme value 
theory, tropical linear algebra and regular PCA. In Section~\ref{sec:defmaxpca} 
we motivate our definition of PCA for max-stable distributions based on an optimization problem that can be solved using data. In Section~\ref{sec:properties} we investigate properties of the minimizer of our optimization problem and give a complete characterization of distributions that can be mapped down to a lower dimensional space without loss of information in Theorem~\ref{thm:spectralrep}. In Section~\ref{sec:stats} we outline a framework for applying max-stable PCA to data, provide a consistency type result for our minimization problem in Theorem~\ref{thm:consistency} and we demonstrate practical viability with a simulation study and an application to a benchmark dataset. Proofs and additional details are deferred to the supplementary material for ease of presentation. 
 
\section{Background}
\label{sec:background}

\subsection{Notation and conventions}

In the main work and in the supplementary material, we write $\bigvee_{i=1}^\infty x_i:=\sup_{i \in \mathbb{N}}x_i, \bigvee_{i=1}^d x_i:=\max_{i=1}^d x_i$ and $\bigwedge_{i=1}^\infty x_i:=\inf_{i \in \mathbb{N}}x_i, \bigwedge_{i=1}^d x_i:=\min_{i=1}^d x_i$ for $x_i \in \mathbb{R}, i \in \mathbb{N}$. Vectors are understood as column vectors, with $x^T, A^T$ denoting the transpose of a vector $x$ or a matrix $A$. Inequalities between vectors are meant componentwise, i.e. $x \leq y$ for $x=(x_1, \ldots, x_d), y=(y_1, \ldots, y_d) \in \mathbb{R}^d$ means that $x_i \leq y_i$ for all $i=1, \ldots, d$, and $x \nleq y$ means that there exists an $i=1, \ldots, d$ such that $x_i > y_i$. Accordingly, taking suprema and infima over vectors is done componentwise. The $d$-dimensional vectors of all zeros and all ones are written as $\mathbf 0_d = (0, \ldots, 0)$ and $\mathbf 1_d = (1, \ldots, 1)$, respectively. Additionally, we denote the $k$-th canonical basis vector by $e_k$. We write $\mathbf 0_{m \times n}$ for an $m \times n$ all zero matrix and for an $n \times n$ identity matrix we use the notation $\mathrm{id}_{n}$. We denote the $i$-th row of a matrix $M$ by $M_i$, the $j$-th column by $M_{\cdot j}$ and the entry in row $i$ and column $j$ of $M$ by $M_{ij}$. The Borel sets on $\mathbb{R}^d$, restricted to a Borel set $A \subseteq \mathbb{R}^d$, are denoted by $\mathcal{B}(A)$. The complement of a set $A$ is given by $A^c$ and $\mathbbm 1_A$ is the indicator function of $A$ with $\mathbbm 1_A(x)=1$ if $x \in A$ and 0 if $x \in A^c$. The closure of a set $A$ is denoted by $\mbox{cl}(A)$. We denote a probability space by $(\Omega, \mathcal A, \P)$ and the law of some random variable or random vector $X$ as $\P^X$. Equality in distribution, as already mentioned in the introduction, is denoted by $\overset{d}{=}$ and in similar fashion we write equality almost everywhere as $\aeeq$. The Dirac measure in a point $x$ is denoted by $\delta_x$. The space of all non-negative, finitely integrable functions on $[0,1]$ is denoted by $L_+^1([0, 1])$. We write $\| \cdot \|$ for an arbitrary norm, $\| x \|_{\infty} := \bigvee_{i=1}^d \lvert x_i \rvert $ as the $\infty$-norm and denote by $\| x \|_1 := \sum_{k=1}^d \lvert x_k \rvert$ the sum norm of a vector $x \in \R^d$. The notations for the norms can also be used for matrices. Additionally, we use the norm $\| f \|_{L^1(E)} := \int_E \lvert f(e) \rvert de $ as the $L^1$ norm over a suitable domain $E$.  

\subsection{Max-stable distributions}\label{subsec:evt}

As we are interested in extremal observations of $d$-dimensional random vectors we shall subsequently work within the framework of multivariate extreme value theory, see, e.g. \cite{Beirlant_soe, DeHaanFerreira_evti, resnick_htp} for introductions on the topic. To this end, 
let us consider i.i.d.~random vectors $X_i := (X_{1i}, \ldots X_{di})^T, i \in \mathbb{N},$ and assume that there exist suitable scaling sequences $a_n = (a_{1n}, \ldots, a_{dn})^T \in (0, \infty)^d, b_n = (b_{1n}, \ldots, b_{dn})^T \in \R^d$, 
 and a distribution function $G$ such that
\begin{equation}\label{Eq:MEVD}
	\P\left( \bigvee_{i=1}^n\frac{X_{1i} - b_{1n}}{a_{1n}} \leq z_1, \ldots  ,\bigvee_{i=1}^n \frac{X_{di} - b_{dn}}{a_{dn}} \leq z_d \right) \to G(z_1, \ldots, z_d), \quad n \to \infty,
\end{equation}
holds in all continuity points $z=(z_1, \ldots, z_d)^T \in \R^d$ of $G$. If additionally the marginals of $G$ are non-degenerate, we call $G$ a multivariate extreme value distribution (MEVD) and by the Fisher-Gnedenko-Tippett theorem it follows that, up to a linear transformation, 
all marginal distribution functions are univariate extreme value distributions. When the interest is, as in our case, in the extremal dependence structure of a random vector $X=(X_1, \ldots, X_d)^T$, it is common practise (see, e.g., \cite{DeHaanFerreira_evti}, Section 6.1.2) to first standardize all its marginal distributions. This can be done in a way to ensure convergence of linearly standardized maxima from i.i.d.\ copies of it to a multivariate extreme value distribution $G^\ast$, such that the marginal distribution functions of $G^\ast$ satisfy
 \begin{equation*}
 G_j^\ast (y)= \exp(-\sigma_j/y)\mathbbm 1_{[0,\infty)}(y), \end{equation*}
where we interpret $c/0=\infty$ for $c>0$ and $0/0=0$. We call these marginal distributions 1-Fréchet with scale parameter $\sigma_j$. The resulting joint distribution function $G^\ast$ has the property that for $n$ i.i.d.\ random vectors $Y_i=(Y_{1i}, \ldots, Y_{di})^T, i = 1, \ldots, n,$ with distribution function $G^\ast$, there exist constants $c_n = (c_{1n}, \ldots, c_{dn})^T \in (0, \infty)^d$ such that 
$$ \left(\bigvee_{i=1}^n\frac{Y_{1i}}{c_{1n}}, \ldots, \bigvee_{i=1}^n \frac{Y_{di}}{c_{dn}}\right)^T \overset{d}{=}(Y_{11}, \ldots, Y_{d1})^T.$$
 This is a particularly simple and convenient form (including only a multiplicative standardization) of the property of max-stability, see \cite{deHaanResnick_ltme} and \cite{DavisResnick_psmsp}, and such distributions will be the object of our subsequent analysis. 
\begin{definition}\label{Def:maxstable}
We call a limiting distribution function $G$ in \eqref{Eq:MEVD} with marginal distribution functions
\begin{equation}\label{Eq:maxstablemargin} G_j(y)=\exp(-\sigma_j/y)\mathbbm 1_{[0,\infty)}(y),
\end{equation}
for $\sigma_j\geq 0, j = 1, \ldots, d,$ a \textbf{max-stable distribution with 1-Fréchet margins}. If a random vector $X \in [0,\infty)^d$ has such a distribution function $G$, then we also call $X$ \textbf{max-stable with 1-Fréchet margins}. If $\sigma_j>0$, then we call the $j$th marginal distribution \textbf{non-degenerate}. 
\end{definition}
Note that the use of the term ``max-stability'' differs in the literature and can also be understood in a more general way (that is including all possible limits in \eqref{Eq:MEVD} without an assumption about the marginal distributions) or a more restricted one (by assuming that $\sigma_j=1, j=1, \ldots, d$, when one also often speaks of \textit{simple} max-stable distribution). The reason for choosing Definition~\ref{Def:maxstable} in our context is that it allows for a unified treatment by fixing the extreme value index of all marginals while at the same time ensuring that max-stable distributions are closed under max-linear transformations, see Lemma~\ref{lem:maxlin_trafo} in the supplementary material for additional details, and compare also \cite{DavisResnick_psmsp}.

In contrast to univariate extreme value distributions, which are essentially determined by their extreme value index (in addition to a location and scale parameter), the class of all multivariate extreme value distributions cannot be parameterized by a finite set of parameters. Nevertheless, there exist several useful representation results, which we gather below.  First, it can be shown (see~\cite[Proposition 5.11]{resnick_erp} for the case where $\sigma_1=\ldots=\sigma_d=1$ and \cite[Corollary 2.7]{Janssen_tdme} for the general one) that a distribution function $G$ is max-stable with 
1-Fréchet margins if and only if there exists a bounded measure $S$ on $\mathcal{B}(\S^{d-1}_+)$ with $\S^{d-1}_+ := \{x \in [0, \infty)^d \colon \| x\| = 1 \}$ (for some arbitrary but fixed norm $\| \cdot \|$) 
such that
\begin{equation}\label{Eq:spectral_meas_rep}
	G(z) = \exp \left( - \int_{\S^{d-1}_+} \bigvee_{i=1}^d \frac {a_i}{z_i} \, S(da) \right) \mathbbm 1_{[0, \infty)^d}(z). 
\end{equation}
The measure $S$ is called the \textit{spectral measure}. Equivalently, see \cite{DeHaan_srmsp}, there exist functions $f_1, \ldots, f_d \in L_+^1([0, 1])$, called \textit{spectral functions}, 
such that $G$ can be written as 
\begin{equation}\label{Eq:spectral_fun_rep}
	G(z) = \exp \left( - \int_{[0,1]} \bigvee_{i=1}^d \frac {f_i(e)}{z_i} \, de \right) \mathbbm 1_{[0, \infty)^d}(z). 
\end{equation}
Note that \eqref{Eq:maxstablemargin}, \eqref{Eq:spectral_meas_rep} and \eqref{Eq:spectral_fun_rep} imply that, for all $y>0$,
\begin{equation}
\label{eq:scale-par-reps}
\sigma_j = -y \log(G_j(y))=\int_{[0,1]} f_j(e) \, de = \int_{\S^{d-1}_+} a_j \, S(da), \;\;\;  j = 1, \ldots, d.
\end{equation}
Finally, for our approach to introduce max-stable PCA via a regression problem, we need a way to measure a distance between different components of a max-stable distribution. As \cite{Davis_maxARMA, DavisResnick_psmsp} and \cite{StoevTaqqu_esi} show, a metric for the components of a bivariate max-stable random vector $(X_1,X_2)$ with spectral functions $f_1, f_2$ is given by 
\begin{equation}
	\label{eq:univmetric}
	\rho(X_1,X_2) :=  \int_{[0,1]} \lvert f_1(e) - f_2(e) \rvert  de = \|f_1-f_2\|_{L^1([0,1])},
\end{equation}
(which is well-defined even if the choice of $f_1, f_2$ in~\eqref{Eq:spectral_fun_rep} is not unique, see \cite[Section 3]{Davis_maxARMA} for details) and it metrizes convergence in probability with $\rho(X_1,X_2)=0$ implying that $\P(X_1=X_2)=1$. This metric has been named \textit{spectral distance} in \cite{Janssen_tdme}. 
It can easily be extended to a metric for two random vectors $X=(X_1, \ldots, X_d)^T, Y=(Y_1, \ldots, Y_d)^T$ for which $(X_1, \ldots, X_d,Y_1, \ldots, Y_d)^T$ is a max-stable random vector by setting 
\begin{equation}
	\label{eq:mmetric}
	\tilde \rho(X, Y) := \sum_{k=1}^d \rho(X_k, Y_k).
\end{equation}
By the above, the metric $\tilde{\rho}$ metrizes convergence in probability for max-stable random vectors and the statement $\tilde{\rho}(X,Y)=0$ is equivalent to $\P(X=Y)=1$. It is thus a suitable metric for measuring how well we can approximate a max-stable random vector by a transformation of it which preserves the max-stability of the random vector but is low-dimensional in a suitable sense. These transformations will be introduced in Section~\ref{subsec:tropalg}.

\subsection{PCA as a regression problem}\label{subsec:regPCA}

An alternative motivation for PCA than the one given in the introduction, is the search for an optimal linear subspace of dimension $p < d$ 
to project a $d$-variate random vector $X$ upon, where optimality is measured in terms of the second moment of the projection error. It is thus equivalent to finding a solution to the regression problem 
\begin{equation}
	\label{eq:pca_regr}
	\min_{(B,W) \in\R^{d \times p} \times \R^{p \times d}} \sum_{k=1}^d \E \left[(B_k W X - X_k)^2 \right],
\end{equation}
where $B_k, k = 1, \ldots, d,$ denotes the $k$th row of $B$, i.e.\ in finding an optimal matrix $W$ which maps our random vector $X$ to a $p$-dimensional subspace and its optimal counterpart $B$ which reconstructs $X$ from this lower dimensional representation.  
If $\E(X_i)=0, i = 1, \ldots, d,$ and the covariance matrix $\Sigma := (\Cov(X_i, X_j))_{i,j = 1, \ldots, d}$ of $X$ exists, then the orthonormal eigenvectors $v_1, \ldots, v_d$ of $\Sigma$, corresponding to the eigenvalues $\lambda_1 \geq \ldots \geq \lambda_d$ can be used to find the optimal matrices as $B = (v_{ij})_{i= 1, \ldots, d, j = 1, \ldots p}$ and $W = B^T$, for a given value of $p$. For a proof see~\cite{Hornik_nnpca}. 

PCA has been adapted to many different settings, see for example~\cite{DreesSabourin_pcame, Hornik_nnpca, Schlather_sgpca}, and our aim in the following shall be to find an analogue to \eqref{eq:pca_regr} which is  suitable to max-stable random vectors. A first component of this approach has already been found in the metric~\eqref{eq:mmetric} which shall replace the mean squared distance in \eqref{eq:pca_regr}. Second, we want to make sure that our reconstructed vector is of the same type as the original one. Therefore, in order to stay within the class of max-stable distributions, we replace in \eqref{eq:pca_regr} the usual matrix product by a max-linear analogue under which max-stable distributions are closed. Details are given in the following subsection. 

\subsection{Tropical linear algebra and max-stable distributions}\label{subsec:tropalg}

If $X=(X_1, \ldots, X_d)^T$ is max-stable with 1-Fréchet margins with distribution function $G$ according to \eqref{Eq:spectral_meas_rep} and \eqref{Eq:spectral_fun_rep}, respectively, and if $c_i \geq 0, i = 1, \ldots, d,$ are constants, then it follows from the above representations that
\begin{align}
\P \left( \bigvee_{i = 1}^d c_i X_i \leq z \right)
\label{Eq:univ_maxlinS}		&= \exp \left(- \frac 1 z \int_{[0,1]} \bigvee_{i=1}^d c_i a_i \, S(da) \right) \mathbbm 1_{[0, \infty)}(z) \\
\label{Eq:univ_maxlinf}	    &= \exp \left(- \frac 1 z \int_{[0,1]} \bigvee_{i=1}^d c_i f_i(e) \, de \right) \mathbbm 1_{[0, \infty)}(z), \;\;\; z \in \mathbb{R}, 
\end{align}
see also \cite{DeHaan_cmev}. Thus, a univariate max-linear combination of the components of a max-stable random vector has a 1-Fr\'{e}chet-distribution with scale coefficient $\int_{[0,1]} \bigvee_{i=1}^d c_i f_i(e) \, de$.
More generally, we can write multiple max-linear combinations of $X$ given by a matrix $C \in [0, \infty)^{p \times d}$ 
conveniently as 
\begin{equation}\label{Eq:max_matrix:prod}
	C \diamond X := \left(\bigvee_{l = 1}^d C_{jl} X_l \right)_{j = 1, \ldots, p }. 
\end{equation}
We show in Lemma~\ref{lem:maxlin_trafo} that the resulting vector has again a max-stable distribution with 1-Fr\'{e}chet margins and explore further properties in the supplementary material. 
Accordingly, we define a max-linear matrix product by 
\begin{equation}\label{eq:maxmatprod}
  \diamond \colon [0, \infty)^{d \times p} \times [0, \infty)^{p \times k} \to [0, \infty)^{d \times k}, \quad (B, W) \mapsto B \diamond W := \left(\bigvee_{l=1}^p B_{il} W_{lj} \right)_{i=1, \ldots, d, j = 1, \ldots, k}. 
\end{equation}
Note that this max-linear matrix product is a bilinear map with respect to the operations $(\vee, \cdot)$, where the maximum operation acts as addition, combined with regular scalar multiplication.  This framework has been extensively studied in the algebraic literature (often, after a logarithmic transformation, in the isomorphic setting of the operations $(\vee, +)$) under the name tropical linear algebra.  For further reading, see the monographs~\cite{Butkovic_mls} and~\cite{MacLaganSturmfels_ittg} and for related applications of tropical linear algebra to extremes, see, e.g. \cite{Amendola_cimln, GissiblKlueppelberg_rmlm, Klueppelberg_rmlmpn}. 

\section{Definition of max-stable PCA}
\label{sec:defmaxpca}
The aim of this section is to develop a dimension reduction procedure for max-stable distributions that shares many similarities with classic PCA. 
With the regression approach to PCA from Section~\ref{subsec:regPCA} in mind, 
we recall that the  metric $\tilde \rho$ from \eqref{eq:mmetric} measures the distance between two 
$d$-variate max-stable random vectors with $1$-Fréchet margins living on the same probability space.
We have shown that we can map a $d$-variate random vector $X$ with $1$-Fréchet margins to 
a $p$-variate max-stable random vector also with $1$-Fréchet margins
by using a matrix $W \in [0, \infty)^{p \times d}$ as follows
\begin{equation} \label{eq:encform} 
  X \mapsto W \diamond X.  
\end{equation} 
We call the mapping~\eqref{eq:encform} an \textit{encoding of $X$ in $p$ variables}.  
If $p > d$, by choosing the first $d$ columns of $W$ as canonical basis vectors, $X$ is a subvector of $W \diamond X$ and the remaining $p - d$ elements in $W \diamond X$ add no additional information about $X$. 
Therefore, we will only consider the case $p \leq d$. 
Mapping the encoding back to $[0, \infty)^d$, by using a matrix $B \in [0, \infty)^{d \times p}$ as follows 
\begin{equation} 
  \label{eq:recform} 
  X \mapsto B\diamond W \diamond X, 
\end{equation} 
we obtain a $d$-variate max-stable random vector with $1$-Fréchet margins again.  
We call~\eqref{eq:recform} a \textit{reconstruction of $X$ using the matrix $H = B \diamond W$}. 
{Note that the matrix pair $(B,W)$ used to construct the matrix $H$ is not necessarily unique and the uniqueness properties differ from regular PCA. We illustrate this in Example~\ref{ex:nonunique} in the Supplement.}
We are interested in encodings and reconstructions of the random vector $X$ 
that describe the original distribution of $X$ as well as possible. 
In order to measure the quality of the reconstruction, {note that $(X, B \diamond W \diamond X)$ as a concatenated random vector is max-stable as well by Lemma~\ref{lem:maxlin_trafo}. 
So we can use
} 
the semimetric~\eqref{eq:mmetric} and are ready to define our approach to PCA for max-stable distributions.  
\begin{definition}\label{def:maxpca} 
  Let $X$ be a $d$-variate max-stable random vector with $1$-Fréchet margins. 
  Then we call a pair of matrices 
  $(B^\ast, W^\ast) \in [0, \infty)^{d \times p} \times [0, \infty)^{p \times d}$ a \textbf{max-stable PCA of $X$ with $p \leq d$ components}, 
  if it satisfies
  \begin{equation} 
    \label{eq:maxPCA}
    \tilde{\rho}(B^\ast \diamond W^\ast \diamond X, X)=\min_{(B,W) \in [0, \infty)^{d \times p} \times [0, \infty)^{p \times d}}\tilde{\rho}(B \diamond W \diamond X, X).  
  \end{equation} 
  Furthermore, we call the function 
  \begin{equation}
  \label{eq:rec_error}
  [0, \infty)^{d \times p} \times [0, \infty)^{p \times d} \to \R, \quad (B,W) \mapsto \tilde \rho(B \diamond W \diamond X, X) 
  \end{equation}
  the \textbf{reconstruction error} of the matrix pair $(B,W)$ for $X$. 
\end{definition} 
\begin{remark} 
  We want to highlight some key properties of max-stable PCA
  \begin{itemize} 
    \item[1.] The construction is a direct adaptation of the regression approach to classic PCA presented in~\eqref{eq:pca_regr}, 
    where a distribution is first mapped to a lower dimensional space and then back to the original space. 
    \item[2.] Just like regular PCA is particularly suited for data with a Gaussian distribution, 
    since the distribution is preserved under matrix products, 
    our approach preserves the max-stable distribution of the data.  
    \item[3.] The encoding $W \diamond X$ is an optimal max-linear combination such that we can reconstruct $X$ as well as possible. 
    \item[4.] We show in Lemma~\ref{lem:compactset}, that there always exists a global minimizer of \eqref{eq:maxPCA}. 
  \end{itemize} 
\end{remark} 

In order to solve the minimization problem \eqref{eq:maxPCA} for a known or estimated distribution of $X$ we shall first investigate how the reconstruction error depends on the spectral measure of $X$. 
\begin{lemma}\label{metric_representation} 
  Let $X=(X_1, \ldots, X_d)^T$ be a $d$-variate max-stable random vector with non-degenerate $1$-Fréchet margins given by a spectral measure $S$ as in~\eqref{Eq:spectral_meas_rep} or spectral functions $f_1, \ldots, f_d \in L^1_+([0,1])$ as in~\eqref{Eq:spectral_fun_rep}. 
  Then for $b_i, c_i \geq 0$, $i = 1, \ldots,d$, 
  it holds that 
  \begin{align}\label{eq:metric_representation} 
    \rho \left(\bigvee_{i= 1}^d b_i X_i, \bigvee_{i = 1}^d c_i X_i \right) 
    & = \int_{\S^{d-1}_+} \left \lvert \bigvee_{i=1}^d b_i a_i - \bigvee_{i=1}^d c_i a_i \right\rvert \, S(da) \\
    \label{metric_representation_sf} & = \int_{[0,1]} \left \lvert \bigvee_{i=1}^d b_i f_i(e) - \bigvee_{i=1}^d c_i f_i(e) \right \rvert \, de. 
  \end{align} 
\end{lemma} 
Proof of this and nontrivial proofs for all our results below can be found in the supplementary material. Observe that using Lemma~\ref{metric_representation}, we can quantify how well a max-stable random vector $X$ with $1$-Fréchet margins is reconstructed by a matrix pair $(B,W)$ if we know the spectral measure. In particular we can rewrite the reconstruction error~\eqref{eq:rec_error} from Definition~\ref{def:maxpca} as 
\begin{equation}\label{eq:rewrite_tf} 
  \tilde{\rho}(B \diamond W \diamond X, X)=\sum_{k=1}^d \rho(B_k \diamond W \diamond X, X_k ) = \int_{\S^{d-1}_+} \sum_{k=1}^d \left \lvert B_k \diamond W \diamond a - a_k \right \rvert \, S(da).  
\end{equation} 
Because there are consistent estimators of the spectral measure $S$, we will use~\eqref{eq:rewrite_tf}  for our statistical analysis of max-stable PCA in Section~\ref{sec:stats}. 

\section{Properties of max-stable PCA}\label{sec:properties}

The properties of a minimizer to~\eqref{eq:maxPCA} differ to properties of classic PCA, 
due to the fact that we use different operations for the reconstruction matrices 
and a semimetric which minimizes an absolute value. 

Intuition from classic PCA could lead to the conjecture that we can make the choice 
$B = W^T$ to half the number of parameters needed. Contrary to intuition, this is
\emph{not} the case, 
{as seen by Example~\ref{ex:nosymmetry} in the Supplement.

Classic PCA is able to retain all information about a random vector with existing covariance matrix in $p$ components if there are only $p$ non-zero eigenvalues.}
We give a description for the class of max-stable distributions that can be reconstructed with zero reconstruction error in the following result. 
\begin{theorem}\label{thm:spectralrep}
  Let $X$ be a $d$-variate max-stable random vector with non-degenerate $1$-Fréchet margins and let $p \leq d$.
  Then the following statements are equivalent. 
  \begin{enumerate}[(i)]
    \item There exists a pair of matrices $(B,W) \in [0, \infty)^{d \times p} \times [0, \infty)^{p \times d}$ such that the reconstruction error of the max-stable PCA with $p$ components satisfies
      \begin{equation} \label{Eq:perfect_B_W}
        \tilde \rho(B \diamond W \diamond X, X ) = 0. 
      \end{equation}
    \item After possibly permuting the entries of $X$, 
      there exists a $p$-variate max-stable random vector $Y$ with non-degenerate $1$-Fréchet margins 
      and a matrix $\Lambda \in [0, \infty)^{(d-p) \times p}$, such that 
      \begin{equation}
        \label{eq:maxlinrep}
        X \disteq \begin{pmatrix}
          \mathrm{id}_{p} \\
          \Lambda
        \end{pmatrix} \diamond Y. 
      \end{equation}
  \end{enumerate}
       Furthermore, if~\eqref{eq:maxlinrep} holds, we can choose $(B,W)$ as 
      \begin{equation}\label{eq:perfrecBW}
	B = \begin{pmatrix}
          \mathrm{id}_{p} \\
          \Lambda
        \end{pmatrix} \quad 
		W = \begin{pmatrix}
          \mathrm{id}_{p}, & \mathbf 0_{p \times (d - p)}
        \end{pmatrix}. 
      \end{equation}
\end{theorem}
We call a random vector $X$ which satisfies the statements above \emph{perfectly reconstructable} (for a given value of $p$). In essence, this theorem says that a distribution is perfectly reconstructable by max-linear combinations if and only if it has the same distribution as the concatenation of a $p$-variate max-stable random vector $Y$
and $d-p$ max-linear combinations of $Y$. 

Note that the matrices $B$ and $W$ in \eqref{eq:perfrecBW} contain many zeros, since $B$ has at least $p(p-1)$ and $W$ has at least $p(d-1)$ zero entries, hence, at least $p^2 - 2p + pd$ out of the $2dp$ entries in the matrix pair $(B,W)$ are zero. Furthermore, $B \diamond W$ also has at least $(d-p)p + (p-1)p$ zero entries, yielding that all appearing matrices can be seen as sparse. 
{Before we can give an alternative characterization of perfectly reconstructable distributions, we introduce some notation about linear independence with respect to the maximum operation. 
\begin{definition}
\label{def:maxlinindepvecs}
If for a family of vectors $x_1, \ldots, x_n \in [0, \infty)^d$ there exist $i \in  \{1, \ldots, n\}$ and constants $\alpha_j \geq 0$, $j = 1, \ldots, n$ such that the following relation holds:  
\begin{equation*}
  x_i = \bigvee_{j \neq i} \alpha_j x_j, 
\end{equation*}
then we call $x_1, \ldots, x_n$ \textbf{$\vee$-linearly dependent}. If $x_1, \ldots, x_n$ are not $\vee$-linearly dependent, we call $x_1, \ldots, x_n$ \textbf{$\vee$-linearly independent}.
\end{definition}
The next result fully characterizes perfectly reconstructable distributions using $\vee$-linear dependence 
} and entails that perfectly reconstructable distributions by max-stable PCA with $p$ components always have a representation with a matrix that has only $p$ $\vee$-linearly independent rows, yielding a nice analogy to classic PCA. 
\begin{corollary}
  \label{cor:mlf_rep}
  Let $X$ be a $d$-variate max-stable random vector with non-degenerate $1$-Fréchet margins. 
Then $X$ is perfectly reconstructable by max-stable PCA with $p\leq d$ components (i.e.\ one of the equivalent statements in Theorem~\ref{thm:spectralrep} holds), if and only if 
there exist $l \in \N$ and an $l$-variate max-stable random vector $Z$ with non-degenerate $1$-Fréchet margins and a matrix $A \in [0, \infty)^{d \times l}$ with at most $p$ $\vee$-linearly independent rows such that $X$ satisfies 
\begin{equation}
\label{eq:genmaxlinfact}
  X \disteq A \diamond Z.
\end{equation}
\end{corollary}
{We give an example here to demonstrate how the two results obtained are related. 
\begin{example}
\label{ex:mlmexample}
Consider the following two $5 \times 3$ matrices
\begin{equation}
  \label{eq:mlm_mats}
    A^{(1)}  := \begin{pmatrix}
      \tfrac 3 5 & \tfrac 3 {10} & \tfrac 1 {10} \\
      \tfrac 1 {20} & \tfrac 3 4 & \tfrac 1 5 \\
      \tfrac 1 {10} & \tfrac 1 {10} & \tfrac 4 5 \\
      \tfrac 6 {17} & \tfrac 3 {17} & \tfrac 8 {17} \\
      \tfrac 1 3 & \tfrac 2 9 & \tfrac 4 9 
    \end{pmatrix}, \quad 
    A^{(2)} : = \begin{pmatrix}
      \tfrac 2 3 & 0 & \tfrac 1 3 \\
      \tfrac 2 3 & \tfrac 1 3 & 0 \\
      0 & \tfrac 2 3 & \tfrac 1 3 \\
      0 & \tfrac 1 3 & \tfrac 2 3 \\
      \tfrac 1 5 & \tfrac 3 5 & \tfrac 1 5 
    \end{pmatrix}.
  \end{equation}
  Straightforward calculations show that in the first matrix the first three rows are $\vee$-linearly independent and the last two rows are $\vee$-linear combinations of the first three rows. This is seen by 
  \begin{equation}
  \label{eq:example_props}
      A_4^{(1)} = \frac {10}{17} A_1^{(1)} \vee \frac {10}{17}A_3^{(1)}, \quad 
      A_5^{(1)} = \frac 5 9 A_1^{(1)} \vee \frac 8 {27} A_2^{(1)} \vee \frac 5 9 A_3^{(1)}. 
  \end{equation}
  Meanwhile, for the second matrix all $5$ rows are $\vee$-linearly independent. Consider now max-stable random vectors $X^{(1)}, X^{(2)}$ with 
   \begin{equation*}
    X^{(i)} := A^{(i)} \diamond Z, \quad i = 1,2, 
  \end{equation*} where $Z = (Z_1, Z_2, Z_3)^T$ has independent, non-degenerate $1$-Fréchet margins with $\sigma_1=\sigma_2=\sigma_3=1$. By Corollary~\ref{cor:mlf_rep}, it holds that $X^{(1)}$ is perfectly reconstructable with $p = 3$, while $X^{(2)}$ is not perfectly reconstructable with $p \leq 4$. From~\eqref{eq:example_props} we can also calculate a representation of $X^{(1)}$ as in Theorem~\ref{thm:spectralrep} as follows. We define the 3-dimensional random vector $Y^{(1)}$ by  $Y_i^{(1)} := A_i^{(1)} \diamond Z, i = 1,2,3$, and it immediately follows that 
  \begin{equation*}
      X^{(1)} \disteq \begin{pmatrix}
          & \mathbf{id}_3 & \\
          \frac {10}{17} & 0 & \frac {10}{17} \\
          \frac 5 9 & \frac 8 {27} & \frac 5 9 
      \end{pmatrix} \diamond Y^{(1)}.
  \end{equation*}
  This highlights how random vectors that satisfy Theorem~\ref{thm:spectralrep} and Corollary~\ref{cor:mlf_rep} are such that $d-p$ components are completely dependent, in the sense of being $\vee$-linear combinations of the other $p$ components of the random vector. Therefore, we expect the procedure to work well if at least some components exhibit high dependence. 
  \end{example}
  Note that the example above falls into the well studied parametric model class of max-linear models~\cite{Amendola_cimln, WangStoev_sdmsrf}. They are popular for evaluating the performance of unsupervised learning procedures in multivariate extremes and often appear as model case for these procedures~\cite{Engelke_ssme, GissiblKlueppelberg_rmlm, JanssenWan_kmce} and we will revisit them in our experimental studies, see Section~\ref{Subsubsec:maxlinear}. 
  }

\begin{remark}
    \begin{enumerate}[(i)]
    \item Corollary~\ref{cor:mlf_rep} resembles the motivation of classic PCA from the introduction, where PCA perfectly reconstructs a Gaussian random vector $X$ with $p$ components, if and only if for $Z = (Z_1, \ldots, Z_p)^T$, $Z_i \stackrel{i.i.d.}{\sim} \mathcal N(0,1)$ and a matrix $A \in \R^{d \times p}$, $X$ is given by 
\begin{equation*}
    X \disteq A \cdot Z.
\end{equation*}
\item Note that in the setting of tropical linear algebra, there can be more than $p$ $\vee$-linearly independent rows in a matrix of dimension $d \times p, p \leq d$, see, e.g., \cite{Butkovic_mls} Theorem 3.4.4. Furthermore, a $p$-variate max-stable random vector can in general not be written as a $\vee$-max linear combination of $p$ \emph{independent} Fr\'{e}chet random variables: this is only possible if the spectral measure is discrete with $p$ points of mass, see~\cite{Yuen_crps}. Accordingly, in the max-stable PCA setting, the latent random vector $Y$ does not necessarily have independent entries. 
\end{enumerate}
\end{remark}

\section{Statistics of max-stable PCA}\label{sec:stats}

In this section we outline how to apply max-stable PCA to data. To this end, we should first revisit our assumption of max-stability of all random vectors which appeared in Sections \ref{sec:defmaxpca} and \ref{sec:properties}. It was motivated by the fact that max-stable distributions are the only possible limit distributions that can arise for maxima of i.i.d.\ random vectors with margins standardized to a 1-Fr\'{e}chet distribution, cf.\ Section~\ref{subsec:evt}. This result can in fact be generalized to observations from stationary multivariate time series that satisfy a mild mixing condition, see \cite[Theorem 10.22]{Beirlant_soe}. Thus, whenever we deal with data that has been obtained by taking maxima from a sufficiently large number of observations we can use a max-stable distribution as an approximation for the resulting multivariate distribution.  

Recall that we have used Lemma~\ref{metric_representation} to obtain a convenient representation of the reconstruction error from Definition~\ref{def:maxpca} by equation~\eqref{eq:rewrite_tf}.
This allows to apply our procedure to data, since there are multiple estimators for the spectral measure, ranging from parametric to nonparametric, and for our purposes, we can use a plugin approach.  For different choices to nonparametric estimators see for example~\cite{DeHaan_npesm, Einmahl_mlesm} and~\cite{Engelke_ehrd, Tawn_bvevt, Wadsworth_dmse} is an incomplete list for semiparametric and parametric alternatives to estimate the spectral measure. We briefly revise a well known non-parametric estimator of the spectral measure $S$ with good theoretical properties~\cite{DeHaan_npesm}. Let $X_{i}=(X_{1i}, \ldots, X_{di})^T \in \R^d$ be i.i.d. realizations from a $d$-variate distribution function $F$ with continuous marginals in the max-domain of attraction of a $d$-variate extreme value distribution $G$.  To ensure that we can recover the spectral measure of the transformed limit distribution $G^*$ with $1$-Fréchet margins and equal scale coefficients $\sigma_1=\ldots=\sigma_d=1$, we transform $X_i$ to a new set of data $\hat X_i^*$ by 
\begin{equation}\label{eq:trafodata} 
  \hat X_{ki}^* =\frac 1 {1 - \hat F_k(X_{ki})}, \;\;\; k=1, \ldots, d,
\end{equation} 
where $\hat F_k$ is an estimator of the marginal distribution $F_k$ obtained from the ranks of the observations given by 
\begin{equation*} 
  R_{ki} := \sum_{l=1}^n \mathbbm 1_{\{X_{kl} \leq X_{ki}\}}, \;\;\; k=1, \ldots, d, i=1, \ldots, n.
\end{equation*} 
Using the ranks $R_{ki}$, we set 
\begin{equation}\label{eq:cdfest} 
  \hat F_k(X_{ki}) = \frac 1 n (R_{ki} - 1).  
\end{equation} 
We do not use the classic empirical distribution function to prevent division by zero in~\eqref{eq:trafodata}.  Note that other choices instead of~\eqref{eq:cdfest} are possible, see~\cite{Beirlant_soe} (9.37) for example. The transformed data $\hat X_i^*$ can be used to construct a non-parametric estimator for the spectral measure $S$ on the sphere given by a norm $\| \cdot \|$ on $\R^d$ as follows
\begin{equation}\label{eq:est_spectralmeasure} 
  \hat S_n(A) := \frac 1 k \sum_{i=1}^n \mathbbm 1_{ \left\{ \| \hat X_i^* \| > \frac n k, \frac {\hat X_i^*} {\| \hat X_i^* \|} \in A \right\}}, \quad A \in \mathcal B(\S^{d-1}_+).  
\end{equation} 
This estimator is based on the "Peaks over threshold" approach, which uses the observations with the largest norms in a data set to estimate extremal characteristics (see, e.g., \cite[Section 9.4]{Beirlant_soe}). According to this principle, as long as \( k = k(n) \) and \( n \) satisfy suitable growth conditions, it is sufficient that the data is a sample from a distribution in the max-domain of attraction of some multivariate extreme value distribution $G$, in order for \eqref{eq:est_spectralmeasure} to be a consistent estimator of the spectral measure. The estimator \eqref{eq:est_spectralmeasure} for the spectral measure has been extensively studied in the literature~\cite{DeHaan_npesm, Beirlant_soe} and is related to many other estimators of the dependence structure (see, e.g., \cite{Drees_barstdf, Einmahl_mlesm, deHaanResnick_eldmve}).
 Thus, using data $X_1, \ldots, X_n$, instead of the true distribution, we calculate the empirical reconstruction error using the estimated spectral measure, given by 
\begin{equation}
\label{eq:emprecerror}
    \hat{\tilde \rho}_n (B \diamond W \diamond X, X) := \int_{\S^{d-1}_+} \sum_{k=1}^d \lvert B_k \diamond W \diamond a - a_k \rvert \hat S_n(da). 
\end{equation}
Because $\hat S_n$ is constructed from the random vectors $X_1,\ldots, X_n$, it is a random measure. The following general theorem applies to max-stable PCA as a special case. 
\begin{theorem}\label{thm:consistency}
  Let $S$ be a finite measure on $(\S^{d-1}_+, \mathcal B(\S^{d-1}_+))$ and 
  $(\hat S_n)_{n \in \N}$ be a sequence of random finite measures on $(\S^{d-1}_+, \mathcal B(\S^{d-1}_+))$ defined on a common probability space $(\Omega, \mathcal A, \P)$, such that for all continuous functions $f\colon \S^{d-1}_+ \to \R$ 
  \begin{equation}
    \label{eq:weakconv_prob}
  \int_{\S^{d-1}_+} f(a) \hat S_n(da) \xrightarrow{P}\int_{\S^{d-1}_+} f(a) S(da).  
  \end{equation}
  Furthermore, let $K \subseteq \R^m, m \in \N$ be a compact set and 
  $c \colon K \times \S^{d-1}_+ \to \R$ be a continuous function. 
  Then it holds that 
  \begin{enumerate}[(i)]
    \item There exists a random sequence $(\hat \theta_n)_{n \in \N} \in K$ 
  on $(\Omega, \mathcal A, \P)$ satisfying
  \begin{equation}
    \label{eq:riskminimizers}
    \int_{\S^{d-1}_+} c(\hat \theta_n, a) \, \hat S_n(da) = \inf_{\theta \in K} \int_{\S^{d-1}_+} c(\theta, a) \, \hat S_n(da) \quad \forall n \in \N. 
  \end{equation}
  \item For any random sequence $(\hat \theta_n)_{n \in \N}$ satisfying~\eqref{eq:riskminimizers}, it holds that 
  \begin{equation}
    \label{eq:ermresult}
    \int_{\S^{d-1}_+} c(\hat \theta_n, a) \, S(da) \xrightarrow{P}  \inf_{\theta \in K} \int_{\S^{d-1}_+} c(\theta, a) \, S(da)
     \quad n \to \infty. 
  \end{equation}
\end{enumerate}
\end{theorem}
This theorem establishes that for suitable parameter spaces and cost functions, 
the empirical risk minimization approach does indeed asymptotically minimize the true risk. 

For i.i.d. data from a distribution that lies in the max-domain of attraction of $G$, the estimator~\eqref{eq:est_spectralmeasure} converges in the desired sense to the true spectral measure $S$ of the limit distribution if $k = k(n) \to \infty$ and $\frac k n \to 0$, see~\cite{DeHaan_npesm}. {Theorem~\ref{thm:consistency} together with Lemma~\ref{lem:compactset} from the supplement thus shows that asymptotically the true risk is minimized when replacing the unknown $S$ by the estimator $\hat{S}_n$ from~\eqref{eq:est_spectralmeasure} in max-stable PCA.

\section{Experimental studies}\label{sec:experiments}

Since we identified a class of theoretical models that can be perfectly reconstructed by max-stable PCA in Theorem~\ref{thm:spectralrep} and Corollary~\ref{cor:mlf_rep}, and established a consistency result in Theorem~\ref{thm:consistency}, it is theoretically justifiable to apply max-stable PCA to data to detect if it approximately follows a distribution as in~\eqref{eq:maxlinrep}. To demonstrate that the procedure also works in practice, we use our implementation (see \ref{subsec:implementation} for details) and investigate the finite sample performance of max-stable PCA for simulated data and a reference real data set. We propose to plot the empirical reconstruction errors for different values of $p$ in an ''elbow plot'' and use a heuristic to choose $p$ at an ''elbow'', for the simulated data. Additionally, we recommend plotting the bivariate margins of the data versus the reconstructed data to visually inspect if the reconstruction captures key properties, if the dimension of the data is small. Simulations are carried out using the programming language \emph{R}~\cite{Rlang} and our R-package maxstablePCA~\cite{Reinbott_maxpcaR}. For reproducibility, the scripts for the simulations are provided in the complementary repository~\url{www.github.com/FelixRb96/maxstablePCA_examples}. 

\subsection{Max-linear models}\label{Subsubsec:maxlinear}

  We evaluate if we can recover the max-linear models from Example~\ref{ex:mlmexample} given by the matrices~\eqref{eq:mlm_mats} by the empirical procedure. We set up the simulation study by simulating $n = 10000$ i.i.d. observations of $X^{(i)}, i=1,2$ and use the estimator~\eqref{eq:est_spectralmeasure} with {$n/k=100$ and the $\| \cdot \|_{\infty}$ norm}. We let $p = 1, \ldots, 4$ be the number of components in max-stable PCA and plot $p$ versus the empirical reconstruction error in Figure~\ref{fig:elbowplot_mlm}. The plot indicates that at $p = 3$, the elbow plot flattens close to zero for the perfectly reconstructable model $X^{(1)}$, agreeing with our theoretical considerations up to small approximation errors by the empirical counterparts and optimization errors. For the second model, the elbow plot never reaches zero, indicating that there is a loss of information for all $p = 1, \ldots, 4$. However, considering the reconstruction for $p = 3$ for the sample of $X^{(2)}$ in Figure~\ref{fig:mlmplots}, the max-stable PCA still captures many of the typical rays the max-linear model concentrates around for large values. The estimated matrices $(\hat B, \hat W)$ for $p = 3$ are reported in the Supplement in ~\eqref{eq:app_mlmest1} for the perfectly reconstructable case and in~\eqref{eq:app_mlmest2} for the second case. In both cases, max-stable PCA tends to prioritize one entry of the random vector $X$ to construct the encoded states $W \diamond X$. Also, the matrix used for the reconstruction $B$ can be seen as close (w.r.t. to the $\infty$-norm for example) to a row permuted version of~\eqref{eq:maxlinrep} up to rescaling of the rows. 
  \begin{figure}[htb]
    \centering
    \begin{subfigure}[b]{0.3\textwidth}
      \centering
      \includegraphics[width=\textwidth]{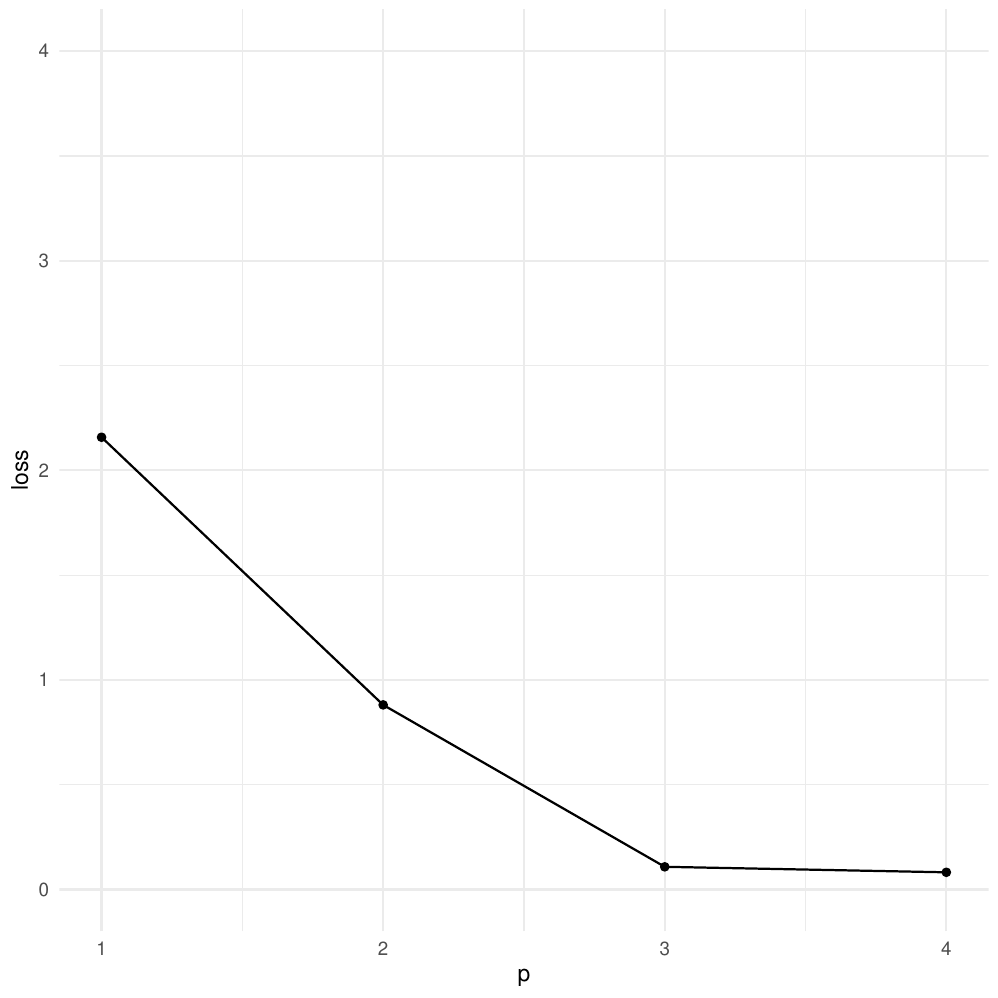}
      \caption{$A^{(1)}$ }
    \end{subfigure}
    \hfill
    \begin{subfigure}[b]{0.3\textwidth}
      \centering
      \includegraphics[width=\textwidth]{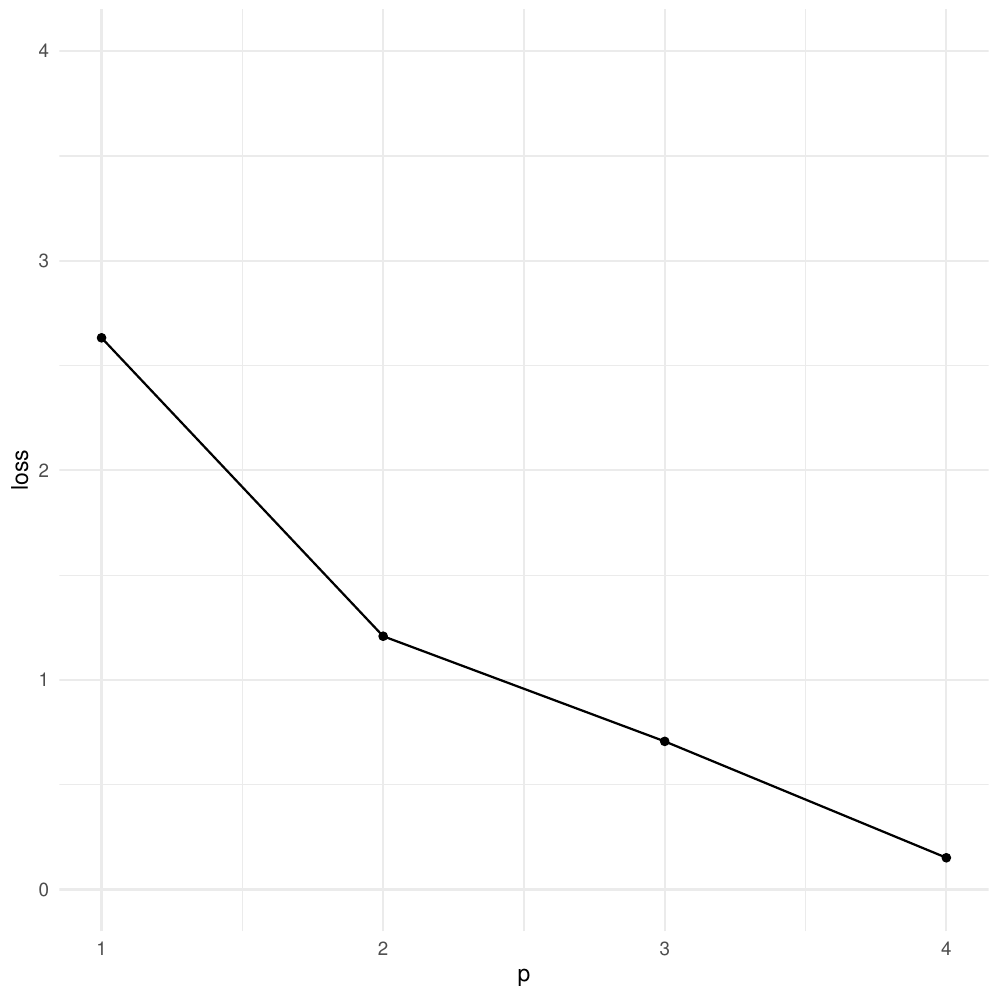}
      \caption{$A^{(2)}$ }
    \end{subfigure}
    \caption{The reconstruction error values of the max-stable PCA for $p = 1, \ldots, 4$ for two max-linear models from example~\ref{ex:mlmexample} by the matrices in~\eqref{eq:mlm_mats}.}
    \label{fig:elbowplot_mlm}
    \end{figure}

\subsection{Logistic models}\label{Subsubsec:logistic}

  Since max-linear models are characterized by a discrete spectral measure (see \cite{Yuen_crps}), 
  we will in this section evaluate the performance for max-stable PCA for a spectral measure that allows for a density. The logistic model is one of the most frequently used parametric models for multivariate extreme value theory, see, e.g., ~\cite[Section 9.2.2]{Beirlant_soe}.
  The distribution function of the logistic model with $1$-Fréchet margins and all scale coefficients equal to one
  with parameter $\beta \in (0, 1]$ is given by the cumulative distribution function
  \begin{equation}
    \label{eq:logmod}
    G \colon \R^d \to [0,1], \quad z \mapsto G(z) = \exp \left( - \left(\sum_{k=1}^d \left( \frac 1 {z_i} \right)^{\frac 1 {\beta}} \right)^{\beta} \right) \mathbbm 1_{(0, \infty)^d}(z). 
  \end{equation}
  Note that it has been shown in~\cite{Tawn_meme}  that this distribution function has a spectral density, 
  and the parameter $\beta$ controls the strength of the dependence between the components. 
  The case $\beta = 1$ corresponds to complete independence and the limit $\beta \downarrow 0$ corresponds to full dependence. 
  We simulate from the model with $d = 5$ and the following three choices for $\beta$
  \begin{equation*}
    \beta_1 = 0.2, \quad \beta_2 = 0.5, \quad \beta_3 = 0.8. 
  \end{equation*}
  Again, we use the same general setup for our simulation study and 
  sample 10000 i.i.d. observations from the model for the different values of $\beta$, {set $\frac n k = 100$, use the $\| \cdot \|_{\infty}$ norm}, fit max-stable PCA 
  to the realizations for $p = 1, \ldots, 4$ and create elbow plots we report in Figure~\ref{fig:elbowplot_logistic}, as well as pairplots of the bivariate margins presented in Figure~\ref{fig:logisticplots}. 
  This time, we do not observe a clear elbow in any of the three models. 
  One can therefore conclude that there is no low dimensional factor model that describes the data well in this case. 
  Figure~\ref{fig:logisticplots} shows the fitted data plotted against the realization in a pairplot for $p = 3$ and the corresponding estimated matrices are reported in~\eqref{eq:app_log1},~\eqref{eq:app_log2} and~\eqref{eq:app_log3} in the Appendix. Again, in all three cases max-stable PCA tends to highly prioritize one entry of $X$ to construct each encoded state in $W \diamond X$ and the matrix $B$ again is close to~\eqref{eq:maxlinrep} up to rescaling and permutation of rows. 
  Here the overall values of reconstruction errors decrease the more dependence there is in the data, i.e. the smaller the value of $\beta$ is. This is not surprising as the limiting $\beta \downarrow 0$ corresponds to a degenerate model of full dependence, i.e. almost surely equal components, which is perfectly reconstructable even for $p=1$. 
  \begin{figure}[htb]
    \centering
    \begin{subfigure}{0.32\textwidth}
      \centering
      \includegraphics[width=\textwidth]{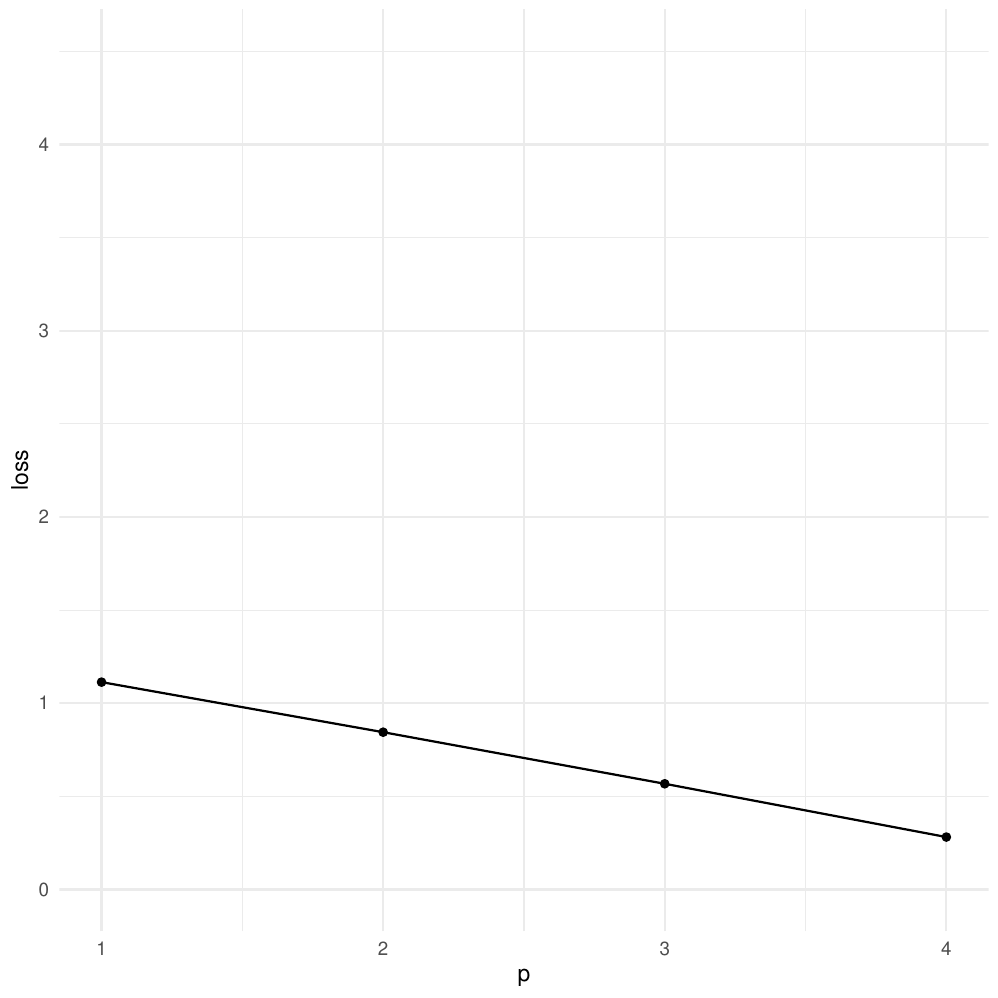}
      \caption{$\beta = 0.2$. }
    \end{subfigure}
    \hfill
    \begin{subfigure}{0.32\textwidth}
      \centering
      \includegraphics[width=\textwidth]{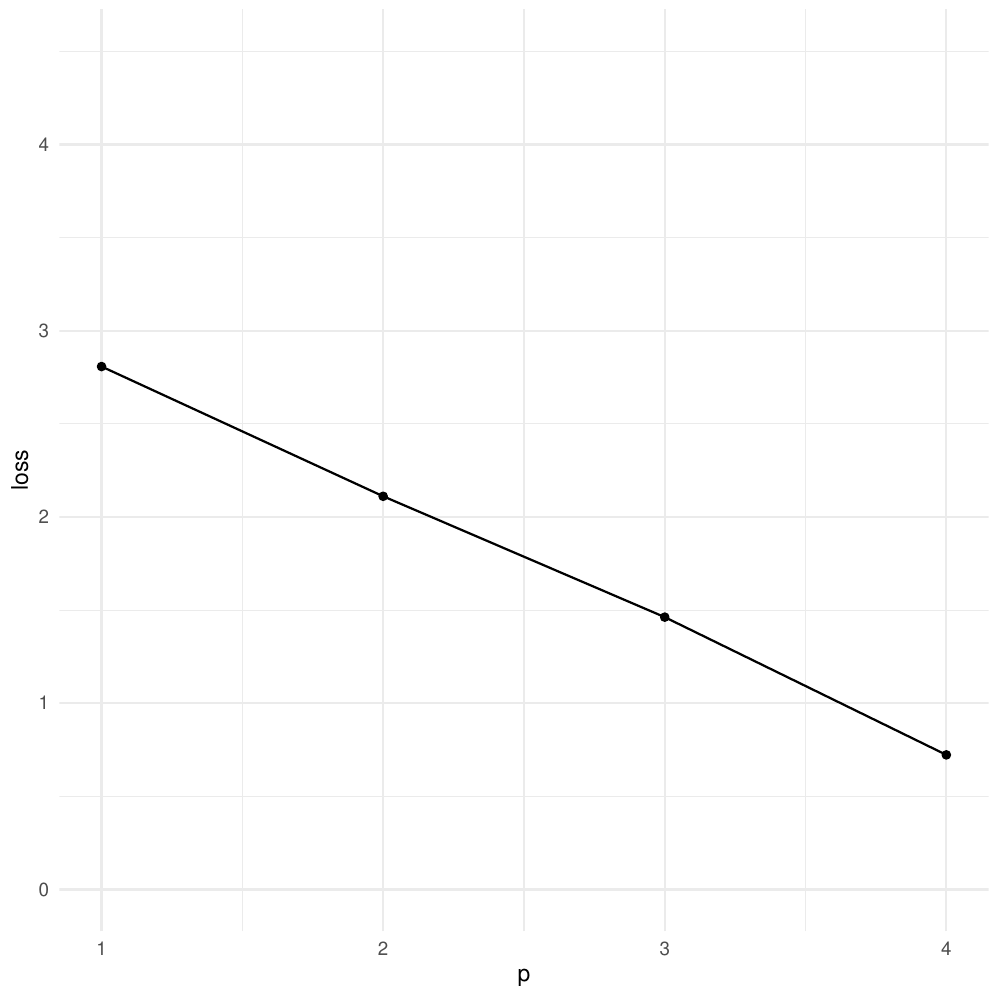}
      \caption{$\beta = 0.5$. }
    \end{subfigure}
    \hfill
    \begin{subfigure}{0.32\textwidth}
      \centering
      \includegraphics[width=\textwidth]{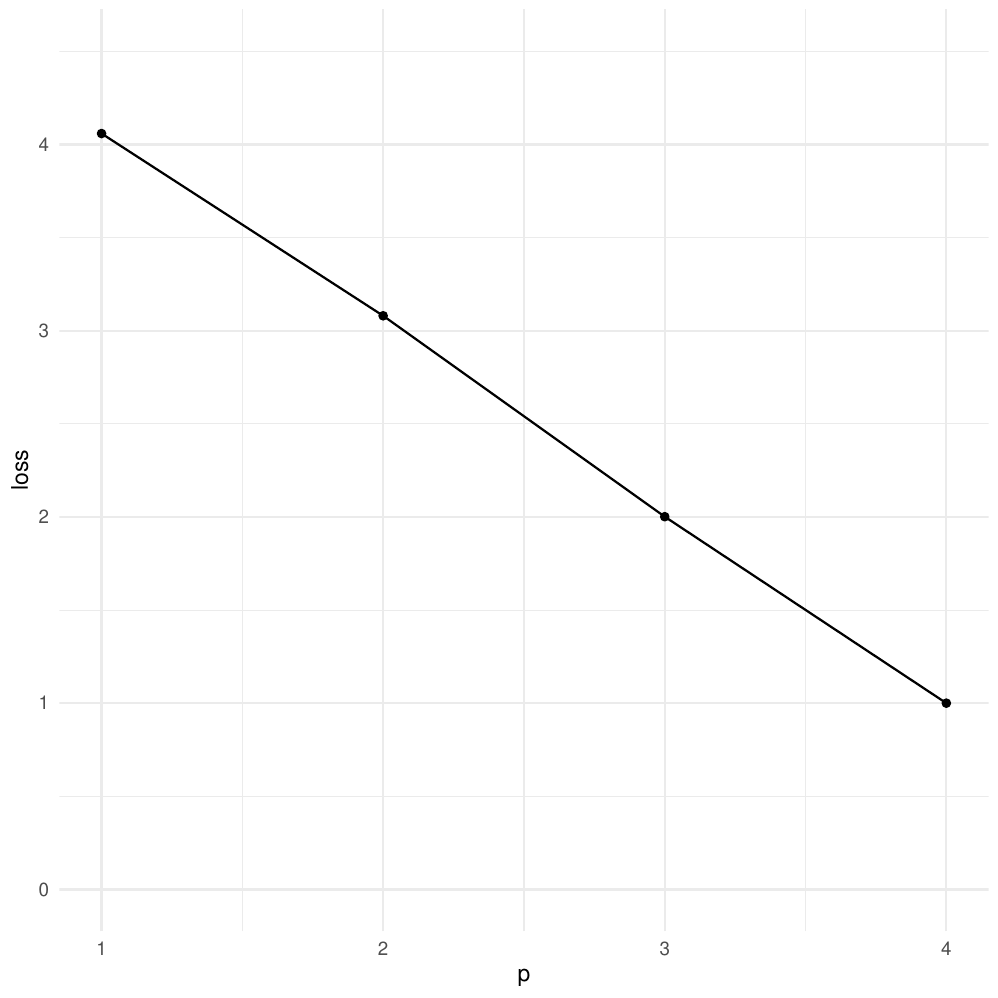}
      \caption{$\beta = 0.8$. }
    \end{subfigure}
    \caption{The reconstruction error values of the max-stable PCA for $p = 1, \ldots, 4$ for the logistic model with three different values of $\beta$.}
    \label{fig:elbowplot_logistic}
    \end{figure}

  \subsection{A perfectly reconstructable model with latent logistic model}\label{Subsubsec:latentlogistic}

  Finally, we simulate from a model where the max-stable random vector $X = A \diamond Z$ is determined by a latent bivariate logistic random vector $Z=(Z_1,Z_2)^T$ with $\beta = 0.5$ and a matrix 
\begin{equation*}
    A = 
    \begin{pmatrix}
      \frac 4 5 & \frac 1 5 \\
      \frac 1 {20} & \frac {19}{20} \\
      \frac 3 5 & \frac 2 5 \\
      \frac 9 {20}  & \frac {11}{20} \\
    \end{pmatrix}^T.
  \end{equation*}
  Again, it can be shown that the conditions of Corollary~\ref{cor:mlf_rep} are met and by defining
\begin{equation*}
Y_1 := \frac 4 5 Z_1 \vee \frac 1 5 Z_2, \quad Y_2 := \frac 1 {20}	Z_1 \vee \frac {19}{20} Z_2 
\end{equation*}
we can write
\begin{equation*}
  X = \begin {pmatrix}
  1 & 0 \\
  0 & 1 \\
  \tfrac 3 4 & \tfrac 8 {19} \\
  \tfrac 9 {16} & \tfrac {11}{19}
  \end{pmatrix} \diamond Y.
\end{equation*}
So $X$ is equal to a model of the form~\eqref{eq:maxlinrep} and for $p = 2$ the theoretical reconstruction error is zero. We again sample $n = 10000$ i.i.d. realizations of $X$, {set $\frac n k = 100$ for the spectral measure estimator, use the $\| \cdot \|_{\infty}$ norm},
fit max-stable PCA for $p = 1, \ldots, 4$, report the reconstruction errors for different values of $p$ in an elbow plot in Figure~\ref{fig:genmlm}
and plot the reconstructed data versus the original data in a bivariate margin plot in Figure~\ref{fig:genmlm}. 
The elbow plot here again highlights that for $p = 2$ the empirical error and optimization error are almost zero. It therefore indicates that in this simulation setup, the procedure recovers the theoretical quantities of interest well. We provide the estimated matrices in~\eqref{eq:app_genmlm} in the supplement. Again, the largest rowwise entries in $W$ correspond to the two entries of $X$ that correspond to the $\vee$-linearly independent rows of the matrix $A$ and $B$ is close to a matrix as in~\eqref{eq:maxlinrep}.

For a more detailed discussion and interpretation of all simulation results we refer to Section~\ref{Sec:estimated_matrices} in the Supplement, and in particular Remark~\ref{Remark:Estimated_Matrices} for a summary over all cases.
  \begin{figure}[hbt]
      \centering
      \includegraphics[width=0.3\textwidth]{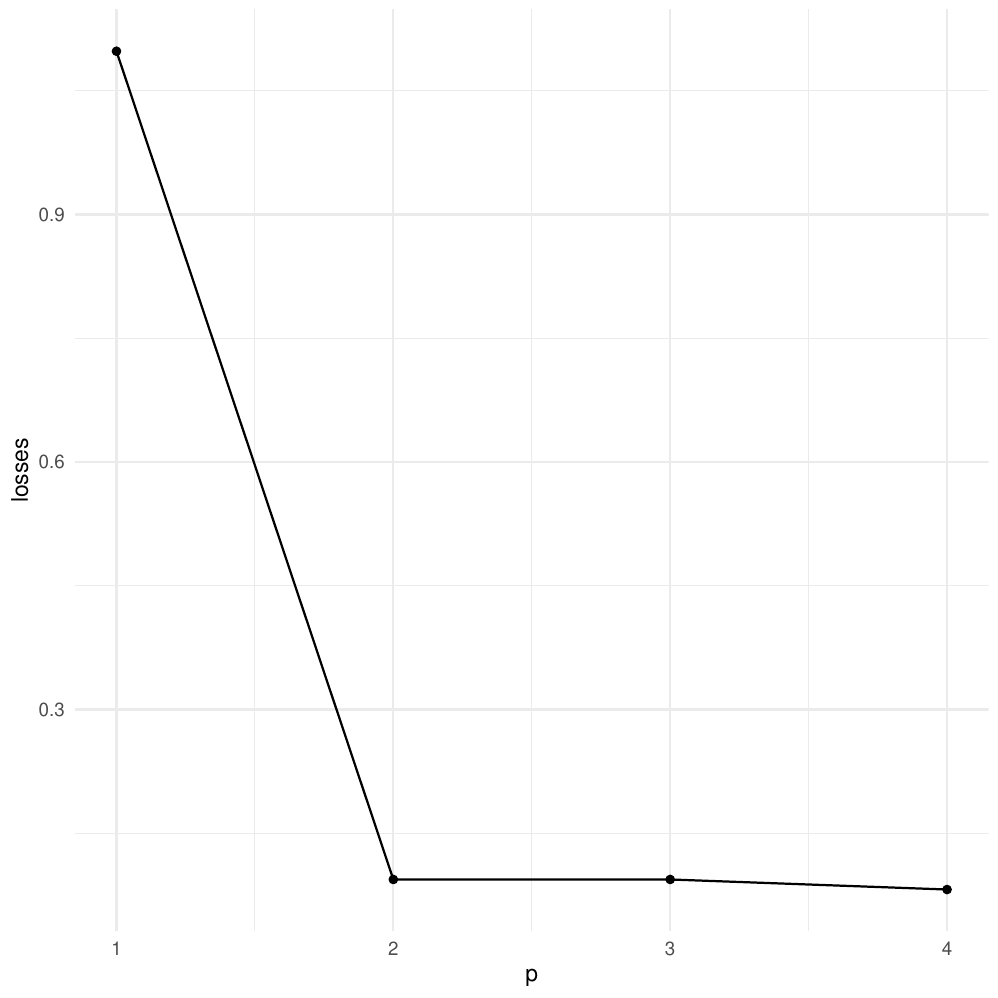}
      \caption{Elbow plot for $p = 1, \ldots, 4$ of the reconstruction error of max-stable PCA for a sample of the model from Section~\ref{Subsubsec:latentlogistic}.}   
  \end{figure}

\subsection{Danube river data}

Finally, we evaluate the performance of max-stable PCA for a real dataset. 
A ready to use dataset about the Danube river can be found in the R package \emph{graphicalExtremes}~\cite{Engelke_ger} and the raw data is available at the website of the Bavarian Environmental agency (\url{http://www.gkd.bayern.de}). The preprocessed version contains measurements about extreme river discharges of 31 stations in the upper Danube river basin in Bavaria from 1960 to 2009. 
The data can be seen as stationary, because only the summer months June to August are taken into account, and independent, due to the application of a declustering technique provided by~\cite{Engelke_eorn}. Therefore, assuming that each row of the data is a realization of an i.i.d. sample is reasonable. This dataset has recently gained a lot of attention in the area of unsupervised learning in extremes, mostly concerned with graphical models~\cite{Chavez_cmerd, Hentschel_sihrgm, Roettger_tpmve, Wan_gle}. We want to investigate if the distribution of the data transformed to $1$-Fréchet margins approximately follows a generalized max-linear factor model, denoted by a random vector $X$, given by $A \in [0, \infty)^{31 \times p}$ with $p$ $\vee$-linearly independent rows and a $p$ variate random vector $Z$ with $1$-Fréchet margins such that 
\begin{equation}
  \label{eq:danubemodel}
  X \disteq A \diamond Z. 
\end{equation}
This model might a priori be a good choice, since the measurement stations are placed on 6 distinct river arms that merge and exploratory data analysis suggests that if we observe an extreme at a measurement station at one arm, it seems highly likely that the adjacent measurement stations will also measure an extreme event.  Furthermore, since some river arms are spatially close, there might be a spatial dependence between river arms, that might not be fully accounted by the river network structure. The model above can account for the latent spatial dependence between $p$ components from the river arms in the latent random vector and by virtue of the high dependence in the river arms and their merging, max-linear combinations might be suitable to model the remaining stations. 

We perform max-stable PCA for $p = 1, \ldots, 12$ and report the reconstruction losses in an elbow plot for the different values of $p$ in Figure~\ref{fig:elbowplot_danube}. We use the estimator~\eqref{eq:est_spectralmeasure} and since the i.i.d. assumption is reasonable as argued above, the estimator is a suitable choice. For the tuning parameters, we follow the proposed practice by~\cite{Engelke_eorn, Hentschel_sihrgm, Roettger_tpmve}, and use the $l_{\infty}$-norm and set the exceedance threshold to $n/k = 10$, 
yielding $117$ observations. In the range of $p = 4$ to $p = 8$, the plot in Fig.~\ref{fig:elbowplot_danube} flattens out, but without a clear elbow. This is somewhat in line with many analyses~\cite{Hentschel_sihrgm, Roettger_tpmve, Wan_gle}, so we further investigate the dataset for max-stable PCA with $p = 6$. 

 \begin{figure}[hbt]
    \centering
    \begin{subfigure}[b]{0.26\textwidth}
      \centering
      \includegraphics[width=\textwidth]{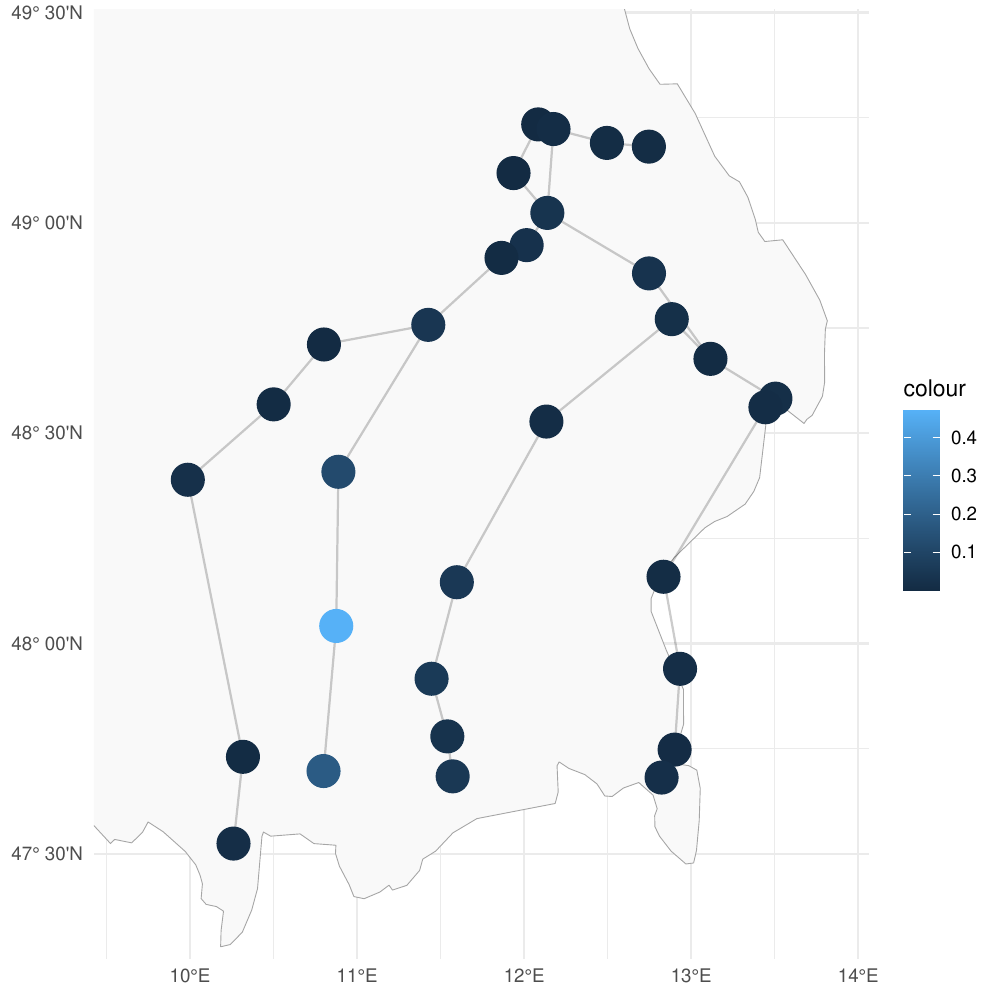}
    \end{subfigure}
    \begin{subfigure}[b]{0.26\textwidth}
      \centering
      \includegraphics[width=\textwidth]{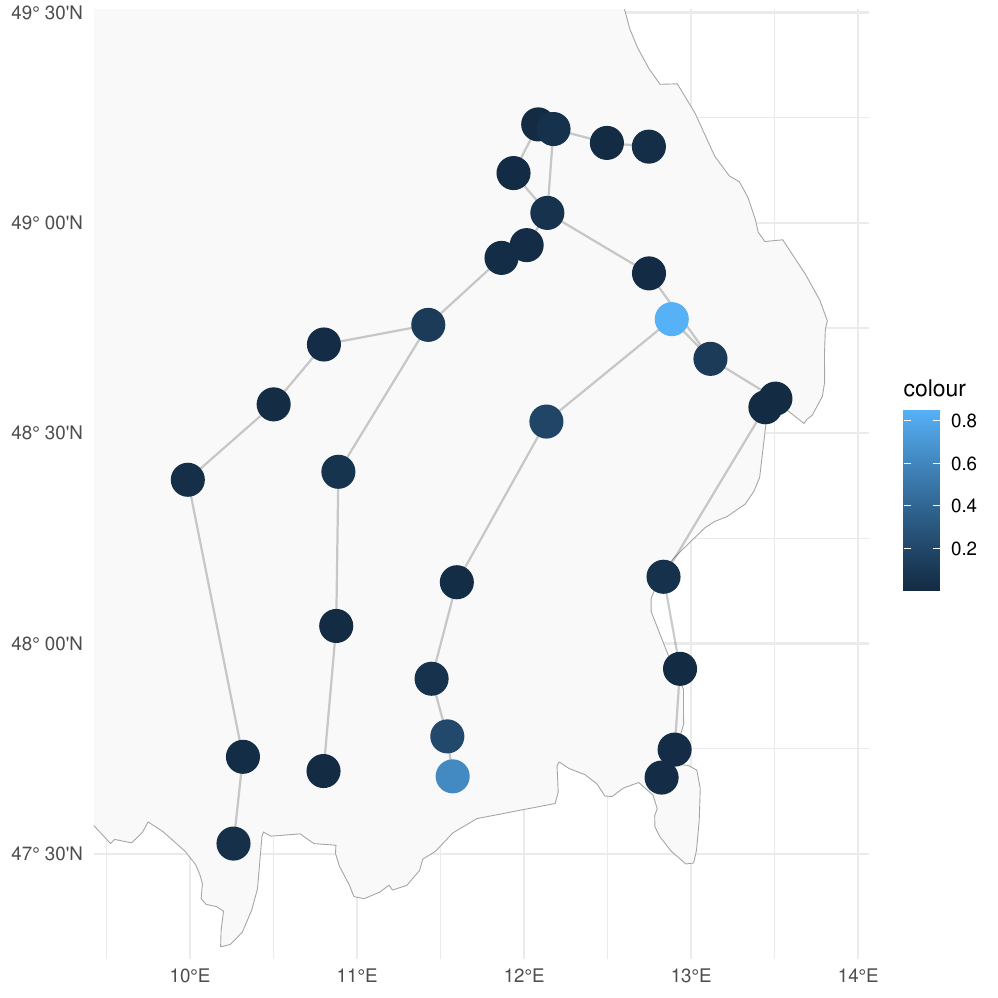}
    \end{subfigure}
    \begin{subfigure}[b]{0.26\textwidth}
      \centering
      \includegraphics[width=\textwidth]{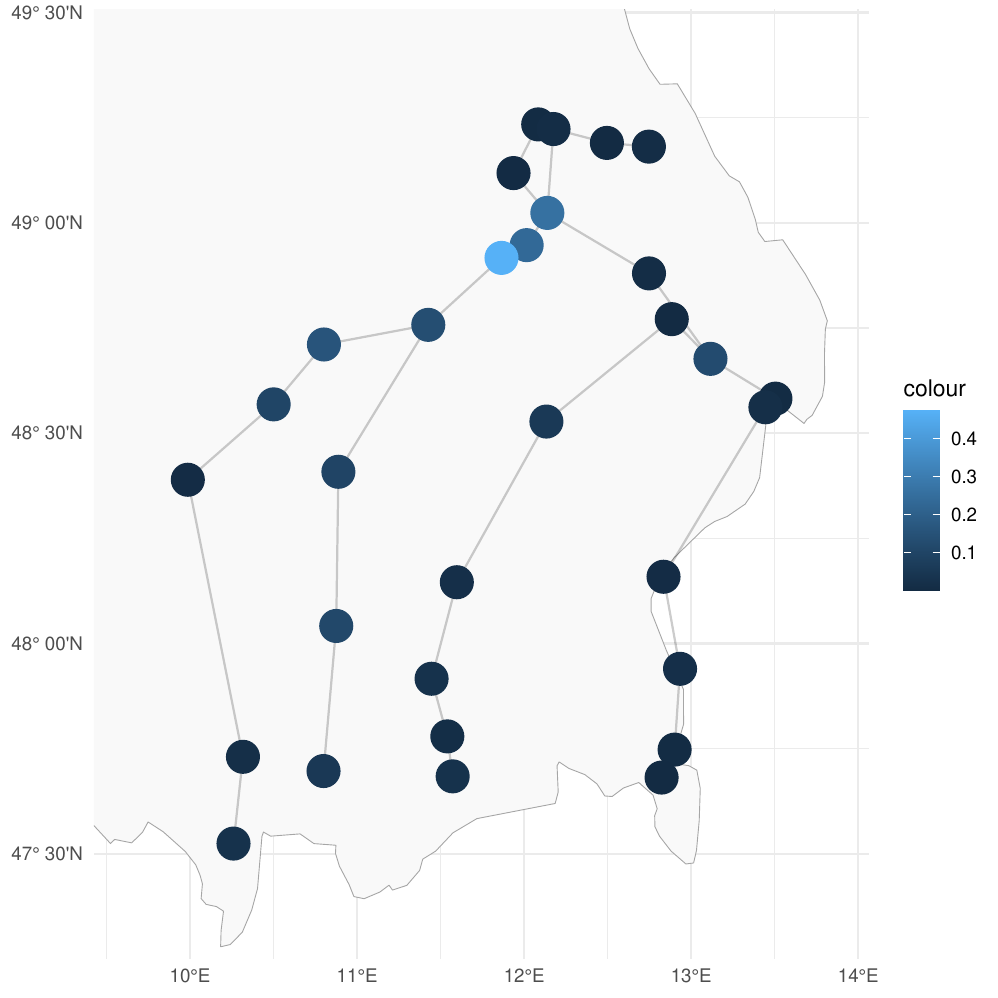}
    \end{subfigure}
    \begin{subfigure}[b]{0.26\textwidth}
      \centering
      \includegraphics[width=\textwidth]{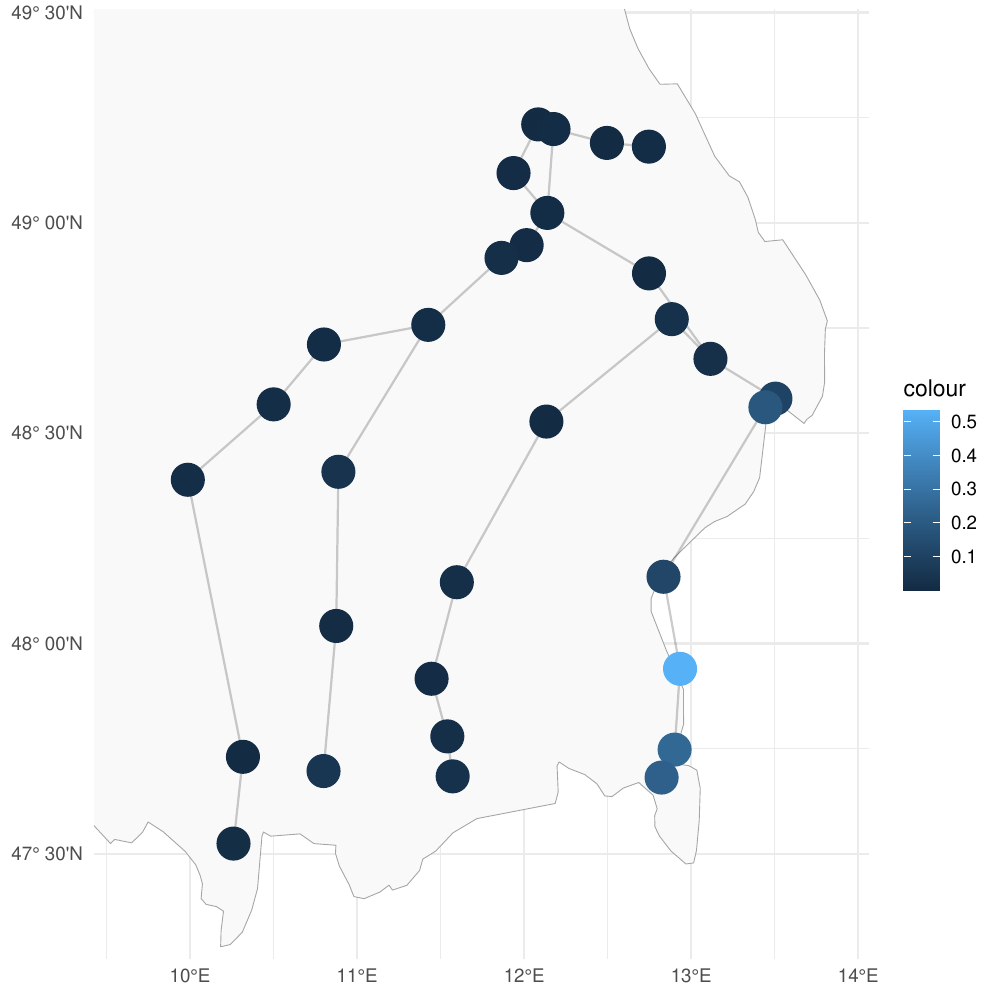}
    \end{subfigure}
    \begin{subfigure}[b]{0.26\textwidth}
      \centering
      \includegraphics[width=\textwidth]{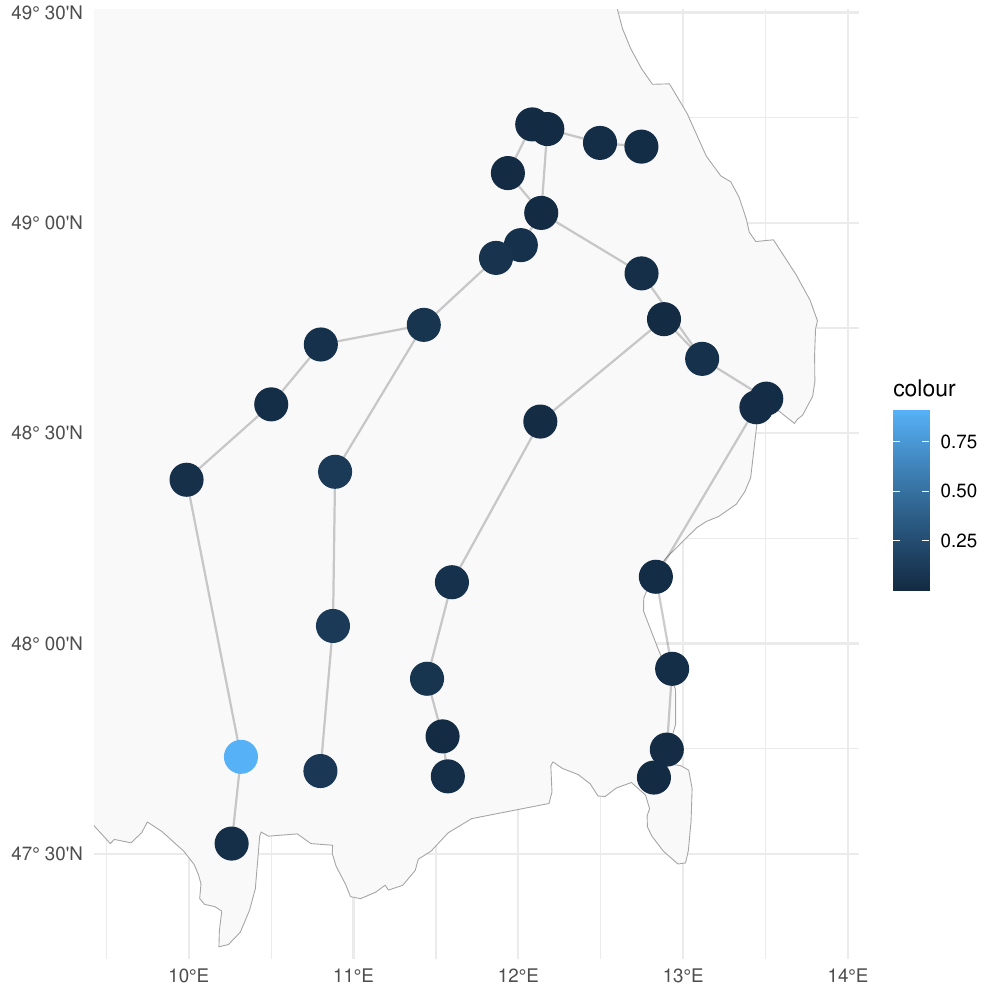}
    \end{subfigure}
    \begin{subfigure}[b]{0.26\textwidth}
      \centering
      \includegraphics[width=\textwidth]{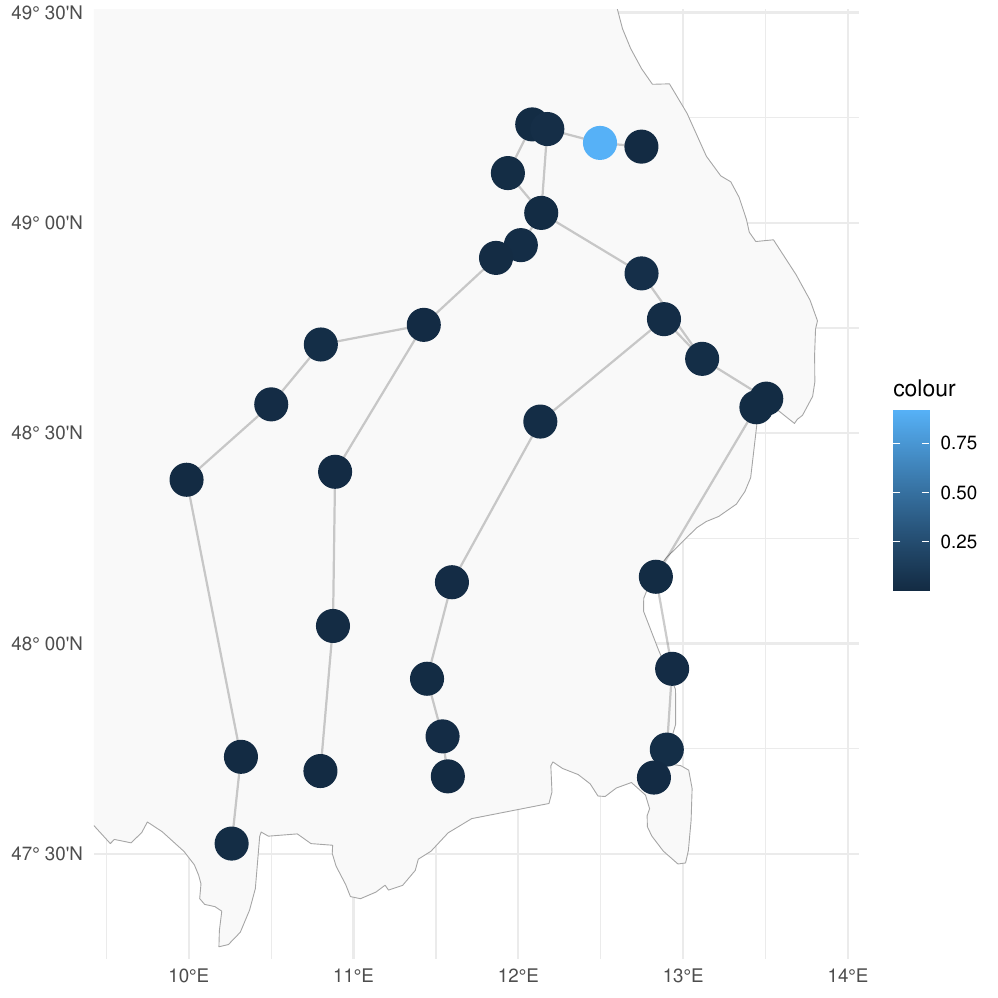}
    \end{subfigure}
    \caption{
        The rows of the estimated $W$ for $p=6$ mapped to their spatial locations ($W_{ji}$ corresponds to station $i$ and component $j$), where larger $W_{ji}$ has the lighter blue color. 
        The top row contains the first three rows and the second row the last three. Observe that each encoded state highly loads stations from the same riverarm. This holds true also for the upper middle plot, compare with the river arm structure provided in Figure~\ref{fig:pairplot_danube}.
        }
    \label{fig:Wspat_danube}
\end{figure}

We visualize the loadings for the principal components, that is the rows of the estimated matrix $W$, by mapping the rows of $W$
for $p = 6$ at the locations of the measurement stations in a map in Figure~\ref{fig:Wspat_danube}. This  illustrates that each row roughly captures one river arm, where the measurement stations that correspond to a large entry in $W$ are chosen by the procedure such that they are spatially far from each other.

\begin{figure}[bht]
    \centering
    \begin{subfigure}[b]{0.26\textwidth}
      \centering
      \includegraphics[width=\textwidth]{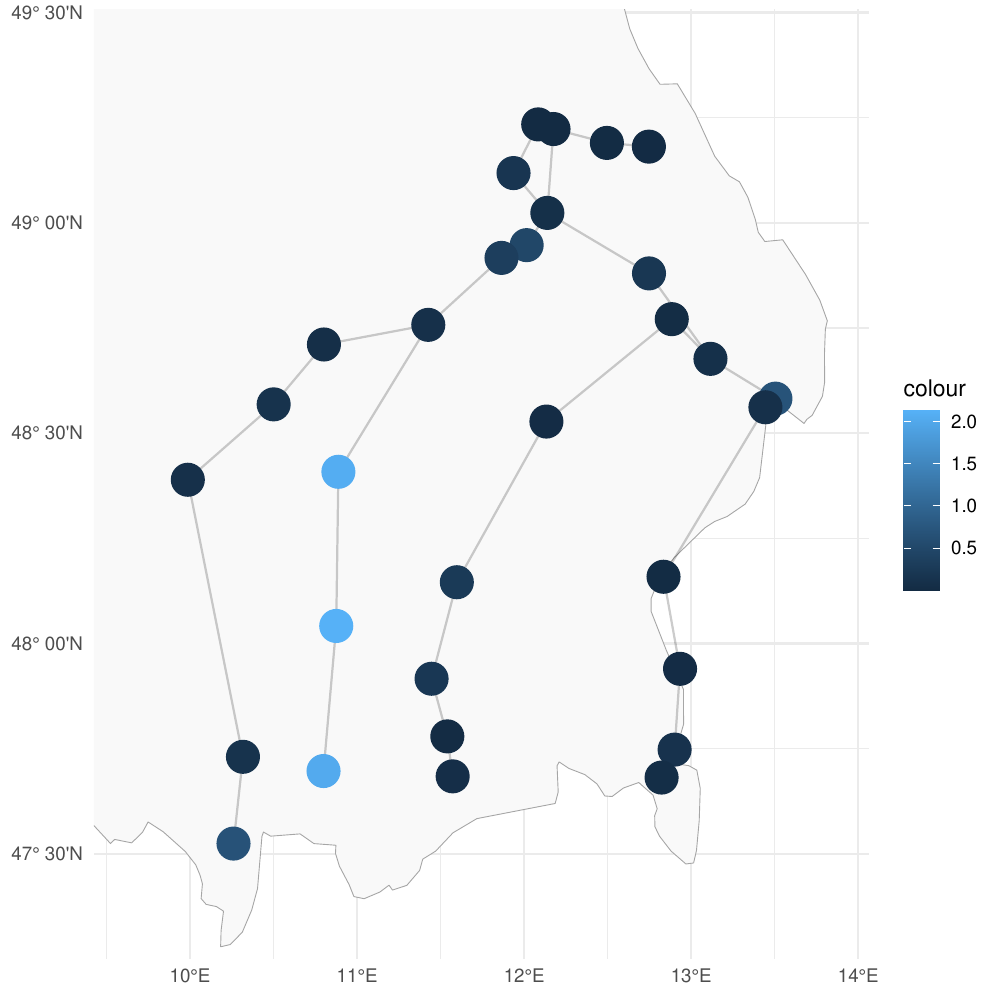}
    \end{subfigure}
    \begin{subfigure}[b]{0.26\textwidth}
      \centering
      \includegraphics[width=\textwidth]{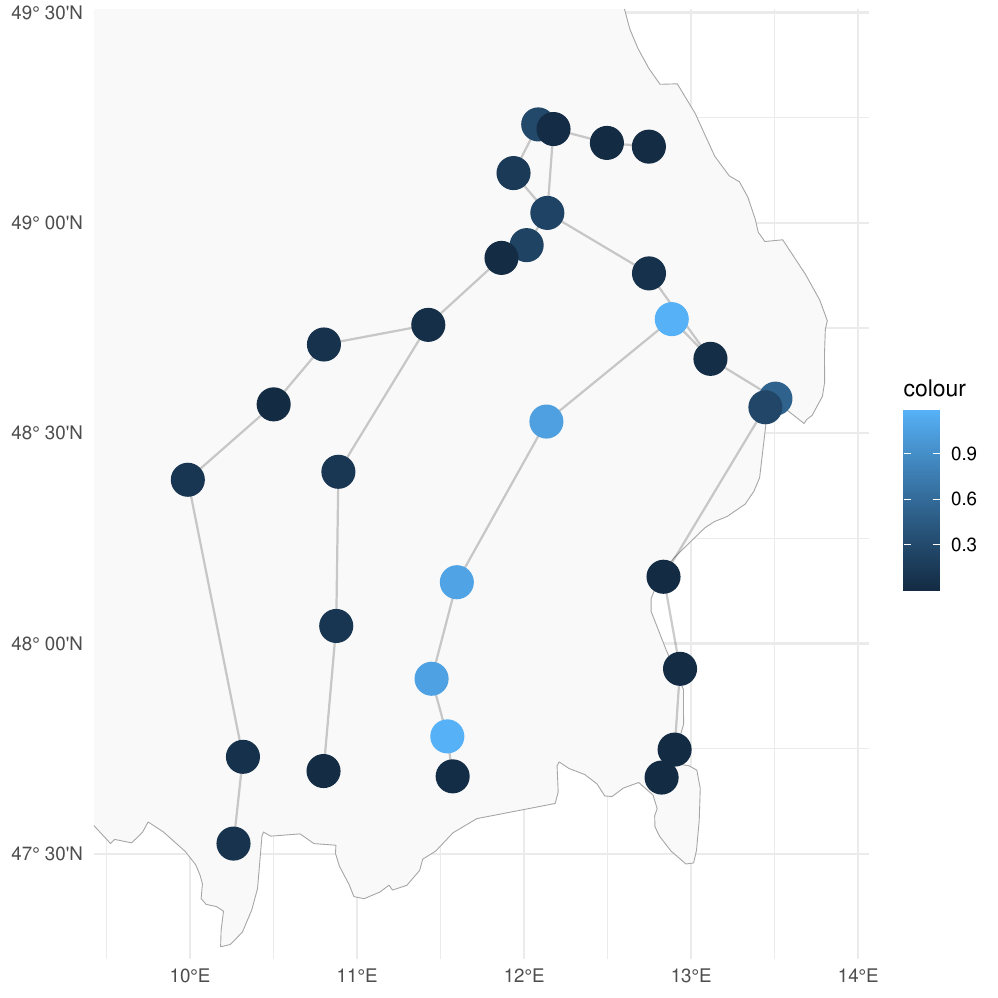}
    \end{subfigure}
    \begin{subfigure}[b]{0.26\textwidth}
      \centering
      \includegraphics[width=\textwidth]{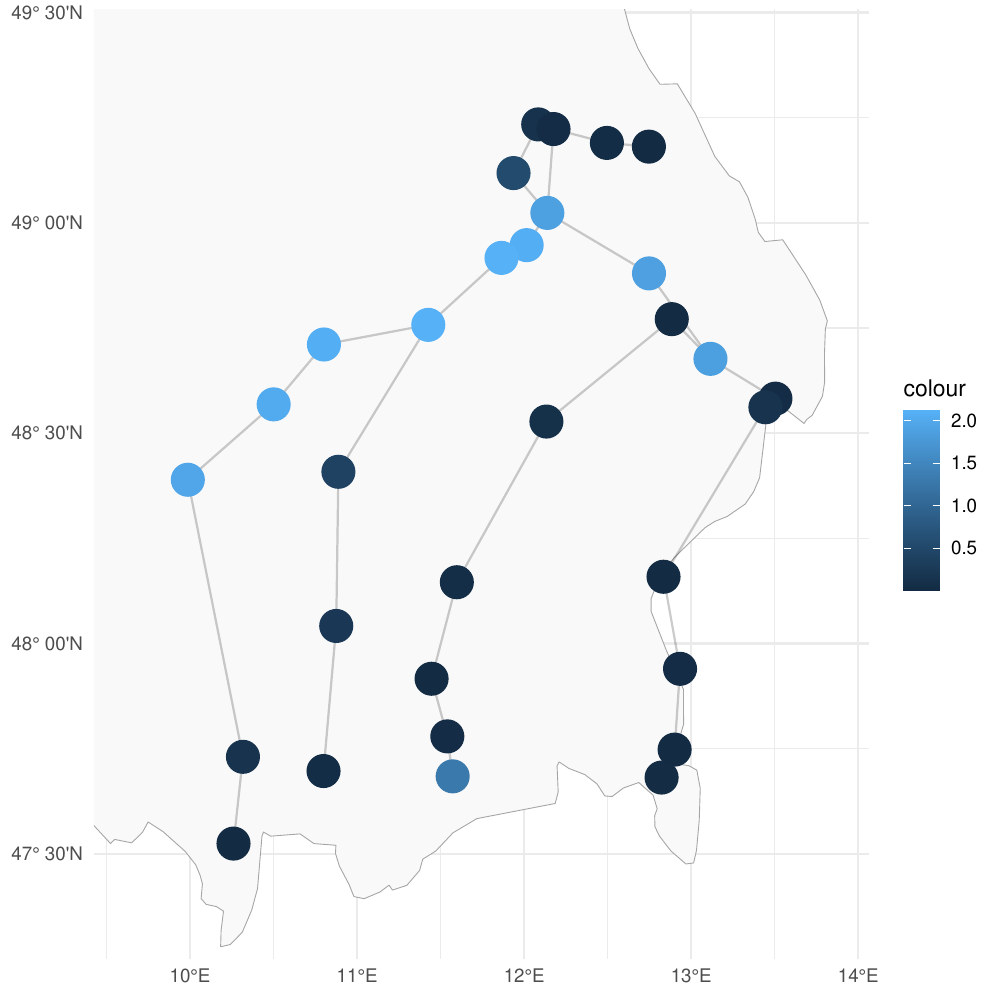}
    \end{subfigure}
    \begin{subfigure}[b]{0.26\textwidth}
      \centering
      \includegraphics[width=\textwidth]{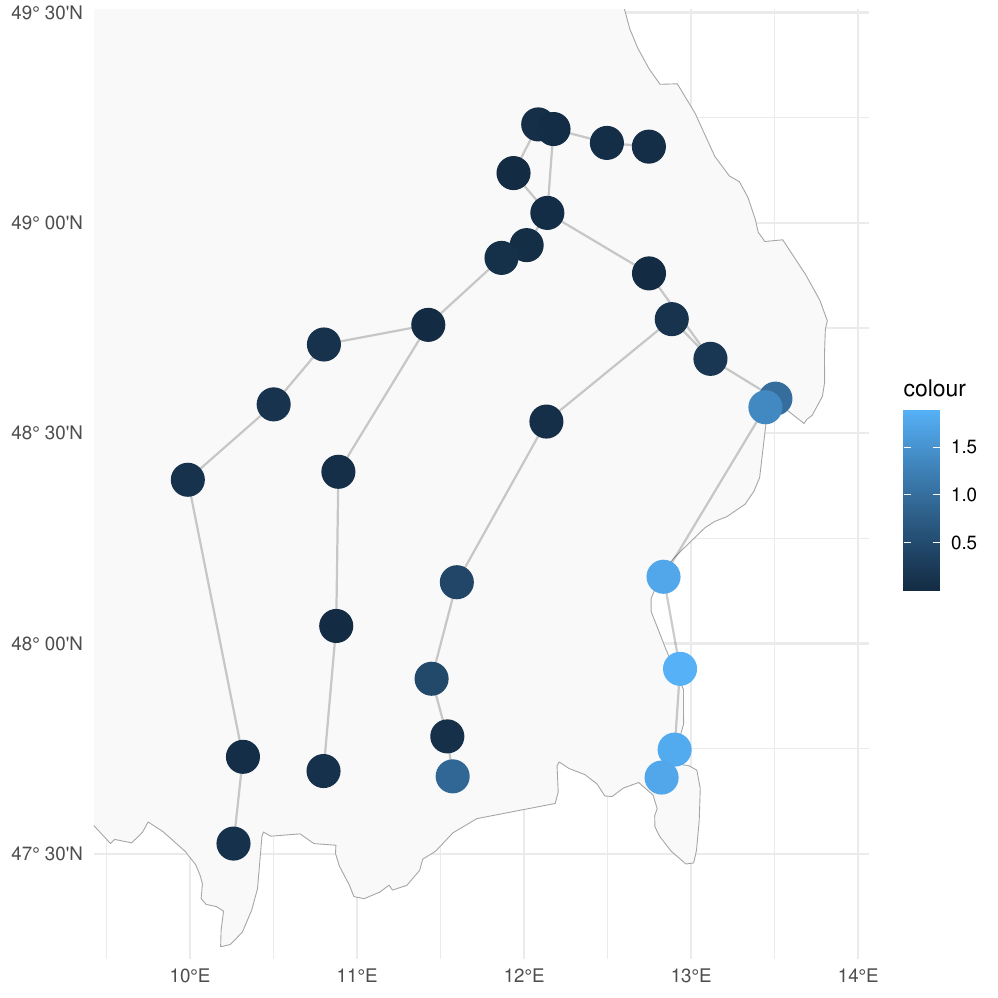}
    \end{subfigure}
    \begin{subfigure}[b]{0.26\textwidth}
      \centering
      \includegraphics[width=\textwidth]{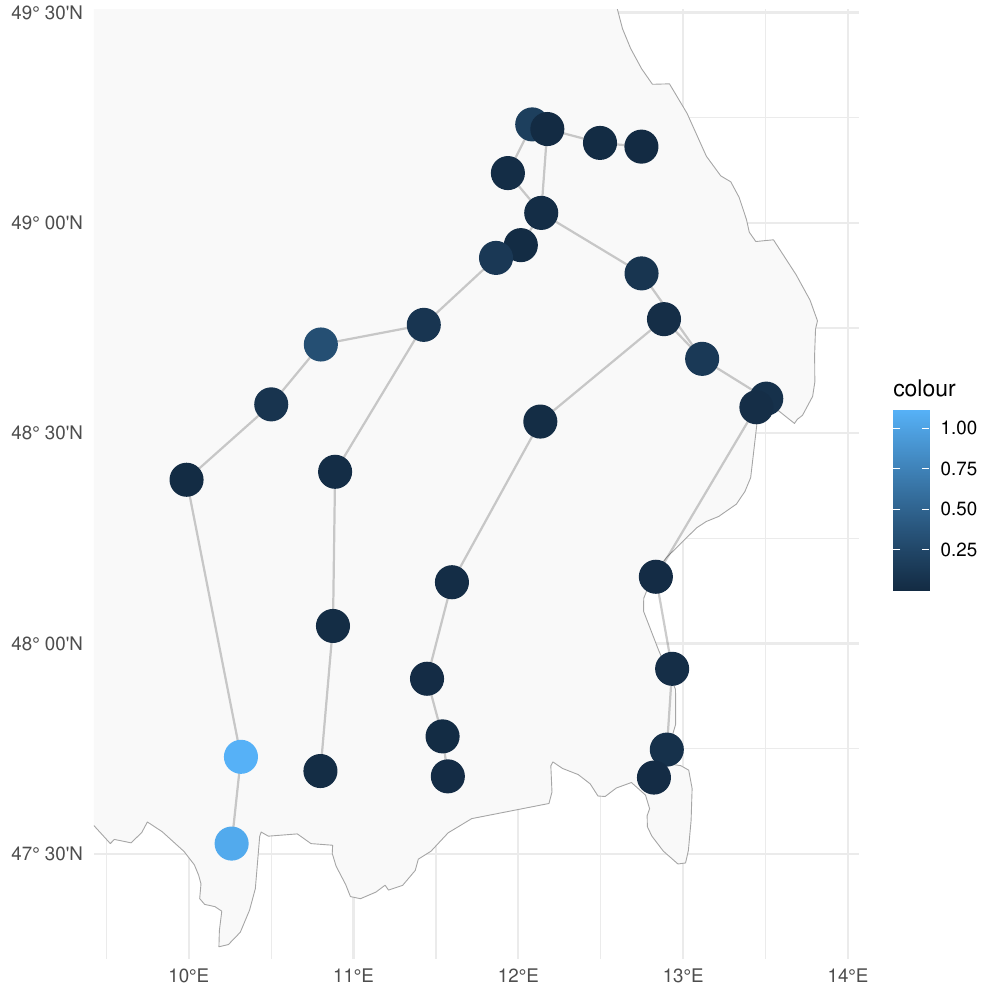}
    \end{subfigure}
    \begin{subfigure}[b]{0.26\textwidth}
      \centering
      \includegraphics[width=\textwidth]{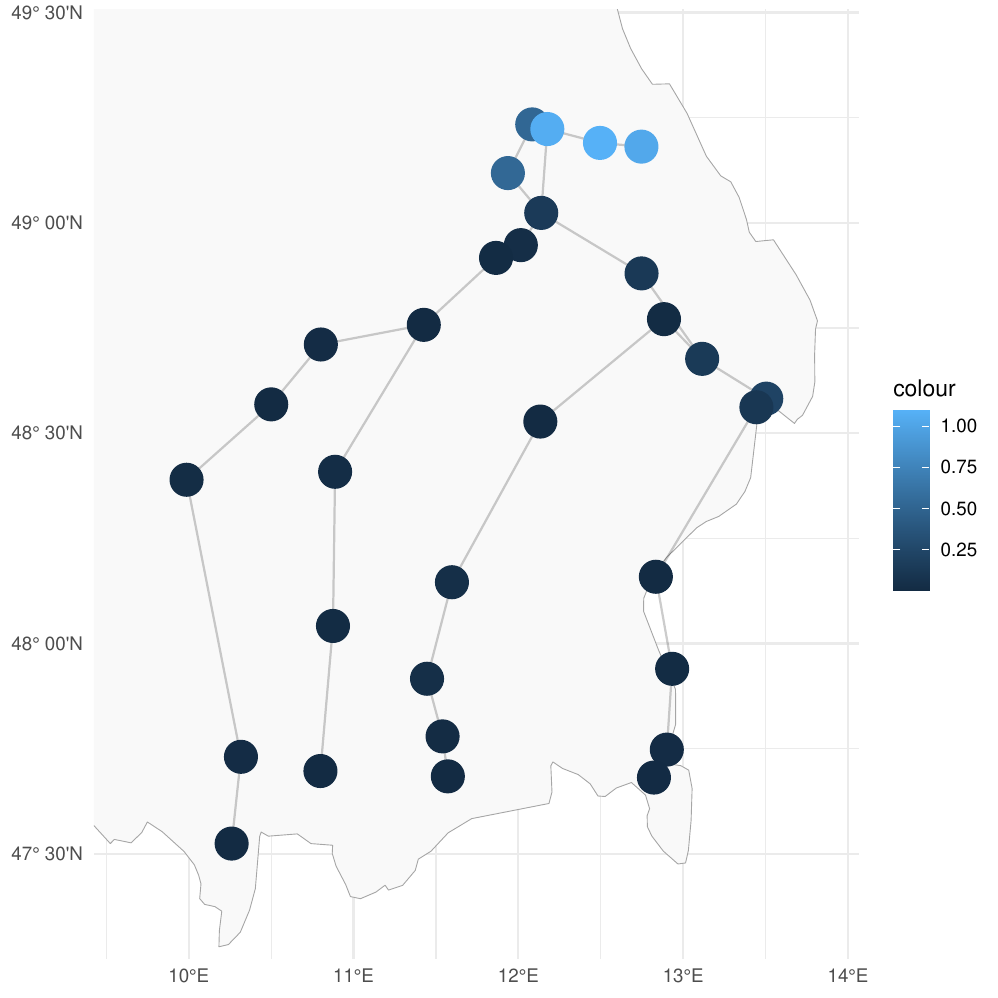}
    \end{subfigure}
    \caption{
        The columns of the estimated $B$ for $p=6$ mapped to their spatial locations ($B_{ij}$ corresponds to station $i$ and encoded state $j$) and colored such that larger $B_{ij}$ has a lighter blue color.
        The top row contains the first three columns and the second row the last three. Note that high values in a column $j$ of $B$ appear on the same riverarm as the point with the highest value in row $j$ of the matrix $W$, compare with Figure~\ref{fig:Wspat_danube}.
        }
    \label{fig:Bspat_danube}
\end{figure}
Similarly, we map the columns of the estimated matrix $B$ to the measurement stations in Figure~\ref{fig:Bspat_danube} to illustrate the generating set of the $\vee$-linear space generated by the columns of $B$. These columns can provide useful insight, 
since the reconstruction $B \diamond W \diamond X$ takes values in the $\vee$-linear combinations of the columns of $B$ almost surely. 
Therefore, these $6$ vectors can be seen as the boundary events when only 
the corresponding entry in $W \diamond X$ is large. 
Observe that these columns illustrated in Figure~\ref{fig:Bspat_danube} not only correspond to river arms, but can also show spatial dependence. In this approach, it differs from the approaches of PCA for extremes studied in \cite{DreesSabourin_pcame} and \cite{CooleyThibaud_ddhde}, where a lower dimensional subspace on which $S$ concentrates is identified. It also differs from the spherical $k$-means clustering for extremes studied in \cite{JanssenWan_kmce}, where the spectral measure is reduced to $k$ points of mass. For a more in-depth comparison we refer to Section~\ref{Suppl:Danube_Comp} in the Supplement, where we apply all methods to the Danube data set and discuss the results. 

We conclude that max-stable PCA is a reasonable tool 
for dimension reduction for this dataset and that further analysis of the encoded state $W \diamond X$ and the entries of the estimated matrix $B$ 
could provide insights into the behavior of extremes for the Danube basin.

{\section{Discussion}

In this work, we showed that a version of PCA using $\vee$-linear combinations and a suitable semimetric for max-stable distributions is a statistically viable procedure that can be used for dimension reduction in multivariate extremes and facilitate a better understanding of high dimensional extreme events in practice.

We demonstrated in Theorem~\ref{thm:spectralrep} that a class of models one might call generalized max-linear factor models given by a matrix $A \in [0, \infty)^{d \times p}$ and a $p$-variate max-stable random vector with $1$-Fréchet margins $Z$ by setting
$X := A \diamond Z$, appears as a model case that is perfectly reconstructable, if $A$ has at most $p$ $\vee$-linearly independent rows. This is in analogy to classic PCA, when the data only takes values in a $p$-dimensional subspace. Note that the rows in the matrix $A$ can potentially all be $\vee$-linearly independent even if $d >> p$, this has for example also been recognized in~\cite{Schlather_sgpca}.

We have shown a general statement in Theorem~\ref{thm:consistency} that minimizing some empirical statistical functional of a spectral measure estimator under reasonable assumptions yields an asymptotic minimizer. This statement entails max-stable PCA as a special case but could also be used to directly deduce a consistency result for other procedures.  Additionally, since we allow for any estimator of the spectral measure, this approach also lends itself to dependent data as long as the estimator of the spectral measure is consistent. Finally, we demonstrated that our procedure works well for finite samples in a simulation study and in the analysis of a benchmark dataset. 

Our theoretical results and our \emph{R}-package can be used directly to facilitate data analysis of multivariate extreme events with moderate to high asymptotic dependence and provide new insights by finding a simpler approximation by the reconstruction $B \diamond W \diamond X$, where $B$ provides $p$ prototypical extreme events that can be seen as building blocks of a good approximation to the data and tell which entries are usually large simultaneously. Additionally, the encoded state $W \diamond X$ can be used as a preprocessing step to reduce the dimension of the data while preserving key features. This might be helpful for applications like estimating the probabilities of extreme events or when fitting a parametric model to the whole dataset of dimension $d$ is infeasible but fitting it to a $p<d$ dimensional model is. Combining this with the matrix $B$, it could be used to create a simple procedure to approximately simulate high dimensional extreme events that is easy to interpret.
}

\FloatBarrier

\setcounter{section}{0}
\setcounter{equation}{0}
\setcounter{figure}{0}
\setcounter{lemma}{0}
\setcounter{proposition}{0}
\setcounter{example}{0}
\setcounter{corollary}{0}
\setcounter{definition}{0}

\renewcommand{\thelemma}{S\arabic{lemma}}
\renewcommand{\theproposition}{S\arabic{proposition}}
\renewcommand{\thecorollary}{S\arabic{corollary}}
\renewcommand{\theexample}{S\arabic{example}}
\renewcommand{\thedefinition}{S\arabic{lemma}}

\renewcommand\thesection{\ifcase\value{section}?\or S.1\or S.2\or S.3\or S.4\or S.5\or S.6\else ?\fi}
\renewcommand\theequation{S\arabic{equation}}
\renewcommand\thefigure{SF\arabic{figure}}

{\section{Supplementary results on max-linear combinations of max-stable random vectors}
\label{subsec:supp_results}}

We give some further motivation for using max-linear maps of max-stable random vectors by taking a closer look at the resulting distribution, in terms of its spectral measure. {The proofs of the statements below can be found in Section~\ref{subsec:supp_proofs}.}

 Our first lemma below combines the main result from~\cite{DeHaan_cmev} and Proposition 2.6 from~\cite{Janssen_tdme}.
\begin{lemma}\label{lem:maxlin_trafo}
  Let $X$ be a $d$-variate max-stable random vector with non-degenerate $1$-Fréchet margins given by a spectral measure $S$ as in~\eqref{Eq:spectral_meas_rep}. 
  Then for $H \in [0, \infty)^{p \times d}$, 
  $H \diamond X$ (cf.\ \eqref{Eq:max_matrix:prod}) is a $p$-variate max-stable random vector with $1$-Fréchet margins 
  and spectral measure
  \begin{equation}\label{eq:trafospecmeas}
    \tilde S(A) = \int_{\{a \in \S^{d-1}_+ \colon \| H \diamond a \| > 0 \}} \mathbbm 1_A \biggl( \frac{ H \diamond a}{ \| H \diamond a \|}\biggl) \| H \diamond a \| \, S(da)
    \quad \forall A \in \mathcal B(\S^{p-1}_+). 
  \end{equation}
\end{lemma}

For the next statement we need to revise some notation. Remember that for a measure $\mu$ on $\mathcal{B}(A)$ for some Borel set $A \subseteq \mathbb{R}^d$, its support is defined as 
\begin{equation}\label{supp_measure}
  \text{supp}(\mu) := \{x \in A \colon \mu(U(x)) > 0 , \quad \mbox{for all open neighborhoods } U(x) \text { of }x \}.
\end{equation}
If $\mu$ is a finite measure, then an important property we will use later on is that the support is equal to the (well-defined) minimal closed set which has measure $\mu(A)$, see~\cite{Parthasarathy_pmms} II.2, Theorem 2.1. 
With this definition in mind we give the following proposition, which, up to our knowledge, has not been shown so far but relates to the 
geometry of max-stable distributions, which was thoroughly investigated in~\cite{Molchanov_cgmsd}. 
\begin{proposition}\label{prop:support}
  Let $X$ be a $d$-variate max-stable random vector with non-degenerate $1$-Fréchet margins given by a spectral measure $S$ as in~\eqref{Eq:spectral_meas_rep}.
  Then, 
  \begin{equation}
    \label{eq:mevsupp}
    \mathrm{supp}(\P^X) = \biggl\{\bigvee_{i = 1}^d \alpha_i s_i \colon  \alpha_i \geq 0, s_i \in \mathrm{supp}(S) \biggl\}. 
  \end{equation}
\end{proposition}
For abbreviation, we borrow the notation
\begin{equation*}
  \vee\mathrm{-span}(\mathrm{supp}(S))
  := \biggl \{ \bigvee_{i=1}^d \alpha_i s_i \colon \alpha_i \geq 0, s_i \in \mathrm{supp}(S) \biggl \}
\end{equation*}
(accordingly $\vee\mathrm{-span}(A)$ is defined for any other subset $A$ of $\mathbb{R}^d$ taking the place of $\mbox{supp}(S)$) from~\cite{Stoev_srmsp} to denote the set on the right hand side of \eqref{eq:mevsupp}.
This allows us to characterize the support of a max-stable distribution with $1$-Fréchet margins transformed by a max-linear map. 
\begin{corollary}\label{cor:maxlin_trafo_img}
  Let $X$ be a $d$-variate max-stable random vector with non-degenerate $1$-Fréchet margins given by a spectral measure $S$ as in~\eqref{Eq:spectral_meas_rep}. Then for $H \in [0, \infty)^{p \times d}$, $p \in \mathbb{N}$,
  the distribution of the random vector $H \diamond X$  satisfies
  \begin{equation}
    \label{eq:spanrep}
    \mathrm{supp}(\P^{H \diamond X}) =  \mbox{cl}\left(\vee\mathrm{-span}\left(\left\{ H \diamond s: s \in \mathrm{supp}(S) \right \} \right)\right). 
  \end{equation}
If $H$ has no zero columns, then this simplifies to
   \begin{equation}
    \label{eq:spanrep:nocl}
    \mathrm{supp}(\P^{H \diamond X}) =  \vee\mathrm{-span}\left(\left\{ H \diamond s: s \in \mathrm{supp}(S) \right \} \right). 
  \end{equation}
\end{corollary}

\begin{remark} The assumption about the non-zero columns of $H$ is necessary for \eqref{eq:spanrep:nocl} to hold as the following example shows: Let $\|\cdot\|=\| \cdot \|_\infty$ and $\mbox{supp}(S)=\{(x,x^2,1)^T: x \in [0,1]\}$. Furthermore, let 
$$ H=\left(\begin{array}{ccc} 1 & 0 & 0 \\
0 & 1 & 0\end{array}\right),$$
thus violating the assumption about no zero columns. Then,
\begin{equation*} 
(1,0)^T \notin \{a(x,x^2)^T: a\geq 0, x \in [0,1]\}=\vee-\mbox{span}(\{(x,x^2)^T: x \in [0,1]\})
\end{equation*}
and it holds that 
\begin{align*}
\{a(x,x^2)^T: a\geq 0, x \in [0,1]\}
& =\vee-\mbox{span}(\{(x,x^2)^T: x \in [0,1]\})\\
& =\vee-\mbox{span}(\{H \diamond s: s \in \mbox{supp}(S)\}).
\end{align*}
But $n(1/n,1/n^2)^T \to (1,0)^T, n \to \infty$ and so $(1,0)^T \in \mbox{cl}(\vee-\mbox{span}(\{H \diamond s: s \in \mbox{supp}(\mathbb{S}_+^{d-1})\}))$.
\end{remark}
The previous statements show that the marginal distributions and the general structure of the support of a max-stable random vector $X$ with $1$-Fréchet margins are preserved by max-linear maps. Consequently, they are suitable candidates for mapping $X$ to a lower-dimensional space.

{\section{Examples of max-stable PCA}}

{\subsection{Example to illustrate non-uniqueness of max-stable PCA}
\label{subsec:uniqueness}
The following example illustrates a situation in which the max-stable PCA is not unique. We also discuss differences in classical and max-stable PCA regarding the topic of uniqueness. 
\begin{example}
\label{ex:nonunique}
Let $X$ be a bivariate random vector with $1$-Fréchet margins, independent components and marginal scale coefficients $\sigma_1 = \sigma_2 = 1$. In this case, the spectral measure $S$ of $X$ (for both the sum norm and the maximum norm) equals $\delta_{e_1} + \delta_{e_2}$, i.e.\ the sum of two Dirac measures with mass in the unit vectors $e_1$ und $e_2$, respectively. We use this to rewrite the reconstruction error of $X$ for max-stable PCA with $p = 1$ as follows: 
\begin{align}
 \nonumber   \tilde \rho(B \diamond W \diamond X, X) 
    & = \int_{\S^{d-1}_+} \sum_{k=1}^2 \lvert B_k \diamond W \diamond a - a_k \rvert S(da) \\
 \nonumber    & = \sum_{k=1}^2 \lvert B_k \diamond W \diamond e_1 - (e_1)_k \rvert + \lvert B_k \diamond W \diamond e_2 - (e_2)_k \rvert \\
 \label{Eq:recerror_Ex_Uniq}   & = \lvert B_1 W_1 - 1 \rvert + \lvert B_2 W_2 - 1 \rvert + B_1 W_2 + B_2 W_1. 
\end{align}
We want to show that all $B^T, W \in [0,\infty)^2$ which minimize this expression are given by
$$ B=(\lambda,0)^T, \; W=(\lambda^{-1},0) \;\;\; \mbox{and} \;\;\; B=(0,\lambda)^T, \; W=(0,\lambda^{-1}), \;\;\; \mbox{with } \lambda>0.$$
We get from \eqref{Eq:recerror_Ex_Uniq} that all such choices lead to $\tilde{\rho}(B \diamond W \diamond X, X)=1$ and we will show that they are the only possible combinations to attain this value and that it is minimal.  

Note first that in the optimal solution at least one entry of $B$ needs to be positive, since otherwise (for $B=(0,0)^T$) the reconstruction error from \eqref{Eq:recerror_Ex_Uniq} will be equal to 2 and thus the solution cannot be optimal. So, let us first assume that the first component of $B$ is positive. We can then always rescale $B$ and $W$ without changing $B \diamond W$ such that this first entry, $B_1$, equals 1, leading to 
\begin{equation}\label{Eq:simple_1_recerror_Ex_Uniq}
   \tilde \rho(B \diamond W \diamond X, X) = \lvert W_1 - 1 \rvert + \lvert B_2 W_2 - 1 \rvert + W_2 + B_2 W_1 . \end{equation}
 
Now, note that neither $W_1> 1$ nor $B_2W_2 > 1$ can hold for an optimal solution, as a decrease in $W_1, B_2$ or $W_2$ in this scenario will always lead to a decrease in the value of \eqref{Eq:simple_1_recerror_Ex_Uniq}. So, for the optimal solution, we have
\begin{equation}\label{Eq:simple_recerror_Ex_Uniq}
 \tilde \rho(B \diamond W \diamond X, X) = 1-W_1+ 1- B_2 W_2  + W_2 + B_2 W_1 .\end{equation}
 We now distinguish three possible options for $B_2$:
 \begin{enumerate}
 \item[$B_2=1$:] In this case the value in \eqref{Eq:simple_recerror_Ex_Uniq} equals 2, so this cannot be optimal.
 \item[$B_2>1$:] In this case $\lvert W_1 - 1 \rvert + B_2 W_1$ is minimized by $W_1=0$ and $\lvert B_2 W_2 - 1 \rvert + W_2$ is minimized by $W_2=B_2^{-1}$, leading to the reconstruction error $1+B_2^{-1}>1$. Again, this solution cannot be optimal.
 \item[$B_2<1$:] Analogous to the previous case, \eqref{Eq:simple_recerror_Ex_Uniq} is minimized by $W_1=1$ and $W_2=0$, leading to the reconstruction error $1+B_2>1$. This error now becomes minimal for $B_2=0$.
 \end{enumerate}
 So, if the first component of $B$ is positive and standardized to 1, the optimal solution is given by $B=(1,0)^T, \; W=(1,0)$. Analogously, if the second component is $B$ is positive and standardized to 1, the optimal solution is given by $B=(0,1)^T, \; W=(0,1)$. The optimal matrix $H=B \diamond W$ is thus not unique but given by

\begin{equation*}
    H^{(1)} = \begin{pmatrix}
        1 & 0 \\
        0 & 0 
    \end{pmatrix}
    \quad \mbox{ or } \;\;\; 
    H^{(2)} = \begin{pmatrix}
        0 & 0 \\
        0 & 1
    \end{pmatrix}.
\end{equation*}
This is in contrast to the uniqueness properties of classic PCA with $p=1$ for the analogous case of a bivariate Gaussian random vector with i.i.d. entries and variance equal to one. There, the solution is not unique as well, but indeed the projection to \textit{any} one dimensional subspace is optimal, yielding infinitely many optimal projection matrices. 
\end{example}
}

\subsection{Example to illustrate lack of symmetry properties}
We give an example where the optimal matrix pair $(B,W)$ does not satisfy $B=W^T$, highlighting the need to estimate both matrices individually in applications.
\begin{example}\label{ex:nosymmetry}
  Let $Z = (Z_1, Z_2)^T$ be a bivariate random vector with $1$-Fréchet margins, independent components and marginal scale coefficients $\sigma_1 = \sigma_2 = 1$. Set
  \begin{equation}\label{Eq:example_non_symm}
    A =\begin{pmatrix}
      1 & 0  \\
      0 & 1 \\
      \frac 1 2 & \frac 1 2 
    \end{pmatrix}
  \end{equation}
  to define
  \begin{equation*}
    X := A \diamond  Z. 
  \end{equation*}
 {Analogously to \eqref{Eq:recerror_Ex_Uniq} it follows that for any $H \in [0,\infty)^{3 \times 3}$ we have
 $$ \tilde{\rho}(H \diamond X, X)= \tilde{\rho}(H \diamond A \diamond Z, A \diamond Z)=\sum_{i,j=1}^3|(H \diamond A)_{ij}-A_{ij}|.$$ }
  Now we define
  \begin{equation*}
    H := \begin{pmatrix}
      1 & 0 & 0  \\
      0 & 1 & 0 \\
      \frac 1 2 & \frac 1 2 & 0
    \end{pmatrix} = 
    \begin{pmatrix}
      1 & 0  \\
      0 & 1 \\
      \frac 1 2 & \frac 1 2 
    \end{pmatrix} \diamond 
    \begin{pmatrix}
       1 & 0 & 0 \\
       0 & 1 & 0 
    \end{pmatrix}, 
  \end{equation*}
  then it is straightforward to verify 
  \begin{equation*}
    H \diamond A = A,
  \end{equation*}
  and we conclude that for $p=2$ there exist matrices $B, W^T \in [0, \infty)^{3 \times 2}$ such that the reconstruction error is zero.
  Now, assume there exists a \emph{symmetric} matrix $H^s = B \diamond W$
  which also satisfies $\tilde \rho(H^s \diamond X, X ) = 0$, i.e.\ such that
  \begin{equation}
    \label{eq:mlmex}
    H^s \diamond A = A. 
  \end{equation}
  Since $H^s$ is assumed to be symmetric, 
  we have that $H_{ij}^s = H_{ji}^s$ for all $i,j = 1, 2, 3$ 
  and \eqref{eq:mlmex} gives us the following set of equations
  \begin{align*}
    H_{11}^s \vee \frac 1 2 H_{13}^s & = 1, 
    && H_{22}^s \vee \frac 1 2 H_{23}^s = 1, \\
    H_{21}^s \vee \frac 1 2 H_{23}^s & = 0, 
    && H_{12}^s \vee \frac 1 2 H_{13}^s = 0, \\
    H_{31}^s \vee \frac 1 2 H_{33}^s & = \frac 1 2, 
    && H_{32}^s \vee \frac 1 2 H_{33}^s = \frac 1 2. 
  \end{align*}
  From the second row of the equations and by symmetry of $H^s$, it is immediately clear that $H_{12}^s = H_{21}^s = 0$, $H_{13}^s = H_{31}^s = 0 $ and $H_{23}^s = H_{32}^s = 0$. This simplifies the set of equations to 
  \begin{equation*}
    H_{11}^s = 1,
    H_{22}^s = 1, 
    H_{33}^s = 1,
  \end{equation*}
  hence, $H^s = \mathrm{id}_3$. 
  Since the identity matrix cannot be written as $B \diamond W$ for any matrices $B \in [0, \infty)^{3 \times 2}$ and $W \in [0, \infty)^{2 \times 3}$, we found a contradiction and conclude that no symmetric matrix can obtain the optimal reconstruction error in this case. 
\end{example}

\section{Additional results on statistics}
\label{subsec:supp_results_statistics}
The next lemma shows that there always exists a max-stable PCA and moreover that it can be found within a compact set that only depends on the dimensions $d, p$ and the scale parameters of $X$. This allows us to apply Theorem~\ref{thm:consistency} to deduce consistency of max-stable PCA.
\begin{lemma}\label{lem:compactset} 
    Let $X$ be a $d$-variate max-stable random vector with non-degenerate $1$-Fréchet margins.  
    Then there exists a constant $\kappa < \infty$ only depending on $d,p$ and the marginal scale parameters of $X$ such that a global minimizer of 
    \begin{equation} \label{eq:lemma_tf}
      (B,W) \mapsto \sum_{k=1}^d \rho(B_k \diamond W \diamond X, X_k) 
  \end{equation} 
    can be found within the compact set 
    \begin{equation}\label{eq:compset} 
      K = [0, \kappa]^{d \times p} \times [0, 1]^{p \times d}.  
    \end{equation} 
\end{lemma}
We give a short Corollary related to Theorem~\ref{thm:consistency} to formalize the consistency of max-stable PCA.
\begin{corollary}\label{cor:cons}
  The parameter set 
  \begin{equation*}
    K = [0, \kappa]^{d \times p} \times [0,1]^{p \times d}
  \end{equation*}
  from \eqref{eq:compset} and the cost function
  \begin{equation*}
    c \colon K \times \S^{d-1}_+ \to \R, \quad (B,W, a) \mapsto \sum_{k=1}^d  \lvert  B_k \diamond W \diamond a - a_k \rvert, 
  \end{equation*}
  satisfy the conditions of Theorem~\ref{thm:consistency}. 
\end{corollary}

\section{Proofs}
\subsection{Proofs of main results}

\begin{proof}[Proof of Lemma~\ref{metric_representation}]
By Lemma~\ref{lem:maxlin_trafo}, the vector $(Z_1,Z_2)^T:=(\bigvee_{i= 1}^d b_i X_i, \bigvee_{i = 1}^d c_i X_i)^T$ is bivariate max-stable with 1-Fr\'{e}chet margins. Thus, there exist two spectral functions $g_1, g_2 \in L_+^1([0,1])$ such that 
$$ P(Z_1\leq z_1, Z_2 \leq z_2)=\exp \left( - \int_{[0,1]} \bigvee_{i=1}^2 \frac {g_i(e)}{z_i} \, de \right) \mathbbm 1_{[0, \infty)^2}((z_1,z_2)^T)$$
for all $(z_1,z_2)^T \in \mathbb{R}^2$. From this we derive that the spectral functions of the max-stable random vector $(Z_1, Z_2, Z_1 \vee Z_2)^T$ are given by $g_1, g_2, g_1 \vee g_2 \in L_+^1([0,1])$.
The definition of $\rho$ in \eqref{eq:univmetric} together with \eqref{eq:scale-par-reps} and the fact that
  \begin{equation}
  \label{eq:absvalrep}
    x \vee y - x + x \vee y - y = \lvert x - y \rvert  
  \end{equation}
   for all $x,y \in \mathbb{R}$ then leads to
  \begin{align*}
  \rho(Z_1,Z_2) = \int_{[0,1]} \lvert g_1(e) - g_2(e) \rvert \, de 
  & = \int_{[0,1]} g_1(e) \vee g_2(e) - g_1(e) + g_1(e) \vee g_2(e) - g_2(e) \, de \\
  & = 2 \int_{[0,1]} g_1(e) \vee g_2(e) \, de - \int_{[0,1]} g_1(e) de - \int_{[0,1]} g_2(e) \, de \\
  & = 2 \sigma_{Z_1 \vee Z_2} - \sigma_{Z_1} - \sigma_{Z_2}, 
  \end{align*}
  where $\sigma_{Z_1 \vee Z_2}$ denotes the scale coefficient of the $1$-Fréchet random variable $Z_1 \vee Z_2$ and $\sigma_{Z_1}, \sigma_{Z_2}$ denote the scale coefficients of $Z_1$ and $Z_2$, respectively.
This allows us to write
  \begin{equation*} 
    \rho \left(\bigvee_{i= 1}^d b_i X_i, \bigvee_{i = 1}^d c_i X_i \right) 
    = 2 \sigma_{\bigvee_{i=1}^d b_i X_i  \vee c_i X_i} - \sigma_{\bigvee_{i=1}^d b_i X_i} - \sigma_{ \bigvee_{i=1}^d c_i X_i}. 
    \end{equation*}
    To show~\eqref{metric_representation}, we use the first equation in \eqref{eq:scale-par-reps} together with~\eqref{Eq:univ_maxlinS} to derive the three scale coefficients and obtain
    \begin{align*}
    \rho \left(\bigvee_{i= 1}^d b_i X_i, \bigvee_{i = 1}^d c_i X_i \right) 
    & = 2 \int_{\S^{d-1}_+} \bigvee_{i=1}^d (b_i \vee c_i) a_i
    - \bigvee_{i=1}^d b_i a_i
    - \bigvee_{i=1}^d c_i a_i \, S(da) \\
    & = \int_{\S^{d-1}_+} \left \lvert \bigvee_{i=1}^d b_i a_i - \bigvee_{i=1}^d c_i a_i \right \rvert \, S(da),  
  \end{align*} 
  where in the last step we used~\eqref{eq:absvalrep} again. The last equality~\eqref{metric_representation_sf} now follows by deriving the scale coefficients using~\eqref{eq:scale-par-reps} and~\eqref{Eq:univ_maxlinf} and calculating by the same arguments as above. We finally arrive at 
  \begin{align*}
  \rho \left(\bigvee_{i= 1}^d b_i X_i, \bigvee_{i = 1}^d c_i X_i \right) 
    & = 2 \int_{\S^{d-1}_+} \bigvee_{i=1}^d (b_i \vee c_i) f_i(e) \, de - \bigvee_{i=1}^d b_i f_i(e) 
    - \bigvee_{i=1}^d c_i f_i(e) \, de \\
    & = \int_{\S^{d-1}_+} \left \lvert \bigvee_{i=1}^d b_i f_i(e) - \bigvee_{i=1}^d c_i f_i(e) \right \rvert \, de. 
  \end{align*}
 \end{proof}

In order to prove Theorem~\ref{thm:spectralrep}], we need some technical setup and adapt some notation from tropical linear algebra to $L_+^1([0,1])$, facilitating our analysis of spectral functions. 
\begin{definition}
\label{def:maxlinindep}
If for a family of functions $f_1, \ldots, f_d \in L^1_+([0,1])$ there exists $i \in  \{1, \ldots, d\}$ and constants $\alpha_j \geq 0$, $j = 1, \ldots, d$ such that the following relation holds:  
\begin{equation*}
  f_i \aeeq \bigvee_{j \neq i} \alpha_j f_j, 
\end{equation*}
then we call $f_1, \ldots, f_d$ \textbf{$\vee$-linearly dependent}. If $f_1, \ldots, f_d$ are not $\vee$-linearly dependent, we call $f_1, \ldots, f_d$ \textbf{$\vee$-linearly independent}.
\end{definition}
\begin{remark}\label{rem:props_independent} If $f_1, \ldots, f_d \in L^1_+([0,1])$ are $\vee$-linearly independent and there exist coefficients $\alpha_j \geq 0, j=1, \ldots, d$ such that
$$ f_1\aeeq \bigvee_{j=1}^d \alpha_j f_j,$$
then both $\alpha_1<1$ and $\alpha_1>1$ lead to a contradiction to our assumptions and so necessarily $\alpha_1=1$.
\end{remark}
We are interested in the sets obtained by taking finite $\vee$-linear combinations of functions from $L^1_+([0,1])$ and denote by 
\begin{equation}
  \vee\mathrm{-span}(\{f_1, \ldots, f_d\}) := \left\{g \in L_+^1([0,1]) \colon g \aeeq \bigvee_{i=1}^d \alpha_i f_i, \alpha_i \geq 0 \right\}
\end{equation}
the set of functions that coincide with a $\vee$-linear combination of $f_1, \ldots, f_d$ almost everywhere. 

\begin{lemma}\label{lemma:basisuniqueL1}
  Let $f_1, \ldots, f_d, g_1, \ldots, g_q \in L_+^1([0,1])$ with $q \leq d$ be such that $g_1, \ldots, g_q$ are $\vee$-linearly independent and
  $$ \vee\mathrm{-span}(f_1, \ldots, f_d)=\vee\mathrm{-span}(g_1, \ldots, g_q). $$ 
  Then for each $j=1, \ldots, q$ there exist $c_j > 0, i_j \in \{1, \ldots, d\}$ such that
  \begin{equation*}
  g_j \aeeq c_j f_{i_j}.
  \end{equation*}
\end{lemma}
Lemma~\ref{lemma:basisuniqueL1} tells us that if two families of functions generate the same $\vee$-span and one of the families is $\vee$-linearly independent, then it is, up to scaling, already contained in the other family. This extends a result from~\cite{Butkovic_mls}, Corollary 3.3.11, to our infinite dimensional setting and is a key argument for our characterization of perfectly reconstructable random vectors.

\begin{proof}[Proof of Lemma~\ref{lemma:basisuniqueL1}]
  Because both $\vee$-spans of the families $f_1, \ldots, f_d$ and $g_1, \ldots, g_q$ coincide, we have that for each $g_j$ and $f_i$, there exist coefficients $\Gamma_ {ji}, \Lambda_{il} \geq 0$, such that 
  \begin{align}
    g_j & \aeeq \bigvee_{i=1}^d \Gamma_{ji} f_i, \label{eq:proof1}\\
    f_i & \aeeq \bigvee_{l=1}^q \Lambda_{il} g_l. \label{eq:proof2} 
    \end{align}
    Let now $j \in \{1, \ldots, q\}$ be fixed and replace the $f_i$ in~\eqref{eq:proof1} by the representation of~\eqref{eq:proof2}, so that we obtain 
    \begin{equation*}
        g_j \aeeq  \bigvee_{i=1}^d \Gamma_{ji} \left( \bigvee_{l=1}^q  \Lambda_{il}  g_l \right) \aeeq \bigvee_{l=1}^q \left( \bigvee_{i=1}^d \Gamma_{ji} \Lambda_{il} \right) g_l. 
    \end{equation*}
    The $\vee$-linear independence of $g_1, \ldots, g_q$ implies by Remark~\ref{rem:props_independent} that for every $j = 1, \ldots, q$, there exists an $i_j \in \{1, \ldots, d\}$ such that $\Gamma_{j {i_j}} \Lambda_{i_j j} = 1$ and furthermore the inequality 
    \begin{equation}
    \label{eq:proof3}
        \bigvee_{l \neq j} \left( \bigvee_{i=1}^d \Gamma_{ji} \Lambda_{il} \right) g_l \leq g_j
    \end{equation}
    holds almost everywhere. This leads to
    \begin{equation*}
        \Gamma_{j i_j} f_{i_j} 
        \aeeq \Gamma_{j i_j} \bigvee_{l=1}^q \Lambda_{i_j l} g_l 
        \aeeq \Gamma_{j i_j} \Lambda_{i_j j} g_j \vee \bigvee_{l \neq j} \Gamma_{j i_j} \Lambda_{i_j l} g_l
        \aeeq g_j \vee \bigvee_{l \neq j} \Gamma_{j i_j} \Lambda_{i_j l} g_l. 
    \end{equation*}
    In the right hand side of the last expression, the second term is in fact negligible, since $\Gamma_{j i_j} \Lambda_{i_j l}\leq \bigvee_{i=1}^d \Gamma_{ji} \Lambda_{il}$ and so by~\eqref{eq:proof3} 
    \begin{equation*}
        \bigvee_{l \neq j} \Gamma_{j i_j} \Lambda_{i_j l} g_l \leq \bigvee_{l \neq j} \left( \bigvee_{i=1}^d \Gamma_{ji} \Lambda_{il} \right) g_l \leq g_j.
    \end{equation*}
    Thus, we have shown that 
    $$  \Gamma_{j i_j} f_{i_j} \aeeq g_j$$
    and so the statement follows for $c_j := \Gamma_{j i_j}$. Since we have chosen $\Gamma_{j i_j}$ to satisfy $\Gamma_{j {i_j}} \Lambda_{i_j j} = 1$ we can furthermore conclude that $c_j>0$.
\end{proof}

\begin{proof}[Proof of Theorem~\ref{thm:spectralrep}]
\phantom{}
\begin{itemize}
 \item[$(ii) \Rightarrow (i)$]  
 Let $(ii)$ hold and set $(B,W)$ 
  as in~\eqref{eq:perfrecBW}. 
	Note that $W \diamond B = \mathrm{id}_{p}$, hence, we obtain
	\begin{equation*}
		B \diamond W \diamond X \disteq B \diamond W \diamond B \diamond Y = B \diamond Y \disteq X, 
	\end{equation*}
 which implies
  \begin{equation*}
    \tilde \rho(B \diamond W \diamond X, X) = \tilde \rho(B \diamond Y, B \diamond Y) = 0.
  \end{equation*}
  This shows (i) and that $(B,W)$ in \eqref{Eq:perfect_B_W} can indeed be chosen as in~\eqref{eq:perfrecBW}.
   \item[$(i) \Rightarrow (ii)$] 
  Assume that $(B,W) \in [0, \infty)^{d \times p} \times [0, \infty)^{p \times d}$ exists with $\tilde \rho(B \diamond W \diamond X, X) = 0$ and let $f = (f_1, \ldots, f_d)$ denote a vector of spectral functions from the representation \eqref{Eq:spectral_fun_rep} of the distribution function of $X$.
 It follows from $\tilde \rho(B \diamond W \diamond X, X ) = 0$ and~\eqref{metric_representation_sf} that
  \begin{equation}
    \label{eq:perfrec_int}
    \int_{[0,1]} \sum_{k=1}^d \lvert B_k \diamond W \diamond f(e) - f_k(e) \rvert \, de = 0. 
  \end{equation}
  Hence, the functions $f_k$ and $B_k \diamond W \diamond f$ coincide Lebesgue-almost everywhere. 
  Now, define functions 
  \begin{equation*}
    g_j \colon [0, 1] \to [0, \infty),  e \mapsto g_j(e):= W_j \diamond f(e), \quad j = 1, \ldots p. 
  \end{equation*}
  By the definition of $g_j$ and using~\eqref{eq:perfrec_int}, we get the following two relations
  \begin{align}
\nonumber    g_j \in & \left \{ \bigvee_{l=1}^d \beta_l f_l: \; \beta_l \geq 0 \right \} \subseteq L_+^1([0,1]), & j = 1, \ldots, p, \\
\label{Eq:fi_in_gs}    f_i \in & \left \{ \bigvee_{j=1}^p \alpha_j g_j: \; \alpha_j \geq 0 \right \} \subseteq L_+^1([0,1]),  & i = 1, \ldots, d,
  \end{align}
  which implies that
   $$ \left \{ \bigvee_{l=1}^d \beta_l f_l: \; \beta_l \geq 0 \right \} =  \left \{ \bigvee_{l=1}^d \beta_l \bigvee_{j=1}^p \alpha_{j,l} g_j: \; \beta_l, \alpha _{j,l} \geq 0 \right \} = \left \{ \bigvee_{j=1}^p \alpha_j g_j: \; \alpha_j \geq 0 \right \}.$$
   In case that the $p$ functions $g_1, \ldots, g_p$ are not $\vee$-linearly independent, we can recursively eliminate functions which can be written as a $\vee$-linear combination of the remainders, thereby leading (after suitable reordering of indices) to $q \leq p$ $\vee$-linearly independent functions $g_1, \ldots, g_q$ such that 
    \begin{equation}\label{Eq:all_gensets_same} \left \{ \bigvee_{j=1}^q \alpha_j g_j: \; \alpha_j \geq 0 \right \} = \left \{ \bigvee_{j=1}^p \alpha_j g_j: \; \alpha_j \geq 0 \right \} = \left \{ \bigvee_{l=1}^d \beta_l f_l: \; \beta_l \geq 0 \right \}.\end{equation}
  Then Lemma~\ref{lemma:basisuniqueL1} applied to the $\vee$-linearly independent $g_1, \ldots, g_q$ yields that there exist $q$ functions among $f_1, \ldots, f_d$ which are rescaled versions of $g_1, \ldots, g_q$. Apply again a suitable reordering to let this be $f_1, \ldots, f_q$, so that
    \begin{equation}\label{Eq:gs_scaled_fs}  g_i \aeeq  c_i f_i,\;\;\; i=1, \ldots, q \end{equation}
   for some suitable $c_i>0, i=1, \ldots, q$.  Combine \eqref{Eq:fi_in_gs}, \eqref{Eq:all_gensets_same} and \eqref{Eq:gs_scaled_fs} to see that 
   \begin{equation*} 
   f_i \in \left \{ \bigvee_{l=1}^q \beta_l f_l: \, \beta_l \geq 0 \right \}= \left \{ \bigvee_{l=1}^p \beta_l f_l: \, \beta_l \geq 0 \right \}, \quad i = 1, \ldots, d.
   \end{equation*}
  Hence, the $(d-p)$ spectral functions $f_{p+1}, \ldots, f_d$ are $\vee$-linear combinations of $f_{1}, \ldots, f_{p}$. This implies that there exists a matrix $\Lambda \in [0, \infty)^{(d-p) \times p}$ such that 
  \begin{equation*}
      f_i \aeeq \bigvee_{j=1}^p \Lambda_{(i - p), j} f_j, i = p+1, \ldots, d
  \end{equation*} 
  and we calculate
  \begin{align*}
  \P((X_{p+1}, \ldots, X_d)^T \leq z) 
  & = \exp \left( - \int_{[0,1]} \bigvee_{i=p+1}^d \frac{f_i(e)}{z_i} \, de \right) \mathbbm 1_{[0, \infty)^{d-p}}(z) \\
  & = \exp \left( - \int_{[0,1]} \bigvee_{i=p+1}^d \bigvee_{j=1}^p \frac{\Lambda_{(i-p),j} f_j(e)}{z_i} \, de \right) \mathbbm 1_{[0, \infty)^{d-p}}(z) \\
  & = \P(\Lambda \diamond (X_1, \ldots, X_p)^T \leq z).
    \end{align*}
    Therefore, the representation of $(X_{p+1}, \ldots, X_d)^T$ by spectral functions is completely determined by the spectral functions of $X_1, \ldots, X_p$ and the matrix $\Lambda$. 
    Now set $Y=(X_1, \ldots, X_p)^T$ to achieve \eqref{eq:maxlinrep} and thereby (ii). 
  \end{itemize}  
\end{proof}

\begin{proof}[Proof of Corollary~\ref{cor:mlf_rep}]
For the first implication, let $X$ be perfectly reconstructable. Then by Theorem~\ref{thm:spectralrep}, we have that $X$, up to permutation of entries, is in distribution of the form ~\eqref{eq:maxlinrep}.
This matrix in~\eqref{eq:maxlinrep} has exactly $p$ $\vee$-linearly independent rows given by the block matrix $\mathrm{id}_p$ and~\eqref{eq:genmaxlinfact} holds with this matrix, $l = p$ and $Z=Y$. 

For the converse statement, let the distribution of $X$ be given by~\eqref{eq:genmaxlinfact} and without loss of generality, let the first $p$ rows of $A$ be such that they can be used to construct the remaining $d-p$ rows by $\vee$-linear combinations, i.e. there exists a matrix $\Lambda \in [0,\infty)^{(d-p) \times p}$ such that $(A_i)_{i=p+1,\ldots,d}=\Lambda \diamond (A_i)_{i=1, \ldots, p}$. Next, define a $p$-variate random vector $Y$ by 
  \begin{equation*}
    Y_j := X_j=A_{j} \diamond Z, \quad j = 1, \ldots, p.
  \end{equation*}
  The vector $Y$ is clearly max-stable with $1$-Fréchet margins and for $X_i, i = p+1, \ldots, d$, we can write
  \begin{equation*}
    X_i= A_i \diamond Z = \left( \bigvee_{j=1}^p \Lambda_{(i-p), j} A_{j} \right) \diamond Z = \bigvee_{j=1}^p \Lambda_{(i-p),j} \left( A_{j} \diamond Z \right) = \bigvee_{j=1}^p \Lambda_{(i-p),j} Y_j.
  \end{equation*}
  Therefore, $X$ has a representation in distribution given by~\eqref{eq:maxlinrep}. Theorem~\ref{thm:spectralrep} thus yields that $X$ is perfectly reconstructable. 
\end{proof}

\begin{proof}[Proof of Theorem~\ref{thm:consistency}]
  
  \emph{(i)}
  The convergence \eqref{eq:weakconv_prob} and the continuity of $c$ imply that the random variables $\int_{\S^{d-1}_+} c(\theta, a) \, \hat S_n(da)$ are $\mathcal{A}, \mathbb{B}$-measurable for each $\theta \in K, n \in \mathbb{N}$. Furthermore, for all finite measures $\mu$ on $(\S^{d-1}_+, \mathcal B(\S^{d-1}_+))$ the function 
  \begin{equation}\label{eq:cintcont} \theta \mapsto  \int_{\S^{d-1}_+} c(\theta, a) \, \mu(da) \end{equation}
  is continuous due to continuity of $c$, compactness of $\S^{d-1}_+$ and since $\mu$ is finite. The function which maps $(\omega,\theta) \in \Omega \times K$ to 
  $$ \int_{\S^{d-1}_+} c(\theta, a) \, \hat S_n(da) $$ 
  is thus a Carath\'{e}odory function and by the Measurable Maximum Theorem of \cite{AliprantisBorger_ida}[Theorem 18.19] there exists a random sequence $(\hat \theta_n)_{n \in \N} \in K$ satisfying \eqref{eq:riskminimizers}. 

   \emph{(ii)} For the second step, note that the function $(\theta, a) \mapsto c(\theta, a)$ is uniformly continuous on the compact set $K \times \S^{d-1}_+$. Thus, for each $\epsilon>0$ there exist $\delta>0, l \in \mathbb{N}$ and $\theta_1, \ldots, \theta_l \in K$ such that $B(\theta_i,\delta):=\{\theta \in K: \|\theta-\theta_i\|<\delta\}, i=1, \ldots, l$ is a cover of $K$ and 
   \begin{equation}\label{eq:modcont} \sup_{\theta \in B(\theta_i,\delta), a \in \S^{d-1}_+}|c(\theta,a)-c(\theta_i,a)|<\epsilon \end{equation} 
   for all $i=1,\ldots, l$. Now choose $\theta \in K$ and let $i \in \{1, \ldots, l\}$ be such that $\theta \in B(\theta_i,\delta)$. Then \eqref{eq:modcont} yields
     \begin{align*}
 &  \biggl \lvert \int_{\S^{d-1}_+} c(\theta, a) \, \hat S_n(da)  - \int_{\S^{d-1}_+} c(\theta, a) \, S(da) \biggl \rvert \\
     \leq & \left\lvert 
    \int_{\S^{d-1}_+} c(\theta, a) \, \hat S_n(da) - \int_{\S^{d-1}_+} c(\theta_i, a) \, \hat S_n(da)\right \rvert + \left \lvert\int_{\S^{d-1}_+} c(\theta_i, a) \, \hat S_n(da) - \int_{\S^{d-1}_+} c(\theta_i, a) \, S(da)\right \rvert \\
  + & \left \lvert\int_{\S^{d-1}_+} c(\theta_i, a) \, S(da) - \int_{\S^{d-1}_+} c(\theta, a) \, S(da)\right \rvert \\
     \leq & \, \epsilon \hat S_n(\S^{d-1}_+) +\left \lvert\int_{\S^{d-1}_+} c(\theta_i, a) \, \hat S_n(da) - \int_{\S^{d-1}_+} c(\theta_i, a) \, S(da)\right \rvert + \epsilon S(\S^{d-1}_+).
   \end{align*}
   Thus, 
   \begin{align*} & \sup_{\theta \in K} \biggl \lvert \int_{\S^{d-1}_+} c(\theta, a) \, \hat S_n(da)  - \int_{\S^{d-1}_+} c(\theta, a) \, S(da)\biggl \rvert \\
   \leq & \epsilon ( \hat S_n(\S^{d-1}_+) + S(\S^{d-1}_+))+ \max_{i=1, \ldots, l} \left \lvert\int_{\S^{d-1}_+} c(\theta_i, a) \, \hat S_n(da) - \int_{\S^{d-1}_+} c(\theta_i, a) \, S(da)\right \rvert,
   \end{align*}
   and since the first and the second term converge in probability to $2 \epsilon S(\S^{d-1}_+)$ and 0, respectively, by letting $\epsilon \to 0$ we conclude
   \begin{equation}\label{eq:uniformconv} \sup_{\theta \in K} \biggl \lvert \int_{\S^{d-1}_+} c(\theta, a) \, \hat S_n(da)  - \int_{\S^{d-1}_+} c(\theta, a) \, S(da) \biggl \rvert \xrightarrow{P} 0 
   \end{equation}
   as $n \to \infty$. 

   To show now~\eqref{eq:ermresult}, note first that there exists a $\theta^* \in K$ that satisfies
   \begin{equation*}
    \int_{\S^{d-1}_+} c(\theta^{\ast}, a) \, S(da) = \inf_{\theta \in K} \int_{\S^{d-1}_+} c(\theta, a) \, S(da)
   \end{equation*}
   since $K$ is compact and the function defined in \eqref{eq:cintcont} (with $\mu$ replaced by $S$) is continuous. Then for any sequence $(\hat \theta_n)_{n \in \N} \in K$ satisfying \eqref{eq:riskminimizers} we have
   \begin{align*}
       0 & \leq \int_{\S^{d-1}_+} c(\hat \theta_n,a) \, S(da) - \int_{\S^{d-1}_+} c(\theta^{\ast}, a) \, S(da) \\
       & =  \int_{\S^{d-1}_+} c(\hat \theta_n,a) \, S(da) -  \int_{\S^{d-1}_+} c(\hat \theta_n,a) \,  \hat S_n(da) + \int_{\S^{d-1}_+} c(\theta^{\ast}, a) \, \hat S_n(da) - \int_{\S^{d-1}_+} c(\theta^{\ast}, a) \, S(da) \\
       & \;\;\;\; + \int_{\S^{d-1}_+} c(\hat \theta_n,a) \, \hat S_n(da) - \int_{\S^{d-1}_+} c(\theta^{\ast}, a) \, \hat S_n(da) \\
       &\leq \int_{\S^{d-1}_+} c(\hat \theta_n,a) \, S(da) -  \int_{\S^{d-1}_+} c(\hat \theta_n,a) \, \hat S_n(da) + \int_{\S^{d-1}_+} c(\theta^{\ast}, a) \, \hat S_n(da) - \int_{\S^{d-1}_+} c(\theta^{\ast}, a) \, S(da) \\
       & =: I + II.
   \end{align*}
   By \eqref{eq:uniformconv} both $I$ and $II$ converge to 0 in probability and so the statement follows. 
\end{proof}

\subsection{Proofs of supplementary results from Section~\ref{subsec:supp_results}}
\label{subsec:supp_proofs}

\begin{proof}[Proof of Lemma~\ref{lem:maxlin_trafo}]
  The max-linear map defined by $H$ is continuous, therefore measurable and $H \diamond X$ is thus a well-defined $p$-variate random vector, 
  the distribution function of which we denote by $\tilde G$. 
  For $z \in (0, \infty)^p$, we get from~\eqref{Eq:spectral_meas_rep} (and with the convention that $z_i/0:=\infty$ for $z_i>0$) that 
\begingroup
\allowdisplaybreaks
\begin{align*}
    \tilde G(z)
    & = \P \left(\bigvee_{j=1}^d H_{ij} X_j \leq z_i, i = 1, \ldots, p \right) \\
    & = \P \left( X_j \leq \frac {z_i} {H_{ij}}, i = 1,\ldots,p, \, j = 1,\ldots, d \right) \\
    & = \P \left(X_j \leq \bigwedge_{i=1}^p  \frac{z_i}{H_{ij}}, j = 1,\ldots,d \right) \\
    & = \exp \left( - \int_{\S^{d-1}_+} \bigvee_{j=1}^d \bigvee_{i=1}^p \frac {H_{ij} a_j} {z_i} \, S(da) \right) \\
    & = \exp \left( - \int_{\S^{d-1}_+} \bigvee_{i=1}^p \frac {H_i \diamond a} {z_i} \, S(da) \right). 
  \end{align*}
\endgroup
  We can rewrite the expression in the exponent as
  \begin{align*}
    \int_{\S^{d-1}_+}  \bigvee_{i=1}^p \frac {H_i \diamond a} {z_i} \, S(da) 
    & = \int_{ \{ a \in \S^{d-1}_+ \colon \| H \diamond a \| > 0\}}  \bigvee_{i=1}^p \frac {H_i \diamond a} {z_i} \, S(da) \\
    & = \int_{ \{ a \in \S^{d-1}_+ \colon \| H \diamond a \| > 0\}} \bigvee_{i=1}^p \frac {\frac {H_i \diamond a}{\| H \diamond a \|}} {z_i} \| H \diamond a \| \, S(da) \\
    & = \int_{\S^{p-1}_+}\bigvee_{i=1}^p \frac {\tilde a_i}{z_i} \, \tilde S(d \tilde a), 
  \end{align*}
  where $\tilde S$ is the measure appearing in~\eqref{eq:trafospecmeas} and the last equality follows by change of variable.  
  Since $\sup_{a \in \S^{d-1}_+}\|H \diamond a\|<\infty$, we have furthermore that $\tilde S$ is a finite measure on $\S^{p-1}_+$, and therefore a valid spectral measure. 
  This gives us that 
  \begin{equation*}
    \tilde G(z) = \exp \left( - \int_{\S^{p-1}_+} \bigvee_{i=1}^p \frac {a_i} {z_i} \, \tilde S(da) \right).
  \end{equation*}
  Clearly, this is the desired form of a max-stable distribution with $1$-Fréchet margins. 
\end{proof}
For the proof of Proposition~\ref{prop:support} we first recapitulate a useful representation of $X$ in terms of the points of a Poisson point processes, see \cite{DeHaan_srmsp}. For an introduction to Poisson point processes, see \cite{resnick_erp}, \cite{resnick_htp} or \cite{Kallenberg_rm}. 
\begin{lemma}[cf.\ Corollary 9.4.5 in \cite{DeHaanFerreira_evti}]\label{lem:PPP_rep} Let $X$ be a $d$-variate max-stable random vector with non-degenerate $1$-Fréchet margins given by a spectral measure $S$ as in~\eqref{Eq:spectral_meas_rep}. Furthermore, let $(U_k, A_k)_{k \in \N}$ be a Poisson point process on $\R \times \S^{d-1}_+$ 
  with intensity measure $u^{-2} \mathbbm 1_{(0, \infty)}(u) du \times S(da)$. Then, 
  \begin{equation*}
    X \disteq \bigvee_{k \in \N} U_k A_k. 
  \end{equation*}
\end{lemma}
\begin{proof} Denote by $N =\sum_{k=1}^\infty \delta_{(U_k,A_k)}$ the random measure associated with $(U_k, A_k)_{k \in \mathbb{N}}$. Then, for any $z \in [0, \infty)^d$, it holds that
  \begin{align*}
    \P \left(\bigvee_{k \in \N} U_k A_k \leq z \right)
    & = \P \left( N\bigl(\{ (u, a) \in (0, \infty) \times \S^{d-1}_+ \colon u a \nleq z \} \bigl) = 0 \right)\\
    & = \exp \left(-\int_{\S^{d-1}_+} \int_0^{\infty} \frac 1 {u^2}\mathbbm 1_{\{ u a \nleq z \}} \, du \, S(da) \right) \\
    & = \exp \left(-\int_{\S^{d-1}_+} \int_0^{\infty} \frac 1 {u^2}\mathbbm 1_{ \bigl\{u > \bigwedge_{i=1}^d \frac{z_i} {a_i}\bigl\} } \, du \, S(da) \right)\\
    & = \exp \left(-\int_{\S^{d-1}_+} \int_{\bigwedge_{i=1}^d \frac{z_i} {a_i}}^{\infty} \frac 1 {u^2} \, du \, S(da) \right) \\
    & = \exp \left(-\int_{\S^{d-1}_+} \bigvee_{i=1}^d \frac{a_i} {z_i} \, S(da) \right)  \\
    & = \P(X \leq z). 
  \end{align*}
\end{proof}
\begin{proof}[Proof of Proposition~\ref{prop:support}]
We show \eqref{eq:mevsupp} by double inclusion. 

\emph{1.)} To show the first inclusion $\mathrm{supp}(\P^X) \subseteq \vee\mathrm{-span}(\mathrm{supp}(S))$, use the point process representation of Lemma~\ref{lem:PPP_rep} and
  note that $\{U_k, k \in \mathbb{N}\} \subseteq [0, \infty)$ and $\{A_k, k \in \mathbb{N}\} \subseteq \mathrm{supp}(S)$ almost surely, since 
  $S((\mathrm{supp}(S))^c)=0$ by~\cite{Parthasarathy_pmms} II.2, Theorem 2.1. 
  Additionally, since almost surely only finitely many $U_k, k \in \mathbb{N}$, exceed a given positive threshold, at most $d$ points from $(U_k, A_k)_{k \in \N}$ suffice to represent the componentwise maximum of $\bigvee_{k \in \N} U_k A_k$, 
  yielding $\bigvee_{k \in \N} U_k A_k \in \vee\mathrm{-span}(\mathrm{supp}(S)) $ almost surely. 
  Since $X \disteq \bigvee_{k \in \N} U_k A_k$, we thus get
  \begin{equation*}
    \P^X(\vee\mathrm{-span}(\mathrm{supp}(S)) ) = \P \left(\bigvee_{k \in \N} U_k A_k \in \vee\mathrm{-span}(\mathrm{supp}(S)) \right) = 1. 
  \end{equation*}
  This yields $\mathrm{supp}(\P^X) \subseteq \vee\mathrm{-span}(\mathrm{supp}(S))$, by the aforementioned equivalent definition of the support as minimal closed set of full measure, see~\cite{Parthasarathy_pmms} II.2, Theorem 2.1.  

\emph{2.)} To show the reverse inclusion $\vee\mathrm{-span}(\mathrm{supp}(S)) \subseteq \mathrm{supp}(\P^X)$, let first $x$ be in $\vee\mathrm{-span}(\mathrm{supp}(S))$, which implies that there exist $\alpha_i \geq 0$ and $s_i \in \mathrm{supp}(S)$,  $i = 1, \ldots, d$, 
  such that 
  \begin{equation}\label{Eq:span-rep}
    x = \bigvee_{i=1}^d \alpha_i s_i. 
  \end{equation} As long as $x > \mathbf 0_d$, we can choose the above representation in a way such that $\alpha_i s_i, i=1, \ldots, d,$ are pairwise different, which we will assume in the following. Now, for every open neighborhood $U(x)$ of $x$, there exists an $\epsilon > 0$ such that
  the open rectangle  $(x - \epsilon \mathbf 1_d, x + \epsilon \mathbf 1_d)$ is contained in $U(x)$. 
  Moreover, by simple continuity arguments, there exists a $\delta > 0$ such that
  \begin{equation*}
    \left( \bigvee_{i=1}^d (\alpha_i - \delta)  (s_i - \delta \mathbf 1_d), \bigvee_{i=1}^d (\alpha_i + \delta) (s_i + \delta \mathbf 1_d) \right) \subseteq (x - \epsilon \mathbf 1_d, x + \epsilon \mathbf 1_d) \subseteq U(x). 
  \end{equation*}
  Define now the sets
  \begin{align*}
    B_{i} & := (\alpha_i - \delta, \alpha_i + \delta) \times \left( (s_i - \delta \mathbf 1_d, s_i + \delta \mathbf 1_d) \cap \S^{d-1}_+\right) \subseteq \mathbb{R} \times \S^{d-1}_+
  \end{align*}
  for $i=1, \ldots, d$ and note that by our assumption about the representation \eqref{Eq:span-rep}, these sets are pairwise disjoint if $\delta$ is chosen small enough, which we will again assume in the following. 
  
  Next, for the point process representation of $X$ from Lemma~\ref{lem:PPP_rep}, we evaluate the void probabilities of the random measure $N=\sum_{k=1}^\infty \delta_{(U_k, A_k)}$ on the sets $B_i$, i.e. $P(N(B_i)=0), i=1, \ldots, d$. We have for each $i = 1, \ldots, d$ that
  \begin{align}
    & \P(N(B_i)=0)\nonumber \\
    = & \P \left(N \left((\alpha_i - \delta, \alpha_i + \delta) \times \left( (s_i - \delta \mathbf 1_d, s_i + \delta \mathbf 1_d) \cap \S^{d-1}_+\right) \right) = 0 \right) \nonumber \\
    = & \exp \left( - \int_{0 \vee (\alpha_i - \delta)} ^{\alpha_i + \delta} \int_{ (s_i - \delta \mathbf 1_d, s_i + \delta \mathbf 1_d) \cap \S^{d-1}_+} \frac 1 {u^2} \, S(da) \, du \right) \nonumber \\
    = & \exp \left( - S\left( (s_i - \delta \mathbf 1_d, s_i + \delta \mathbf 1_d) \cap \S^{d-1}_+ \right) \int_{0 \vee(\alpha_i - \delta)}^{\alpha_i + \delta} \frac 1 {u^2} \, du \right). \nonumber
  \end{align}
  Since $s_i \in \mathrm{supp}(S)$, the set $(s_i - \delta \mathbf 1_d, s_i + \delta \mathbf 1_d) \cap \S^{d-1}_+$ has positive measure with respect to $S$ and 
  $(\alpha_i - \delta, \alpha_i + \delta)$ has positive (not necessary finite) measure with respect to $u^{-2} du$, hence it holds that $P(N(B_i)=0)<1$ or, equivalently, 
$$
    \P(N(B_i) > 0) >0, \;\;\; i=1, \ldots, d.
$$
  Since the sets $B_i, i=1, \ldots, d,$ were designed to be disjoint, we can conclude from the properties of a Poisson point process that 
  \begin{equation} \label{eq:pos_hittingprob} \P(N(B_i)>0, i=1, \ldots, d)=\prod_{i=1}^d \P(N(B_i) > 0)>0. \end{equation}
  
  To conclude our argument, define the exceedance set
  \begin{equation*}
    E(x, \epsilon) := \{ (u,a)^T \in \R \times \S^{d-1}_+ \colon u a \nleq x +\epsilon \mathbf{1}_d \}. 
  \end{equation*}
  
   Since $x +\epsilon \mathbf 1_d > 0$ componentwise, by Lemma~\ref{lem:PPP_rep} it holds that 
  \begin{equation}\label{Eq:exceedance_pos}
\P(N(E(x, \epsilon)) = 0) = \P \left(\bigvee_{k \in \N} U_k A_k \leq x + \epsilon \mathbf 1_d \right) = G(x + \epsilon \mathbf 1_d) > 0. 
  \end{equation}
  Finally, observe that 
\begin{equation}\label{Eq:PPP-inclusion} \{N(B_i) > 0, i=1, \ldots, d \} \cap  \{N(E(x, \epsilon)) = 0 \}  \subseteq \left \{ \bigvee_{k \in \N} U_k A_k \in U(x)\right \}, \end{equation}
  where the upper interval bound holds due to $\{N(E(x, \epsilon)) = 0 \}$ and the lower interval bound holds by $\{N(B_i) > 0, i=1, \ldots, d \}$, see Figure~\ref{fig:proofsets} for an illustration of this fact for $d=2$. 

  \begin{figure}[h!tb]
  \centering
  \resizebox{0.65\columnwidth}{!}{
    \begin{tikzpicture}
    \draw[style=help lines] (0,0) grid (6, 6);

    \fill[black] (0, 0) node[below left]{0}; 
    \fill[black] (1, 0) node[below] {1}; 
    \fill[black] (2, 0) node[below] {2}; 
    \fill[black] (3, 0) node[below] {3}; 
    \fill[black] (4, 0) node[below] {4}; 
    \fill[black] (5, 0) node[below] {5}; 
    \fill[black] (6, 0) node[right] {$x_1$}; 

    \fill[black] (0, 1) node[left] {1}; 
    \fill[black] (0, 2) node[left] {2}; 
    \fill[black] (0, 3) node[left] {3}; 
    \fill[black] (0, 4) node[left] {4}; 
    \fill[black] (0, 5) node[left] {5}; 
    \fill[black] (0, 6) node[above] {$x_2$}; 

\draw[fill = gray] (0,1) -- (1,0) -- (0,0) -- cycle; 


\fill[gray!20, fill opacity = 0.9] (0,5) -- (6, 5) -- (6,6) -- (0, 6) -- cycle;
\fill[gray!20, text = black, fill opacity = 0.9] (6,6) -- (6, 0) -- (4.5, 0) -- (4.5, 6) -- cycle;
\fill[black] (3.4,5.5) node[above right] {\fontsize{6}{12} \selectfont Exceedance set $E(x, \epsilon)$};

\draw[fill=gray! 50, thick] (3.5,4) circle (1.5cm);

\draw[fill = red] (3.92, 0.98) -- (2.94, 1.96) -- (3.06, 2.04) -- (4.08, 1.02) -- cycle node[below] {\fontsize{6}{12} \selectfont $B_2$};
\fill[black] (3.5, 1.5) circle (2pt);

\draw[fill = red] (0.49, 4.41) -- (1.47, 3.43) -- (1.53, 3.57) -- (0.51, 4.59) -- cycle node[left] {\fontsize{6}{12} \selectfont $B_1$};
\fill[black] (1, 4) circle (2pt); 

\draw[fill = gray] (2.5, 5) -- (4.5, 5) -- (4.5, 3) -- (2.5, 3) -- cycle node[below right] {\fontsize{4}{3}\selectfont \hspace{-0.15cm} $(x - \epsilon \mathbf 1_d, x + \epsilon \mathbf 1_d)$}; 

\draw[fill = orange] (2.94, 3.43) -- (2.94, 4.59) -- (4.08, 4.59) -- (4.08, 3.43) -- cycle node[above right]
{\fontsize{4}{12}\selectfont };

\fill [black] (3.5,4) circle (1pt) node[above] {\fontsize{4}{3}\selectfont $x$};
\fill (3.1, 5.2) node[right] {\fontsize{6}{3}\selectfont$U(x)$}; 
\fill (-0.1,-0.1) node[above right] {\fontsize{5}{3}\selectfont  $1$-sphere}; 

\draw (0, 5) -- (2.5, 5); 
\draw (4.5, 0) -- (4.5, 3); 

\draw[red] (0.51, 4.59) -- (2.94, 4.59); 
\draw[red] (1.47, 3.43) -- (2.94, 3.43); 

\draw[red] (4.08, 1.02) -- (4.08, 3.43); 
\draw[red] (2.94, 1.96) -- (2.94, 3.43); 

\draw[red] (0.1, 0.9) -- (0.49, 4.41);
\draw[red] (0.3, 0.7) -- (1.47, 3.43);

\draw[red] (0.8, 0.2) -- (3.92, 0.98);
\draw[red] (0.6, 0.4) -- (2.94, 1.96);

\end{tikzpicture}
}
\caption{
    A visualization of the sets appearing in the proof for $d = 2$ and $x = (3.5,4)$ in Cartesian coordinates. 
    The neighbourhood $U(x)$ is given by a circle of radius $r = 1.5$ around $x$ and
    $\epsilon = 1$ defines the dark grey rectangle $(x - \epsilon  \mathbf 1_2, x + \epsilon \mathbf 1_2)$. 
    Let $s_1 = (\frac 7 {10}, \frac 3 {10}), s_2 = (\frac 1 5, \frac 4 5)$, then we write $x = \alpha_1 s_1 \vee \alpha_2 s_2$
    with $\alpha_1 =  \alpha_2 = 5$ and setting $\delta = \frac 1 {10}$ yields the sets $B_1$ and $B_2$, where the graphic shows the preimage of those sets under the polar coordinate transform.
    The orange set contains all coordinatewise maxima of one point from the set $B_1$ and one point from the set $B_2$. 
    The exceedance set $E(x, \epsilon)$ is the light grey $L$ shaped set. 
}
\label{fig:proofsets}
\end{figure}
  Now, since $B_1, \ldots, B_d, E(x, \epsilon)$ are pairwise disjoint and because of \eqref{eq:pos_hittingprob}, \eqref{Eq:exceedance_pos} and \eqref{Eq:PPP-inclusion}, we finally arrive at
  \begin{align*}
    \P^X(U(x))=\P\left(\bigvee_{k \in \N} U_k A_k \in U(x) \right) 
  & \geq \P \left(\{N(B_i) > 0, i=1, \ldots, d\} \cap \{N(E(x, \epsilon)) = 0 \} \right) \\
  & = \P(N(B_i) > 0, i=1, \ldots, d) \, \P(N( E(x, \epsilon)) = 0) \\
  & > 0. 
  \end{align*}
Thus, each $x \in \vee\mathrm{-span}(\mathrm{supp}(S))\cap (0, \infty)^d$ is an element of $\mathrm{supp}(\P^X)$. Since the support is a closed set, the same statement follows for all $x \in \vee\mathrm{-span}(\mathrm{supp}(S))$, and thus we have shown both inclusions. 
\end{proof}

\begin{proof}[Proof of Corollary~\ref{cor:maxlin_trafo_img}]
It follows from~\cite{Bourbaki_int}, Chapter \Romannum{5}, §6 Corollary 4 that \begin{equation}\label{Eq:support_Bourbaki}
  \mathrm{supp}(\P^{H \diamond X}) = cl \left( \left\{ H \diamond x: x \in \mathrm{supp}(\P^X) \right\} \right). 
  \end{equation}
  Proposition~\ref{prop:support} gives an explicit representation for $\mathrm{supp}(\P^X)$. The bilinearity of $\diamond$ allows us to rewrite 
  \begin{align}
\nonumber \left\{ H \diamond x, x \in \mathrm{supp}(\P^X) \right\} 
\nonumber    & = \left\{ H \diamond \bigvee_{i=1}^d \alpha_i s^i: \alpha_i \geq 0, s^i \in \mathrm{supp}(S) \right\} \\
\nonumber    & =  \left\{ \bigvee_{i=1}^d \alpha_i (H \diamond s^i): \alpha_i \geq 0, s^i \in \mathrm{supp}(S) \right\} \\
\nonumber    & =  \left\{ \bigvee_{i=1}^p \alpha_i (H \diamond s^i): \alpha_i \geq 0, s^i \in \mathrm{supp}(S) \right\}\\
\label{Eq:H_diamond_S}    & =  \vee\mathrm{-span}(\{H \diamond s: s \in \mathrm{supp}(S)\}), 
  \end{align}
 where we used in the pen-ultimate equality that $\bigvee_{i=1}^d \alpha_i (H \diamond s^i)$ can always be represented as a maximum over $p$ components, since the dimension of this vector is $p$. Now, \eqref{Eq:support_Bourbaki} and \eqref{Eq:H_diamond_S} give \eqref{eq:spanrep}.

For the remainder of this proof, we assume that $H$ has no zero columns. In this case there exists some $\epsilon>0$ such that $\vee_{i=1}^p H_{ij}>\epsilon$ for all $j=1, \ldots, d$ and thus for any $s=(s_1, \ldots, s_d)^T \in \mathbb{S}^{d-1}_+$ it follows that 
\begin{align*}
\|H \diamond s\|_\infty = \bigvee_{i=1}^p \bigvee_{j=1}^d H_{ij}s_j =\bigvee_{j=1}^d  s_j \bigvee_{i=1}^p  H_{ij} \geq \epsilon \bigvee_{j=1}^d s_j=\epsilon \|s \|_\infty.
\end{align*}
By equivalence of norms there exists some $c>0$ such that $\|s \|_\infty>c$ for all $s \in \mathbb{S}^{d-1}_+$ and so 
 $$\| H \diamond s\|_\infty > \epsilon \cdot c>0$$
 for all $s \in \mathbb{S}^{d-1}_+$. The set $\{H \diamond s, s \in \mathrm{supp}(S)\}$ does therefore not include $\mathbf{0}_d$ and is compact (since $\mathrm{supp}(S)$ is compact and $s \mapsto H \diamond s$ is continuous). Thus, by~\cite{Butkovic_geb} Proposition 25, the set $\vee\mathrm{-span}(\{H \diamond s, s \in \mathrm{supp}(S)\})$ is closed and the right hand sides of \eqref{eq:spanrep} and \eqref{eq:spanrep:nocl} coincide, which finishes the proof. 
\end{proof}

\subsection{Proofs of supplementary results from Section~\ref{subsec:supp_results_statistics}}
 \begin{proof}[Proof of Lemma~\ref{lem:compactset}]
  We split the proof into two parts. First, we derive in \eqref{eq:constrep} the value of $\kappa$ such that for any global minimizer $(B,W)$, the largest entry of $H:=B \diamond W$ is bounded by $\kappa$. In the second step we show that for any such matrix $H$, there exist rescaled matrices $(\tilde B, \tilde W)$ in the set \eqref{eq:compset} such that $H=\tilde{B} \diamond \tilde{W}$. Together with the continuity of the map~\eqref{eq:lemma_tf}, this yields the existence of a global minimizer.  
  
  \emph{1.} 
  Let $(B^*, W^*)$ be a global minimizer of~\eqref{eq:lemma_tf} and set $H^* := B^* \diamond W^*$. We start by finding an upper bound on $\tilde \rho(H^* \diamond X, X)$ by choosing 
  \begin{equation*}
    M := \begin{pmatrix}
      \mathrm{id}_{p} \\
      \mathbf 0_{(d-p) \times p}
    \end{pmatrix} 
    \diamond \begin{pmatrix}
      \mathrm{id}_{p} & \mathbf 0_{p \times (d - p)}
    \end{pmatrix}
       = \begin{pmatrix}
      \mathrm{id}_{p} & \mathbf 0_{p \times (d-p)} \\
      \mathbf 0_{(d-p) \times p}  & \mathbf 0_{(d-p) \times (d-p)} 
    \end{pmatrix},
  \end{equation*}
  Now we use Lemma~\ref{metric_representation} and \eqref{eq:scale-par-reps} to calculate
  \begin{align}
  \nonumber  \tilde \rho(M \diamond X, X) 
   \nonumber   & = \int_{\S^{d-1}_+} \sum_{k=1}^d \lvert M_{k} \diamond a - a_k \rvert \, S(da) \\
   \nonumber   & = \int_{\S^{d-1}_+} \sum_{k=1}^p \lvert e_k \diamond a - a_k \rvert + \sum_{k=p+1}^d \lvert \mathbf 0_d \diamond a - a_k \rvert \, S(da) \\
    \nonumber  & = \int_{\S^{d-1}_+} \sum_{k=p+1}^d a_k  \, S(da) \\
   \label{Eq:sigma_upper_bound} & = \sum_{k=p+1}^d \sigma_k. 
  \end{align}
  This means that $\tilde \rho(H^* \diamond X, X) \leq \sum_{k=p+1}^d \sigma_k$. 
  Now let $i, j \in \{1, \ldots, d\}$ be such that $\| H^\ast \|_{\infty}=H_{ij}$, then
  \begin{align}
   \nonumber  \tilde \rho(H^\ast \diamond X, X) 
    \nonumber & = \int_{\S^{d-1}_+} \sum_{k=1}^d \lvert H^\ast_k \diamond a - a_k \rvert \, S(da) \\
    \nonumber & \geq \int_{\S^{d-1}_+} \lvert H^\ast_i \diamond a - a_i \rvert \, S(da) \\
    \nonumber & \geq \int_{\S^{d-1}_+}  H^\ast_i \diamond a - a_i \, S(da) \\
    \nonumber & \geq \int_{\S^{d-1}_+}  H^\ast_{ij} a_j - a_i  \, S(da) \\
   \label{Eq:optimal_sigma_bound} & = \| H^\ast \|_{\infty} \sigma_j - \sigma_i 
  \end{align}
  where in the last equality we used again \eqref{eq:scale-par-reps}. As we assumed $H^* = B^* \diamond W^*$ for a global minimizer $(B^\ast, W^\ast)$, we get from \eqref{Eq:sigma_upper_bound} and \eqref{Eq:optimal_sigma_bound} that
  \begin{equation*}
\| H^* \|_{\infty} \sigma_j - \sigma_i \leq \tilde \rho(H^* \diamond X, X) \leq \sum_{k=p+1}^d \sigma_k. 
  \end{equation*}
  Since all scale coefficients are positive by assumption, rearranging now gives us the following inequality
  \begin{equation}
   \label{Eq:H_opt_bound} \| H^* \|_{\infty} \leq \frac 1 {\sigma_j} \left( \sum_{k=p+1}^d \sigma_k + \sigma_i \right).
  \end{equation}
  We can now find an upper bound independent of the indices $i$ and $j$ by defining $\sigma_{\mathrm{min}} := \bigwedge_{i=1}^d \sigma_i > 0$ and $\sigma_{\mathrm{max}} := \bigvee_{i=1}^d \sigma_i$ and the fact that 
  \begin{equation}
   \label{Eq:min_max_sigma_bound}   \frac 1 {\sigma_j} \left( \sum_{k=p+1}^d \sigma_k + \sigma_i \right) \leq \frac 1 {\sigma_{\mathrm{min}}} \left( \sum_{k=p+1}^d \sigma_k + \sigma_{\mathrm{max}} \right). 
  \end{equation}Setting 
  \begin{equation}
    \label{eq:constrep}
    \kappa := \frac 1 {\sigma_{\mathrm{min}}} \left( \sum_{k=p+1}^d \sigma_k + \sigma_{\mathrm{max}} \right) < \infty, 
  \end{equation}
  and using \eqref{Eq:H_opt_bound} and \eqref{Eq:min_max_sigma_bound}, we see that
  \begin{equation}
    \label{eq:hbound}
    H^* \in [0, \kappa]^{d \times d}. 
  \end{equation}
  for all optimal $H^\ast=B^\ast \diamond W^\ast$. 
  
  \emph{2.} 
  For any matrix pair $(B, W)\in [0, \infty)^{d \times p} \times [0, \infty)^{p \times d}$ and $H := B \diamond W$, such that $\| H \|_{\infty} \leq \kappa$, 
  we define a diagonal matrix $D^{(1)} \in [0, \infty)^{p \times p}$ given by the max-norms of the rows of $W$ as follows
  \begin{equation*}
    D^{(1)}_{ij} := \begin{cases}
      \| W_{i} \|_{\infty}, & \text{if } \, i = j, \\
      0, & \text{else}. 
    \end{cases}
  \end{equation*}
  From this matrix $D^{(1)}$, 
  we define another diagonal matrix $D^{(2)} \in [0, \infty)^{p \times p}$ by
  \begin{equation*}
    D^{(2)}_{ij} := \begin{cases}
      (D_{ij}^{(1)})^{-1}, & \text{if} \,i = j, D^{(1)}_{ij} > 0, \\
      0, & \text{else}.
    \end{cases}
  \end{equation*}
  Then by construction of the matrices $D^{(1)}$ and $D^{(2)}$, it holds that
  \begin{equation}\label{eq:proofequality}
    H = B \diamond W = B \diamond D^{(1)} \diamond D^{(2)} \diamond W. 
  \end{equation}
  Define now $\tilde W := D^{(2)} \diamond W$ and $\tilde B := B \diamond D^{(1)}$ and note that $\tilde{W} \in [0,1]^{p \times d}$ and 
  \begin{align*}
    \| \tilde B \|_{\infty} 
    & = \| B \diamond D^{(1)} \|_{\infty} \\
    & = \bigvee_{k=1}^d \bigvee_{j=1}^p B_{kj} \| W_j \|_{\infty} \\
    &  = \bigvee_{k=1}^d \bigvee_{j=1}^p \bigvee_{l=1}^d B_{kj} W_{jl} \\
    & = \| B \diamond W \|_{\infty} \\
    & \leq \kappa,
  \end{align*}
  such that $\tilde B \in [0, \kappa]^{d \times p}$. Thus, there exists a global minimizer $(B^*, W^*)$ in $[0, \kappa]^{d \times p} \times [0,1]^{p \times d}$, because the set is compact and the map~\eqref{eq:lemma_tf} is continuous. 
\end{proof}
\begin{proof}[Proof of Corollary~\ref{cor:cons}]
  Clearly, the given set is compact, and the cost function $c$ 
  is a continuous function, since it is the composition of continuous functions, so the result follows from Theorem~\ref{thm:consistency}.
\end{proof}

\section{Additional material for the simulation studies and data analysis}
\subsection{Implementation}\label{subsec:implementation}

We provide a ready to use implementation of max-stable PCA for the \emph{R} programming language at~\url{www.github.com/FelixRb96/maxstablePCA}, where we find an approximate local minimizer of the reconstruction error using non-linear optimization techniques, since we are not able to analytically compute solutions except for very simple cases. For ease of implementation, we provide the $\| \cdot \|_1$ and $\| \cdot \|_{\infty}$ norms in our $R$-package, but for better comparability of results, we restrict ourselves to only the $\| \cdot \|_{\infty}$ norm here. We now discuss some properties related to the optimization of the empirical reconstruction error. As a function of \( (B, W) \), the empirical reconstruction error is generally not convex; however, it is a Lipschitz continuous function that is piecewise linear. Furthermore, the following minimization problem yields an upper bound on the minimizer of the empirical reconstruction error 
\begin{align*}
  & \min_{B_{kj}, k = p+1, \ldots, d, j = 1, \ldots, p} \sum_{k=p+1}^d \int_{\S^{d-1}_+} a_k - \bigvee_{j=1}^p B_{kj} a_j \, \hat S_n(da), \\
  \quad \mathrm{s.t.}\quad  &  \max_{k = p+1, \ldots, d,j = 1, \ldots, p, a \in \mathrm{supp}(\hat{S}_n)} B_{kj} a_j \leq a_k,  B_{kj} \geq 0.
\end{align*}
In the above problem, it can be shown that both the objective function and the constraint function are convex. Therefore, this is a convex optimization problem that provides an upper bound on the optimal empirical reconstruction error.

Even though the function we minimize is not convex, the empirical results presented in the simulation study and data analysis yield good outcomes and encourage the use of our non-linear optimization techniques. We chose to use sequential least squares quadratic programming~\cite{Kraft_slsqp, NLopt} as our optimization algorithm, despite the reconstruction loss not being differentiable everywhere, because it performed best in simulations compared to other out-of-the-box optimizers available in \emph{R}. It converges reasonably fast, and ready-to-use implementations exist. 

There are also algorithms with theoretical convergence guarantees to a local minimizer for our problem that could be used instead, see for example~\cite{Burke_gsa} or~\cite{Bagirov_nso}. However, to our knowledge, these alternatives lack fast implementations in \emph{R} and depend heavily on good starting values for the problem at hand.

To ensure reasonable results with sequential least squares programming, we uniformly sample a predefined number of points $(B, W)$ from the set $[0,1]^{d \times p} \times [0,1]^{p \times d}$ and then transform the columns of $B$ such that $\| B_{\cdot j} \|_{\infty} = 1$ for all columns $j = 1, \ldots, p$. We initialize the optimization algorithm with the sampled value that achieves the smallest empirical reconstruction error. Since the optimization problem is not convex, it is recommended to inspect the result and possibly rerun this procedure multiple times to achieve a good fit. We found that rerunning the simulation studies below with different random seeds will usually lead to the same results up to simultaneous permutations of rows and columns in the result $(\hat B,\hat W)$, if enough candidates for starting values are tested. 

\subsection{Estimated matrices of simulation studies in Section~\ref{sec:experiments}}\label{Sec:estimated_matrices}

We provide below the estimated matrices of the simulation studies from Sections~\ref{Subsubsec:maxlinear}, \ref{Subsubsec:logistic} and \ref{Subsubsec:latentlogistic}, rounded to two digits for ease of presentation. 
First, for Section~\ref{Subsubsec:maxlinear}, the estimated matrices for the max-linear models of $X^{(1)}$ and $X^{(2)}$, respectively, from Example~\ref{ex:mlmexample} are given by 
{\begin{align}
\label{eq:app_mlmest1}
    \hat B_{mlm}^{(1)} & = \begin{pmatrix} 
    1.00 & 0.23 & 0.03 \\
    0.01 & 1.00 & 0.13 \\
    0.03 & 0.06 & 1.00 \\
    0.61 & 0.18 & 0.57 \\
    0.58 & 0.31 & 0.53
    \end{pmatrix},
    && \hat W_{mlm}^{(1)} = \begin{pmatrix} 
    1.00 & 0.26 & 0.10 & 0.12 & 0.11 \\
    0.06 & 1.00 & 0.12 & 0.09 & 0.02 \\
    0.07 & 0.10 & 1.00 & 0 & 0.22
    \end{pmatrix} \end{align}
    \begin{align}\label{eq:app_mlmest2}
    \hat B_{mlm}^{(2)} & = \begin{pmatrix}
    0.54 & 0.69 & 0 \\
    0 & 1.00 & 0 \\
    0.24 & 0.09 & 1.08 \\
    1.00 & 0 & 0 \\
    0.35 & 0.16 & 1.00
    \end{pmatrix}, 
    && \hat W_{mlm}^{(2)} = \begin{pmatrix}
    0 & 0 & 0.08 & 1.00 & 0 \\
    0 & 1.00 & 0 & 0 & 0 \\
    0.03 & 0.03 & 0.17 & 0.43 & 1.00
    \end{pmatrix}.
\end{align}
For the logistic models from Section~\ref{Subsubsec:logistic} with $\beta_1=0.2, \beta_2=0.5, \beta_3=0.8$, the estimated matrices for $p = 3$ are given by 
\begin{align}
\label{eq:app_log1}
    \hat B_{log}^{(1)} & = \begin{pmatrix}
    1.00 & 0.29 & 0.20 \\
    0.16 & 0.15 & 1.00 \\
    0.29 & 0.91 & 0.66 \\
    0.75 & 0.71 & 0.81 \\
    0.16 & 1.00 & 0.32
    \end{pmatrix}, 
    && \hat W_{log}^{(1)} = \begin{pmatrix}
    1.00 & 0.17 & 0.22 & 0.29 & 0.18 \\
    0.16 & 0.26 & 0.50 & 0.13 & 1.00 \\
    0.12 & 1.00 & 0.13 & 0.23 & 0.19
    \end{pmatrix} \\
    \label{eq:app_log2}
    \hat B_{log}^{(2)} & = \begin{pmatrix}
    0.00 & 0 & 1.00 \\
    1.00 & 0 & 0.01 \\
    0.40 & 0.14 & 0.55 \\
    0.48 & 0.26 & 0.10 \\
    0.01 & 1.00 & 0
    \end{pmatrix}, 
    && \hat W_{log}^{(2)} = \begin{pmatrix}
    0 & 1.00 & 0.02 & 0 & 0.01 \\
    0.01 & 0.01 & 0.01 & 0.01 & 1.00 \\
    1.00 & 0 & 0.01 & 0.01 & 0
    \end{pmatrix} \\
    \label{eq:app_log3}
        \hat B_{log}^{(3)} & = \begin{pmatrix}
        0.09 & 0.09 & 0.01 \\
        1.00 & 0 & 0 \\
        0.08 & 0.06 & 0.02 \\
        0 & 1.00 & 0 \\
        0 & 0 & 1.00
    \end{pmatrix}, 
    & & \hat W_{log}^{(3)} = \begin{pmatrix}
    0 & 1.00 & 0.01 & 0 & 0 \\
    0 & 0 & 0 & 1.00 & 0 \\
    0 & 0 & 0 & 0 & 1.00
    \end{pmatrix}.
\end{align}
Finally, for the simulation of the generalized max-linear model of Section~\ref{Subsubsec:latentlogistic}, the estimated matrices are given by 
\begin{equation}
\label {eq:app_genmlm}
    \hat B_{gen-mlm} = \begin{pmatrix}
    1.00 & 0.11 \\
    0.00 & 1.00 \\
    0.85 & 0.53 \\
    0.65 & 0.76
    \end{pmatrix}, \quad  
    \hat W_{gen-mlm} = \begin{pmatrix}
    1.00 & 0.17	& 0.14 & 0.11 \\
    0.02 & 1.00	& 0.03 & 0.08
    \end{pmatrix}.
\end{equation}
We summarize our interpretation of the estimated matrices in the following remark. 
\begin{remark}\label{Remark:Estimated_Matrices}
    \begin{enumerate}[(i)]
    \item In all cases, the matrix $W$ has only one dominant entry per row, and no dominant entries share the same column. This means that the encoded state $W \diamond X$ tends to select $p$ distinct components of the original random vector $X$ as the encoding in our examples, whether $X$ is perfectly reconstructable or not. 
    \item Consequently, $W \diamond X$ has almost the same distribution as $(X_{k_1}, \ldots, X_{k_p})$, for some indices $k_1, \ldots, k_p \in \{1, \ldots, d\}$ all different. Accordingly, the rows $k_1, \ldots, k_p$ of the matrix $B$ have one dominant entry such that the corresponding marginal again appears in the right place. 
    \item If the random vector $X$ is perfectly reconstructable, the remaining rows of $B$ are the $\vee$-linear combinations of Theorem~\ref{thm:spectralrep} up to numerical and statistical approximation error and scaling. 
    \item If $X$ is not perfectly reconstructable, the remaining rows of $B$ are $\vee$-linear combinations of the encoded states that tend to display higher values when the dependence of the corresponding marginal to the encoded states is strong, and gets close to zero when the dependence is weak.
    \end{enumerate}
\end{remark}
}
{\subsection{Additional plots for the simulation studies}

We provide plots for the simulation studies from Sections~\ref{Subsubsec:maxlinear}, \ref{Subsubsec:logistic} and \ref{Subsubsec:latentlogistic} to visually compare reconstruction against simulated data in bivariate margin plots for the max-linear models in Figure~\ref{fig:mlmplots}, for the logistic models in Figure~\ref{fig:logisticplots} and for the perfectly reconstructable generalized max-linear model in Figure~\ref{fig:genmlm}. If the data are from a model as in Theorem~\ref{thm:spectralrep}, the model finds a perfect reconstruction, and if the model is not perfectly reconstructable, the fit looks much better if the dependence in the data is high. 
\begin{figure}[h!tb]
    \centering
    \begin{subfigure}{.45\textwidth}
      \centering
      \includegraphics[width=\textwidth]{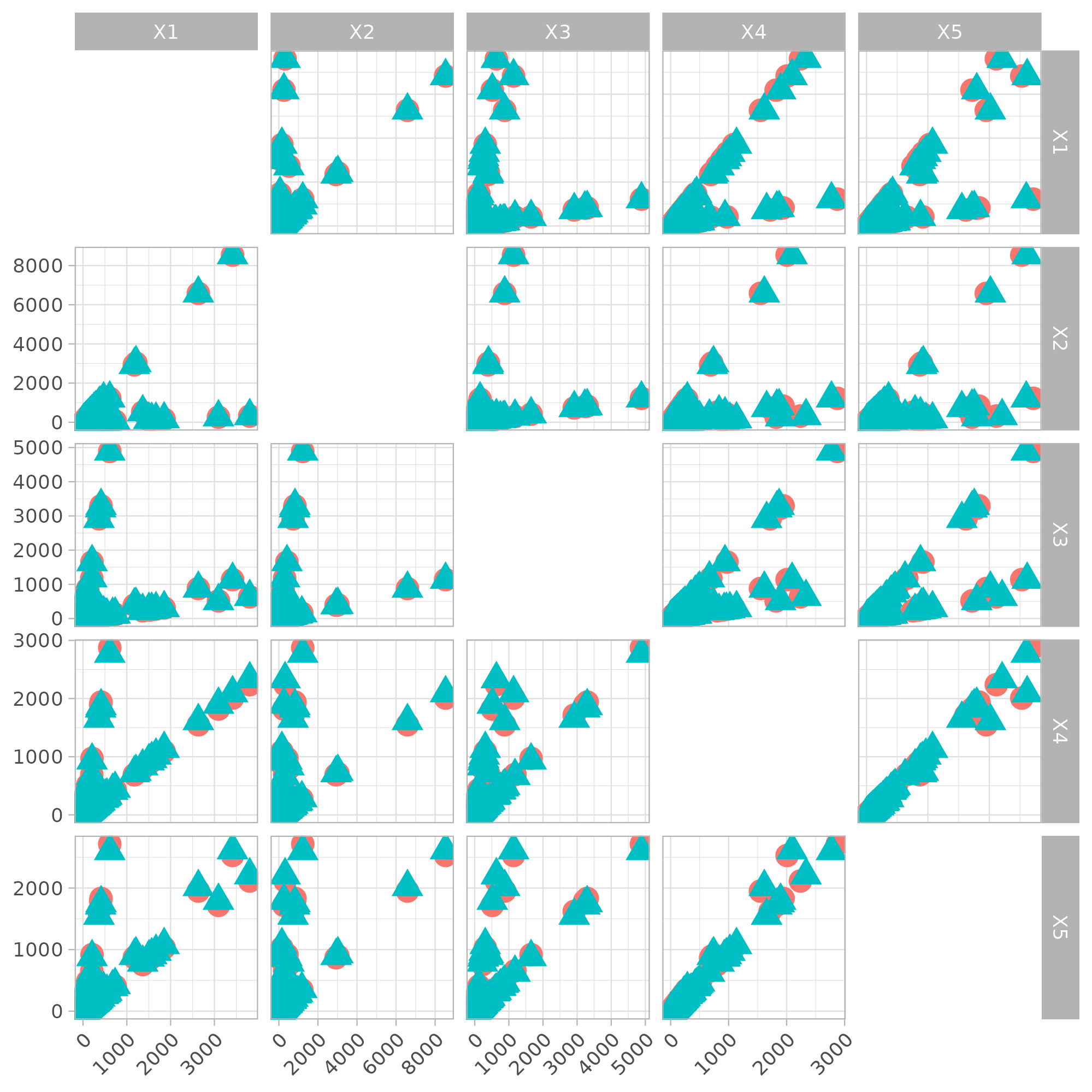}
      \caption{Realization of max-linear model with matrix $A^{(1)}$. Columns are permuted such that the upper $3 \times 3$ submatrix is given by the index of the column of the maximal entry in each row of $\hat W$, as seen in~\eqref{eq:app_mlmest1}.}
    \end{subfigure}
    \hfill
    \begin{subfigure}{.45\textwidth}
      \centering
      \includegraphics[width=\textwidth]{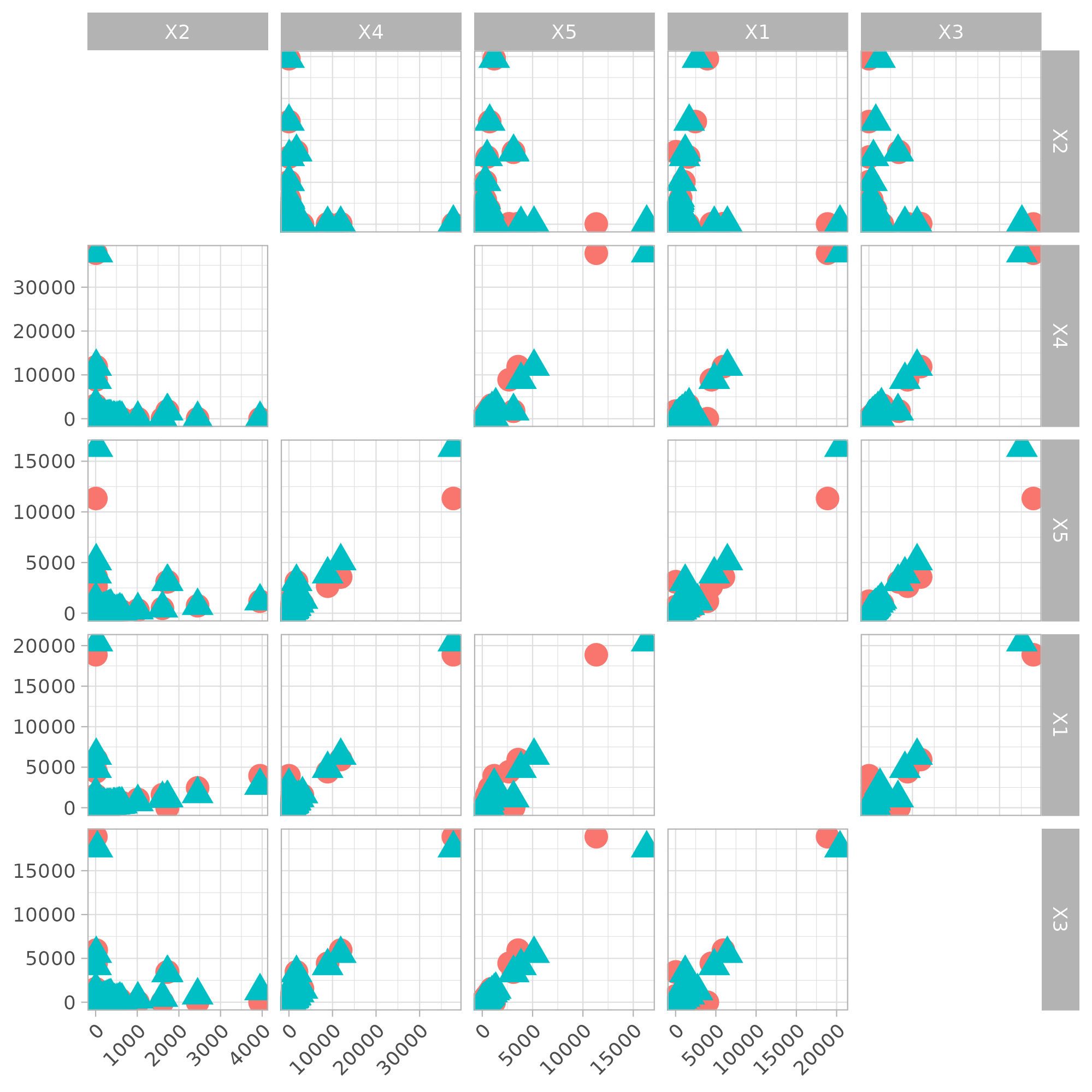}
      \caption{Realization of max-linear model with matrix $A^{(2)}$. Columns are permuted such that the upper $3 \times 3$ submatrix is given by the index of the column of the maximal entry in each row of $\hat W$, as seen in~\eqref{eq:app_mlmest2}.}
    \end{subfigure}
    \caption{Pairplots for the two max-linear models. 
    Sample of 10000 i.i.d. realizations (red) of max-linear model with $1$-Fréchet margins and all scale coefficients equal to one
      given by matrix $A^{(i)}, i = 1,2$. Plotted against the reconstruction (blue) of max-stable PCA with $p = 3$.}\label{fig:mlmplots}
  \end{figure}
\begin{figure}[h!bt]
    \centering
    \begin{subfigure}{0.32\textwidth}
      \centering
      \includegraphics[width=\textwidth]{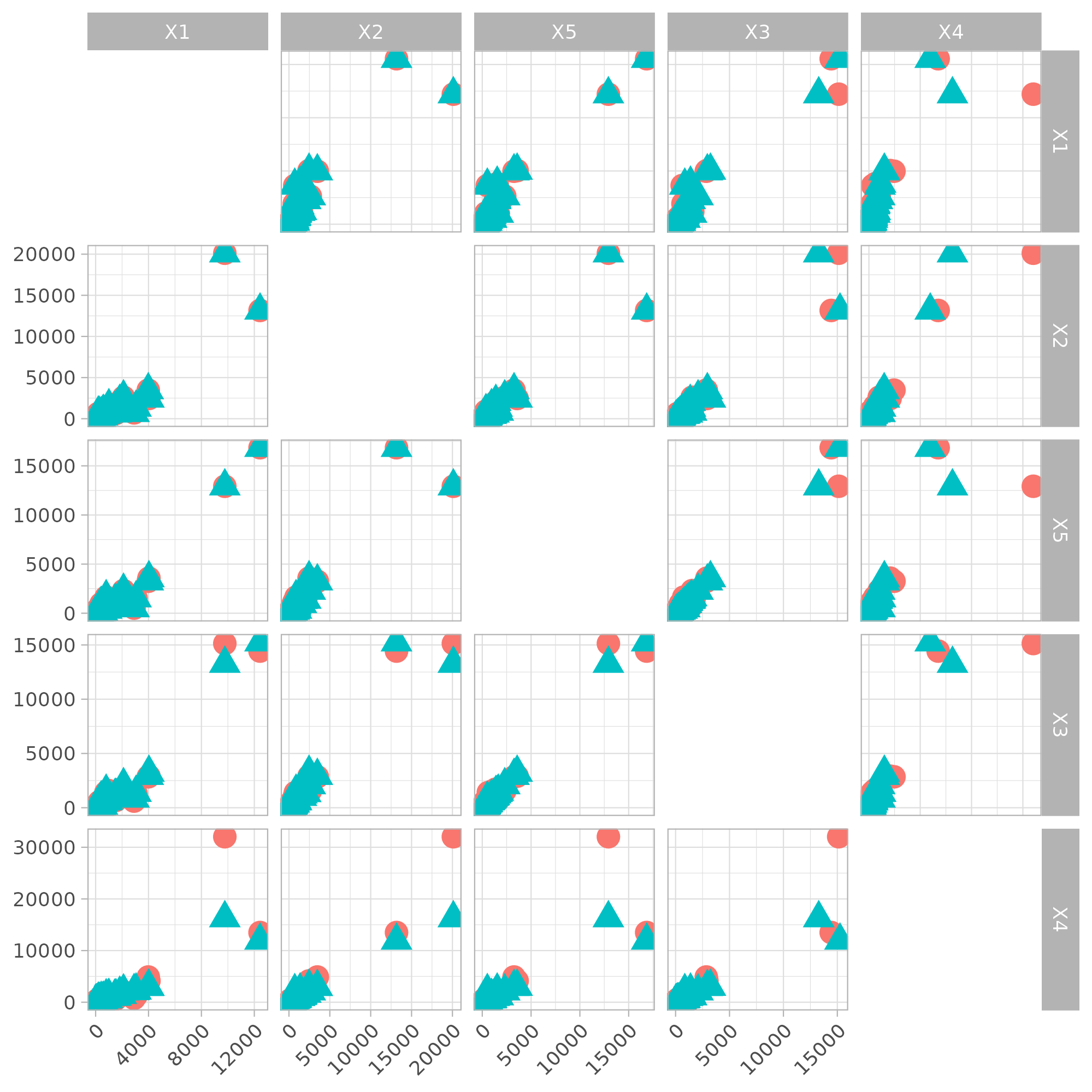}
      \caption{Realization of logistic model with $\beta = 0.2$. Columns are permuted such that the upper $3 \times 3$ submatrix is given by the index of the column of the maximal entry in each row of $\hat W$, as seen in~\eqref{eq:app_log1}.}
    \end{subfigure}
    \hfill
    \begin{subfigure}{0.32\textwidth}
      \centering
      \includegraphics[width=\textwidth]{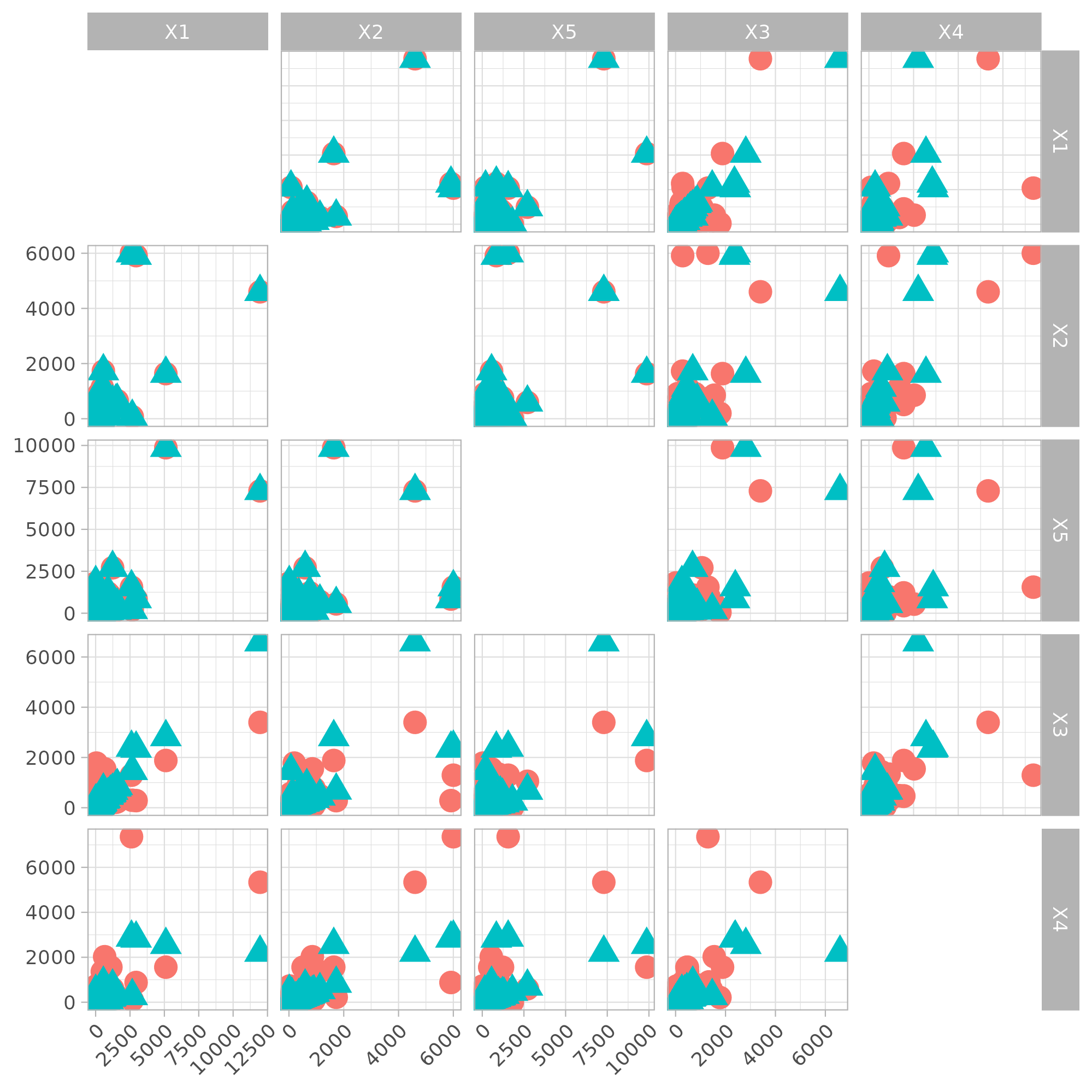}
      \caption{Realization of logistic model with $\beta = 0.5$. Columns are permuted such that the upper $3 \times 3$ submatrix is given by the index of the column of the maximal entry in each row of $\hat W$, as seen in~\eqref{eq:app_log2}.}
    \end{subfigure}
    \hfill
    \begin{subfigure}{0.32\textwidth}
      \centering
      \includegraphics[width=\textwidth]{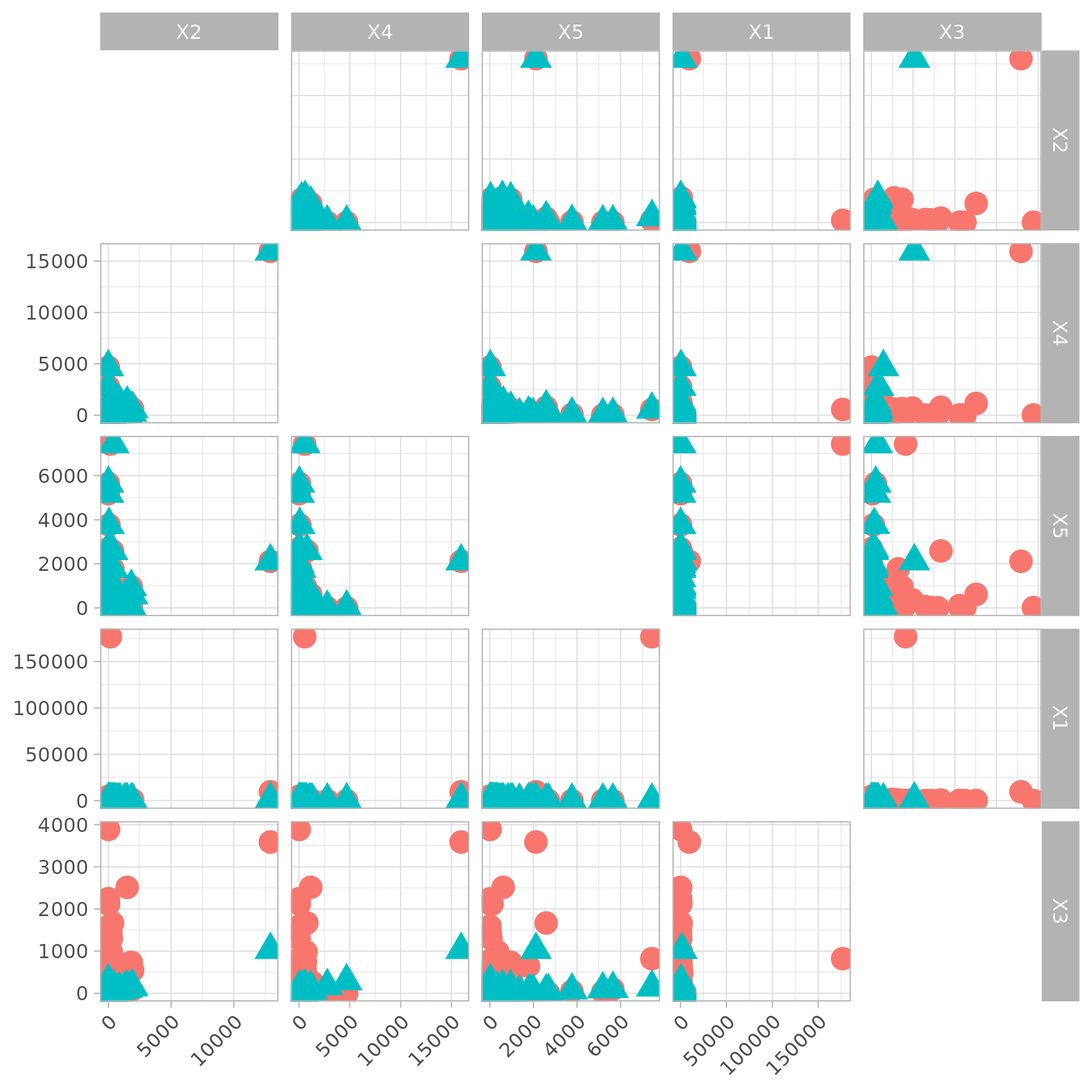}
      \caption{Realization of logistic model with $\beta = 0.8$. Columns are permuted such that the upper $3 \times 3$ submatrix is given by the index of the column of the maximal entry in each row of $\hat W$, as seen in~\eqref{eq:app_log3}.}
    \end{subfigure}
    \caption{Pairplots for the three logistic models. 
    Sample of 10000 i.i.d. realizations (red) of logistic model with $1$-Fréchet margins and all scale coefficients equal to one
      given by different choices of $\beta$. Plotted against the reconstruction (blue) of max-stable PCA with $p = 3$. 
      }\label{fig:logisticplots}
  \end{figure}
  \begin{figure}
      \centering
      \includegraphics[width=0.4\textwidth]{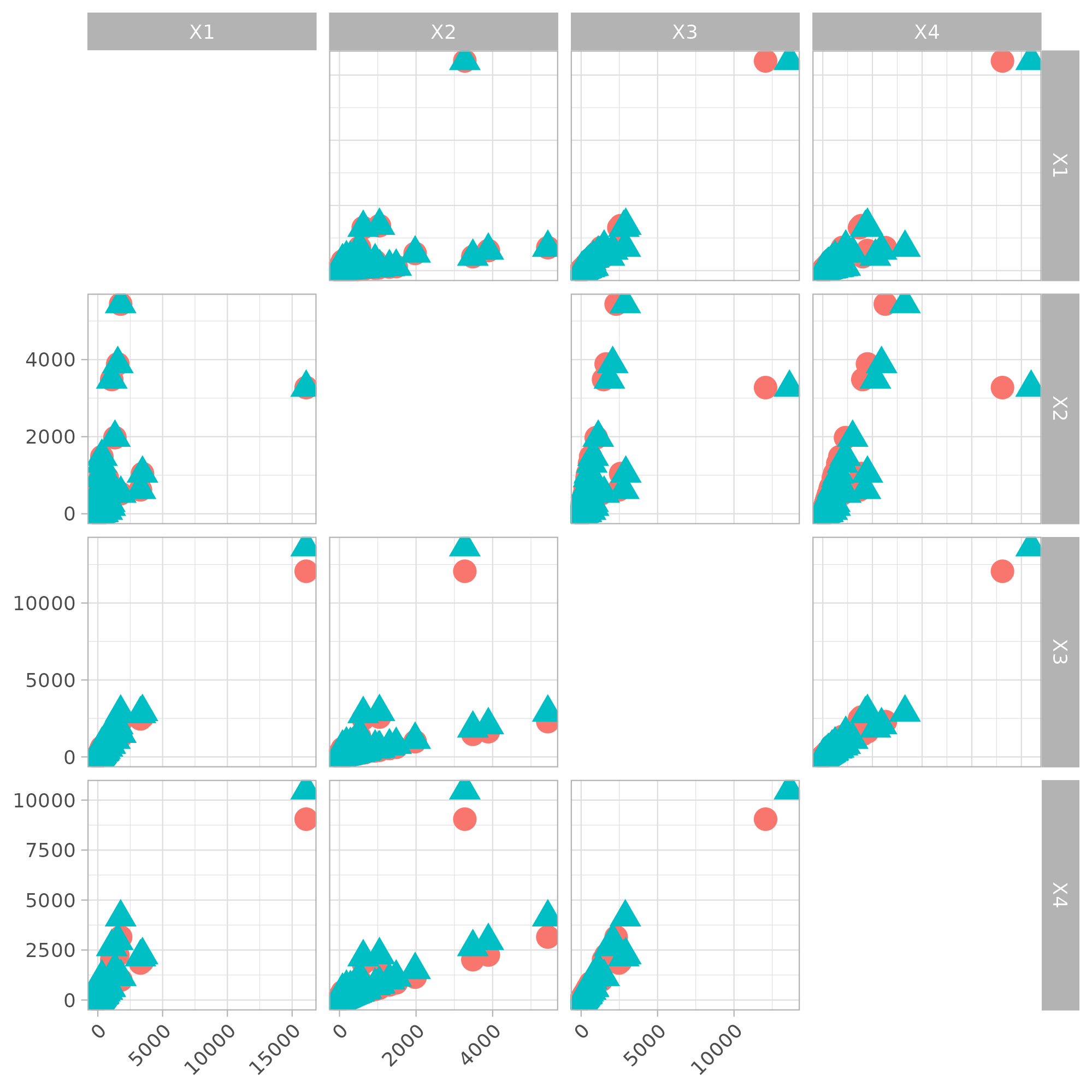}
      \caption{Realization of a sample of the model from Section~\ref{Subsubsec:latentlogistic} (red) against the reconstruction by max-stable PCA for $p = 2$ (blue). Columns are permuted such that the upper $3 \times 3$ submatrix is given by the index of the column of the maximal entry in each row of $\hat W$, as seen in~\eqref{eq:app_genmlm}.}
      \label{fig:genmlm}
  \end{figure}

\subsection{Additional plots for the Danube dataset}

To create an elbow plot for the max-stable PCA, we run the optimization procedure for $p = 1, \ldots 12$ for 50 different starting values and use the fit for each $p$ that achieved the lowest reconstruction error. We report the elbow plot of the reconstruction errors with the best fit in Figure~\ref{fig:elbowplot_danube}. 
We provide a plot of the colored Danube river arms and a bivariate marginplot over the maxima of the river arms to show that max-stable PCA with $p = 6$ indeed captures the key characteristics of each river arm in Figure~\ref{fig:pairplot_danube}.
\begin{figure}[h!tb]
    \centering
    \begin{subfigure}{0.4\textwidth}
      \centering
      \includegraphics[width=\textwidth]{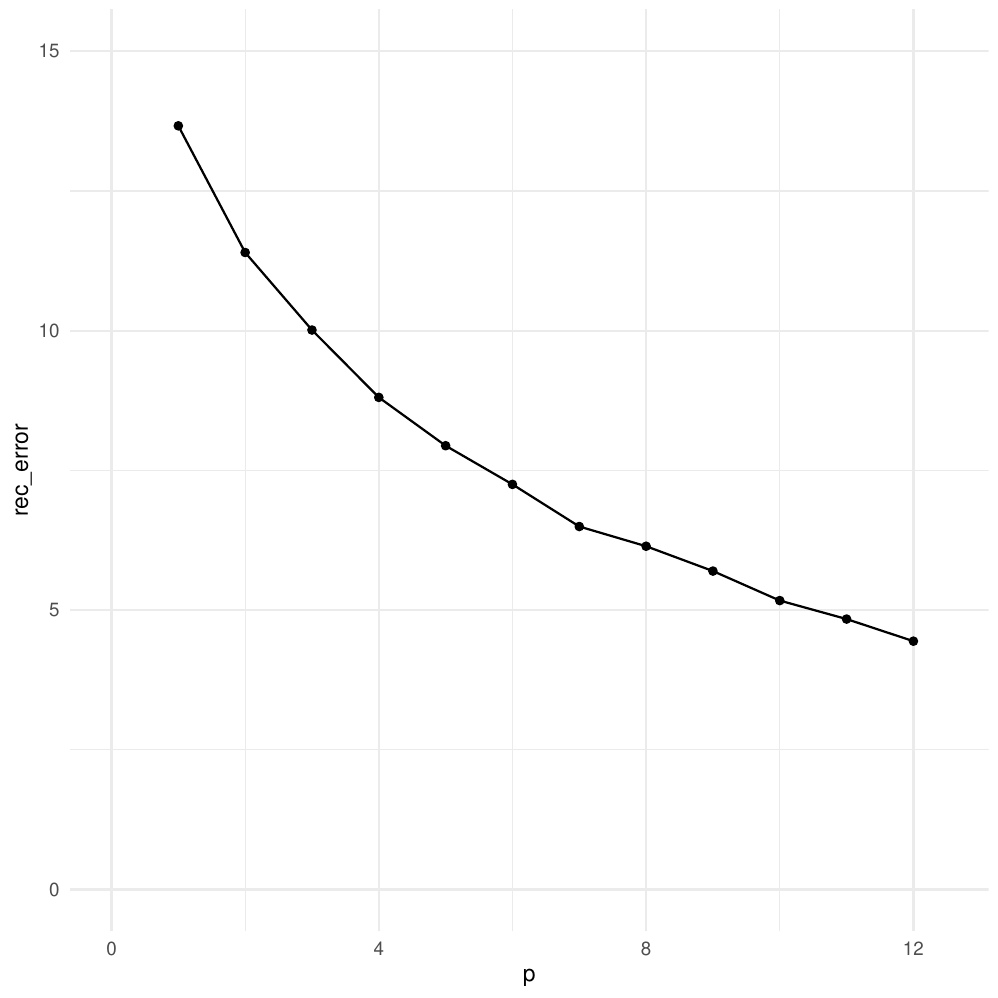}
    \end{subfigure}
    \caption{The reconstruction error values of the max-stable PCA for $p = 1, \ldots, 12$. 
    We used $50$ different starting values to find the best reconstruction. 
    }
    \label{fig:elbowplot_danube}
\end{figure}
\begin{figure}[h!bt]
\centering
\begin{subfigure}{0.45\textwidth}
      \centering
      \includegraphics[width=\textwidth]{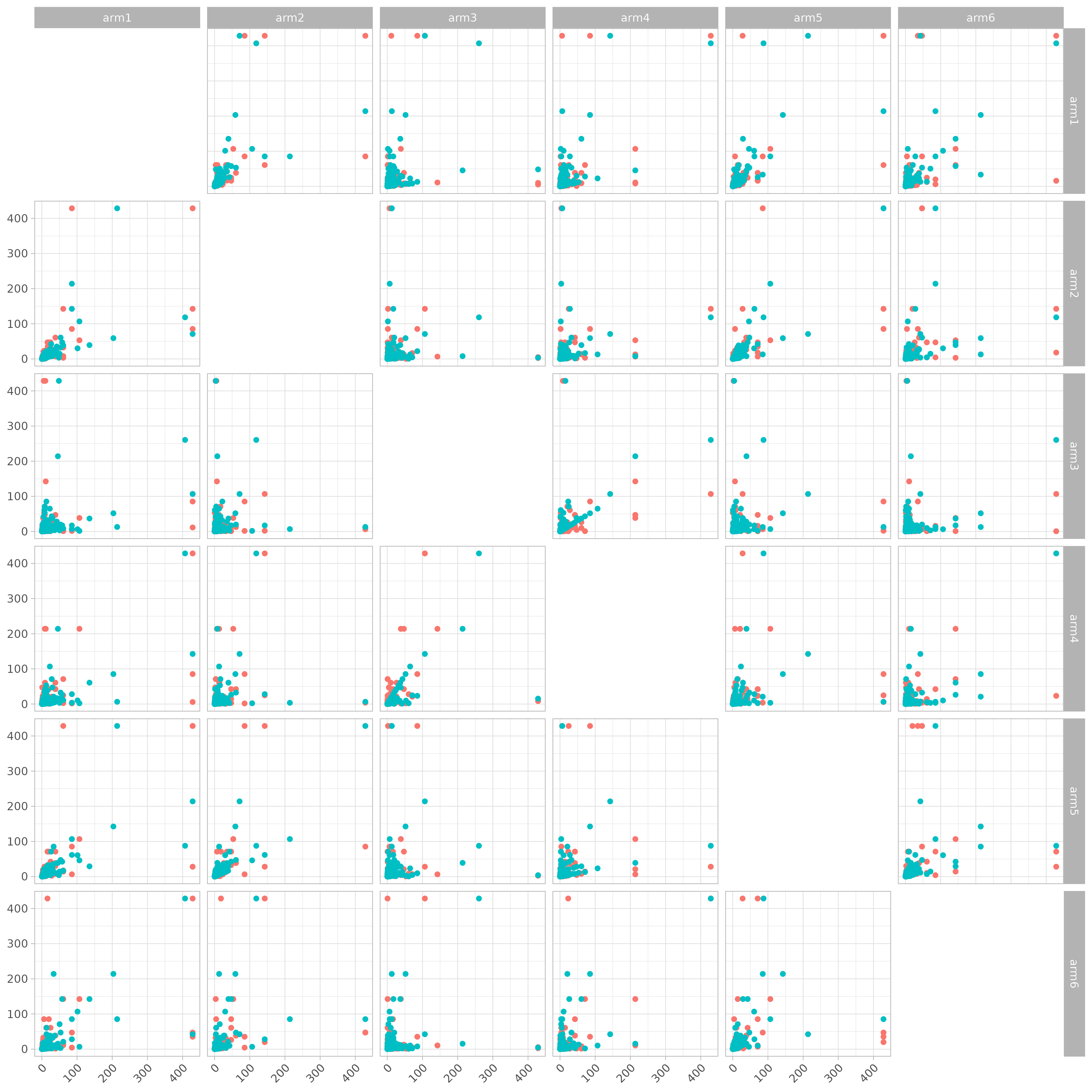}
    \end{subfigure}
    \hfill
    \begin{subfigure}{0.45\textwidth}
      \centering
      \includegraphics[width=\textwidth]{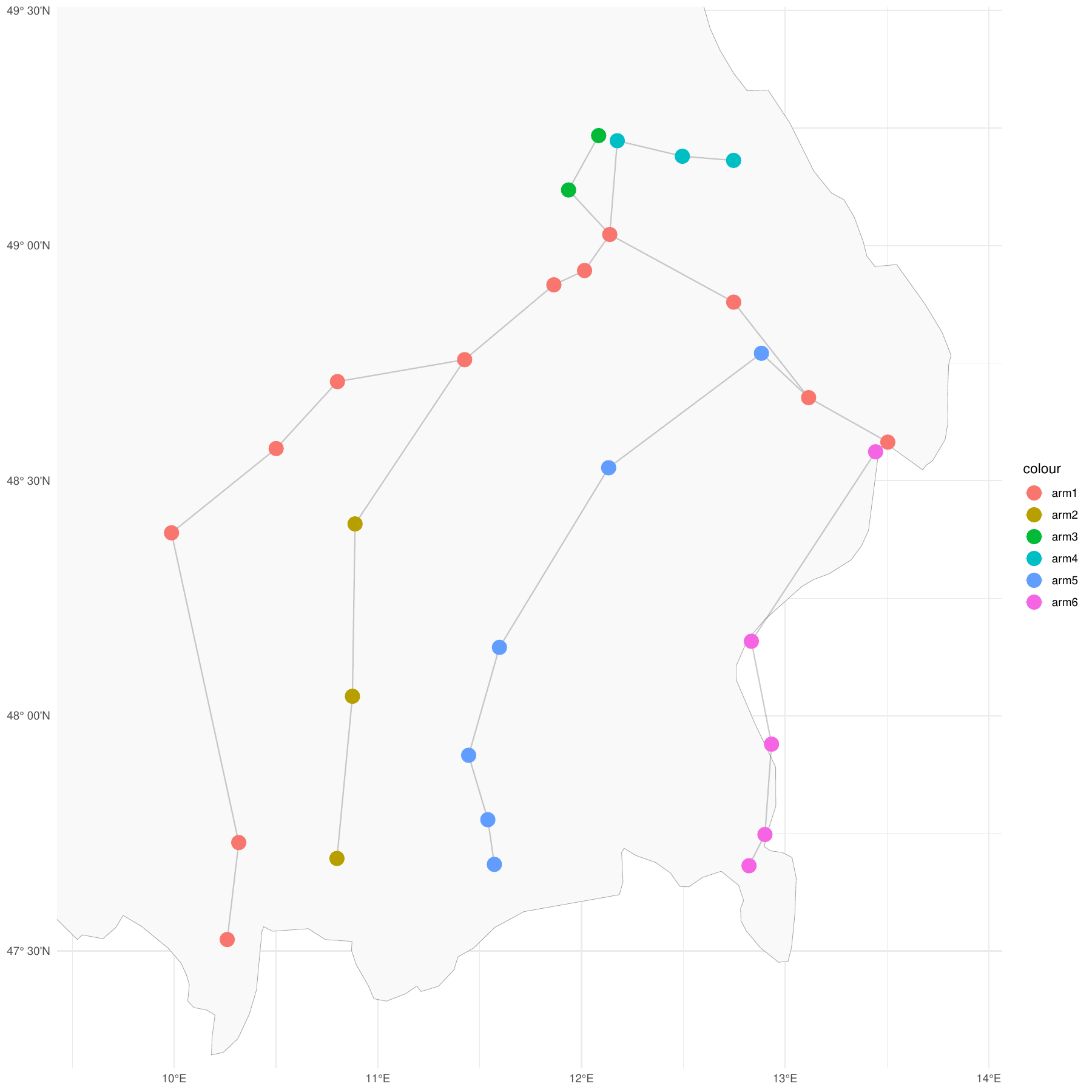}
    \end{subfigure}
    \caption{Left: Data of the maxima over each river arm (red) versus the maxima of each river arm of the reconstruction by max-stable PCA for $p = 6$ on $1$-Fréchet scale. Right: The stations colored by the river arm they belong to. Note that this is only the river network of stations appearing and other inflows are not taken into account.}\label{fig:pairplot_danube}
\end{figure}

{\section{Comparison to other dimension reduction techniques}\label{Suppl:Danube_Comp}

To compare our approach to other dimension reduction approaches for extremes, we applied the methodologies of~\cite{CooleyThibaud_ddhde, DreesSabourin_pcame, JanssenWan_kmce} to the Danube dataset. 

\subsection{PCA for multivariate extremes as proposed by Drees and Sabourin (\cite{DreesSabourin_pcame})}

An approach to PCA for multivariate extremes was proposed in~\cite{DreesSabourin_pcame}, where for a regularly varying random vector $X$ the aim is to identify a lower dimensional subspace on which the spectral measure concentrates. To this end, for data $X_1, \ldots, X_n$, assumed to be i.i.d.\ copies of $X$, only points $
X_i$ that satisfy $\| X_i \| > t$ for a suitably large threshold $t > 0$ are considered. The data is rescaled to the unit sphere by setting $\theta_i := X_i / \| X_i \|$ for the selected exceedances and then a classical PCA is performed with the resulting data matrix $\Theta$. Statistical guarantees for the procedure are provided in the sense of consistency results and convergence rates, for details see~\cite{DreesSabourin_pcame} or~\cite[Section 4.2]{Engelke_ssme} for a short summary. 

Just like our procedure, this approach solves a minimization problem. However, in \cite{DreesSabourin_pcame} the optimal value of the objective function can be expressed in terms of the eigenvalues of the matrix $\Theta \Theta^T$. In order to find a suitable dimension, an elbow plot of the objective function can be helpful. 

We apply the PCA procedure to the Danube dataset to visually inspect the elbow plot provided in Figure~\ref{fig:elbow_danube_ds}. The elbow plot suggests that the spectral measure might be concentrated on a lower dimensional subspace around dimension 5. In order to facilitate comparison to our approach, we choose $p=6$ and plot the vectors of principal components in Figure~\ref{fig:pcs_ds}. It should be noted, however, that the interpretation of dimension in \cite{DreesSabourin_pcame} and our approach differs: In \cite{DreesSabourin_pcame} one searches for a lower dimensional space on which the spectral measure is concentrated, whereas in our approach we look for a minimal number of relevant components which are sufficient to reconstruct the rest. 
\begin{figure}[h!bt]
\centering
\begin{subfigure}{0.4\textwidth}
      \centering
      \includegraphics[width=\textwidth]{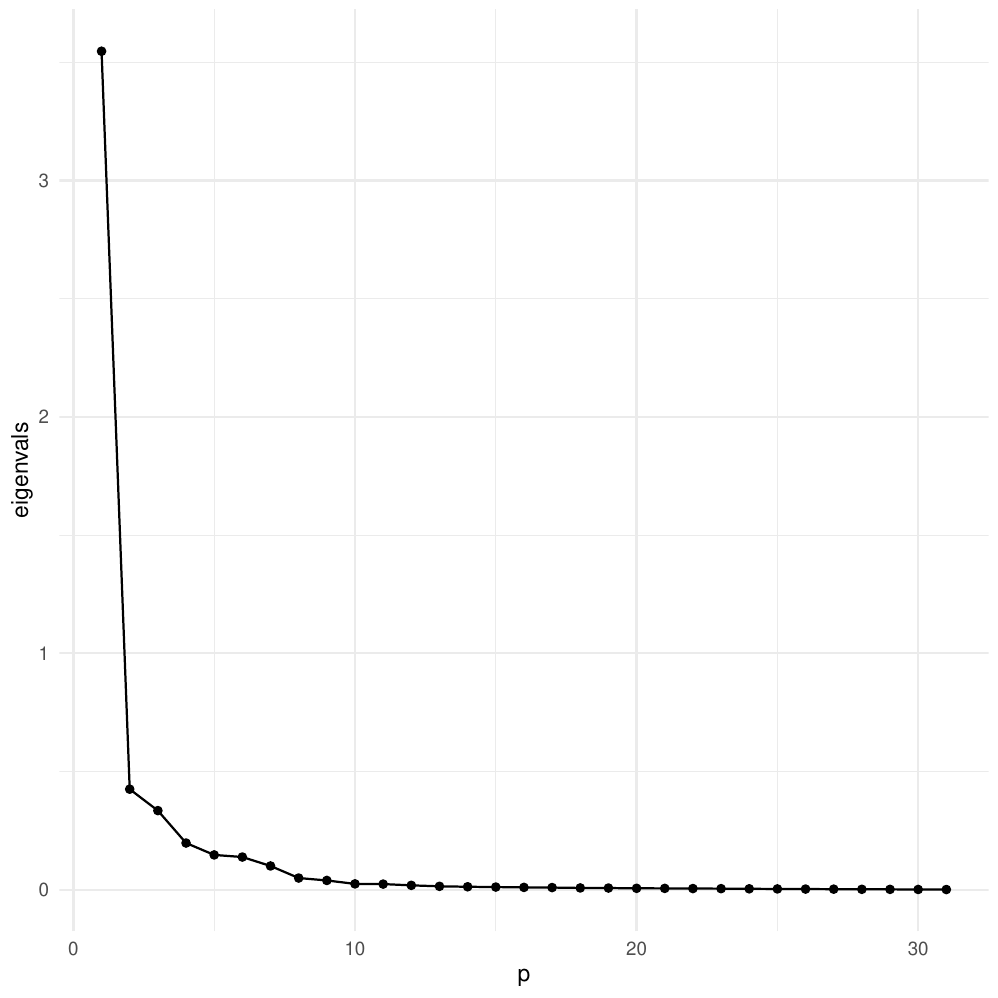}
    \end{subfigure}
    \caption{An elbow plot of the eigenvalues of the covariance matrix of the extreme data rescaled to the sphere from the procedure by Drees and Sabourin (\cite{DreesSabourin_pcame}). Note that arguably around $p = 5,6$ the curve flattens out significantly.}\label{fig:elbow_danube_ds}
\end{figure}
\begin{figure}[h!bt]
\centering
\begin{subfigure}{0.3\textwidth}
      \centering
      \includegraphics[width=\textwidth]{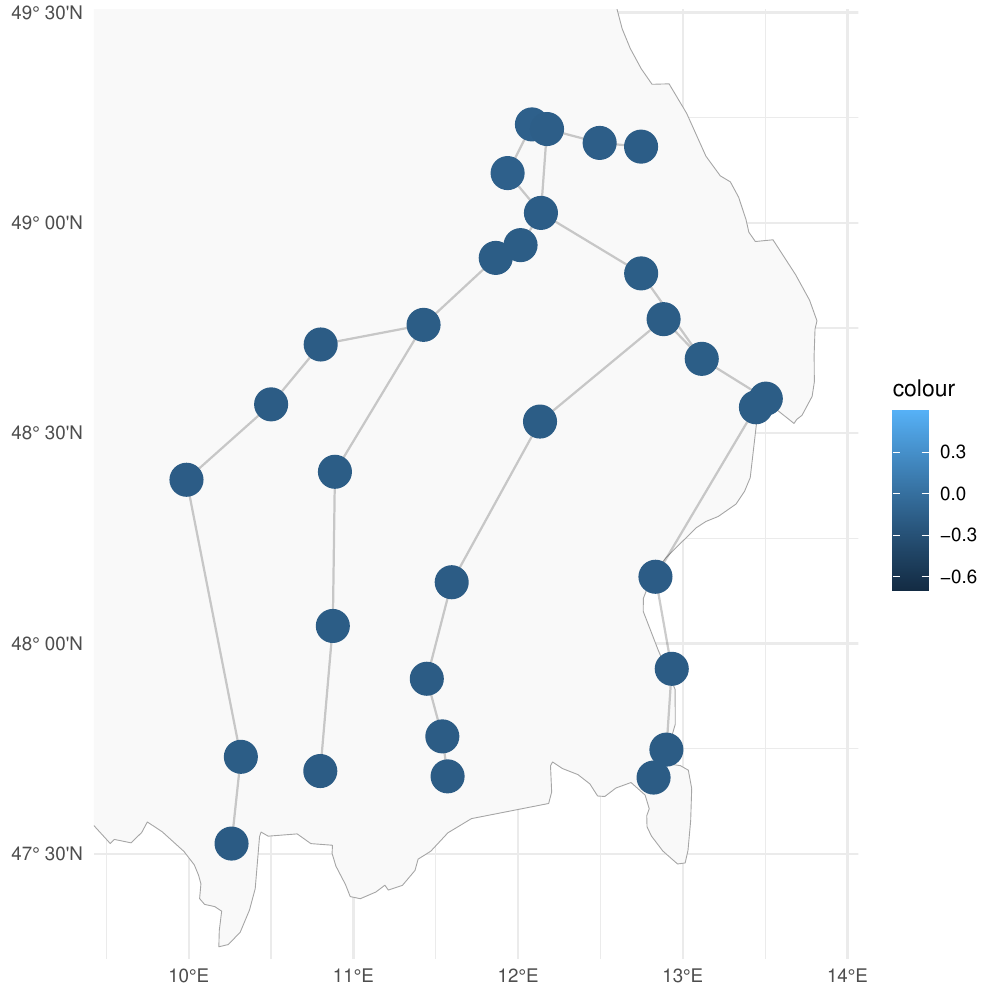}
    \end{subfigure}
    \hfill\begin{subfigure}{0.3\textwidth}
      \centering
      \includegraphics[width=\textwidth]{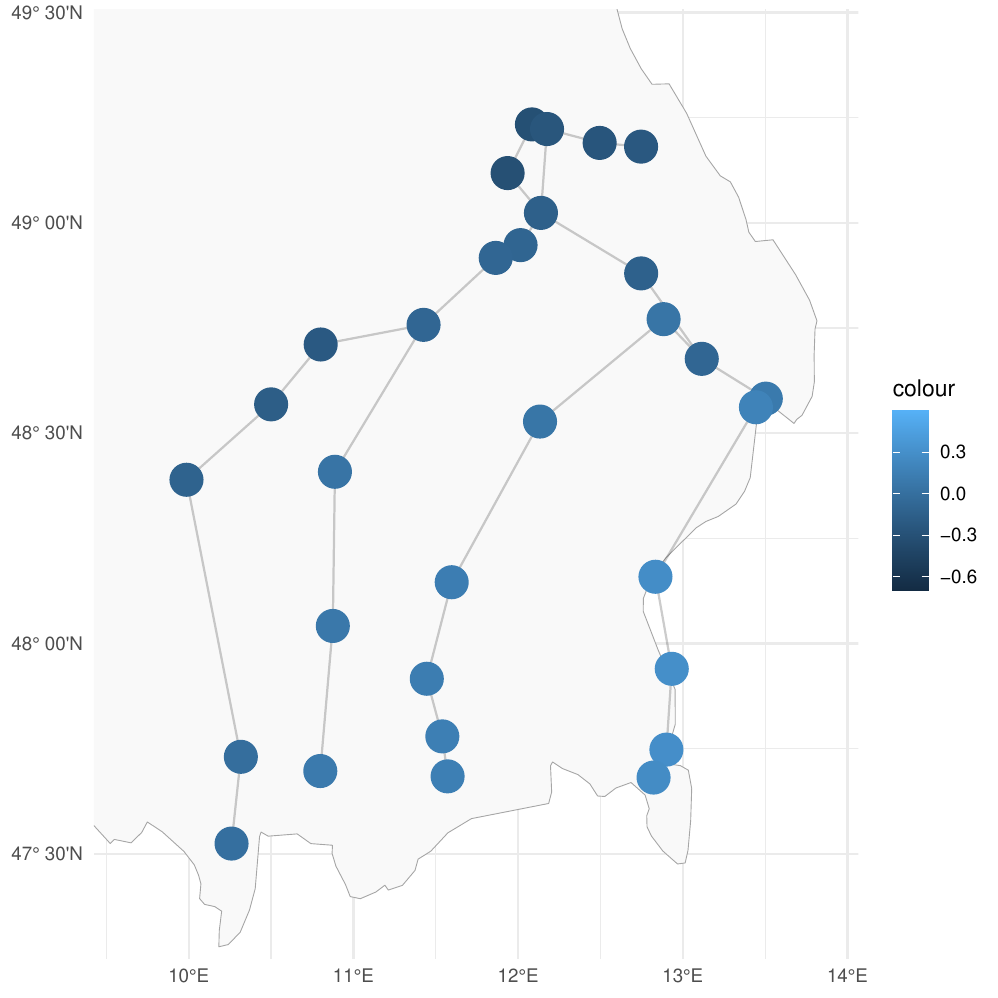}
    \end{subfigure}
    \hfill\begin{subfigure}{0.3\textwidth}
      \centering
      \includegraphics[width=\textwidth]{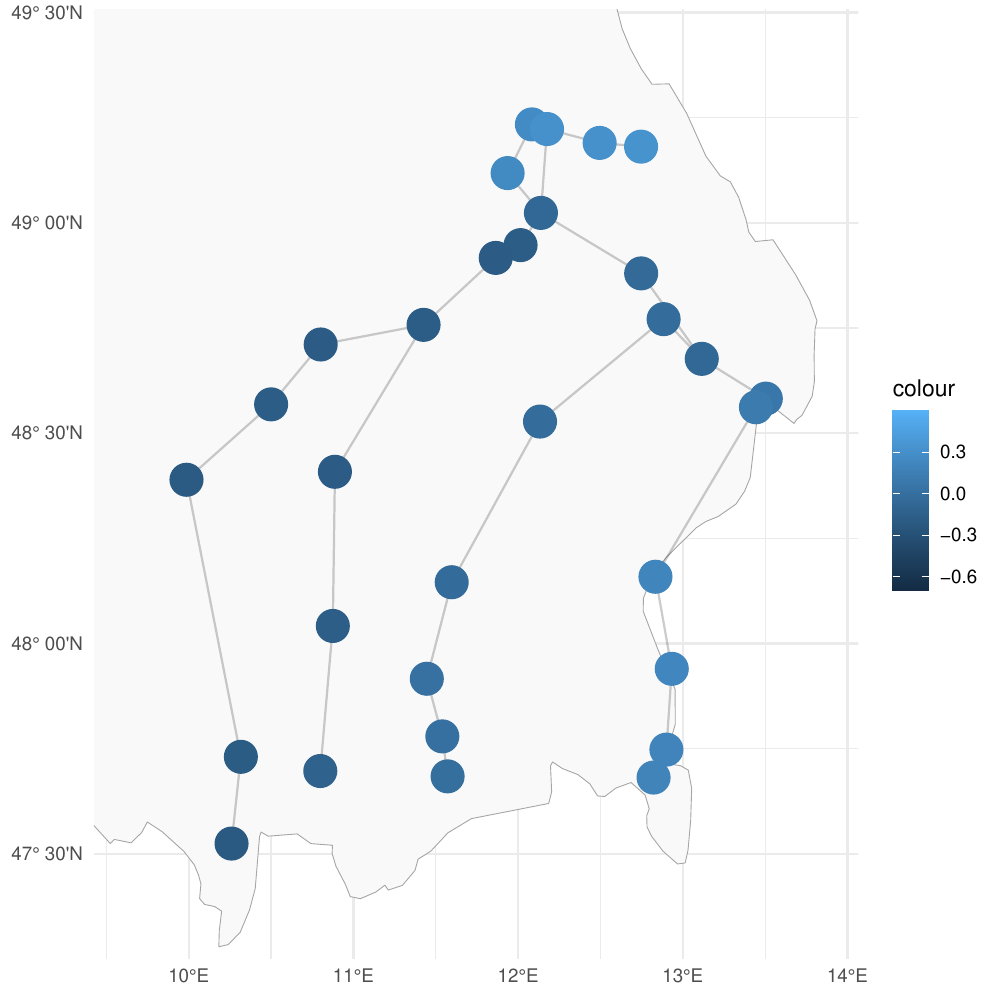}
    \end{subfigure}
    \hfill
    \begin{subfigure}{0.3\textwidth}
      \centering
      \includegraphics[width=\textwidth]{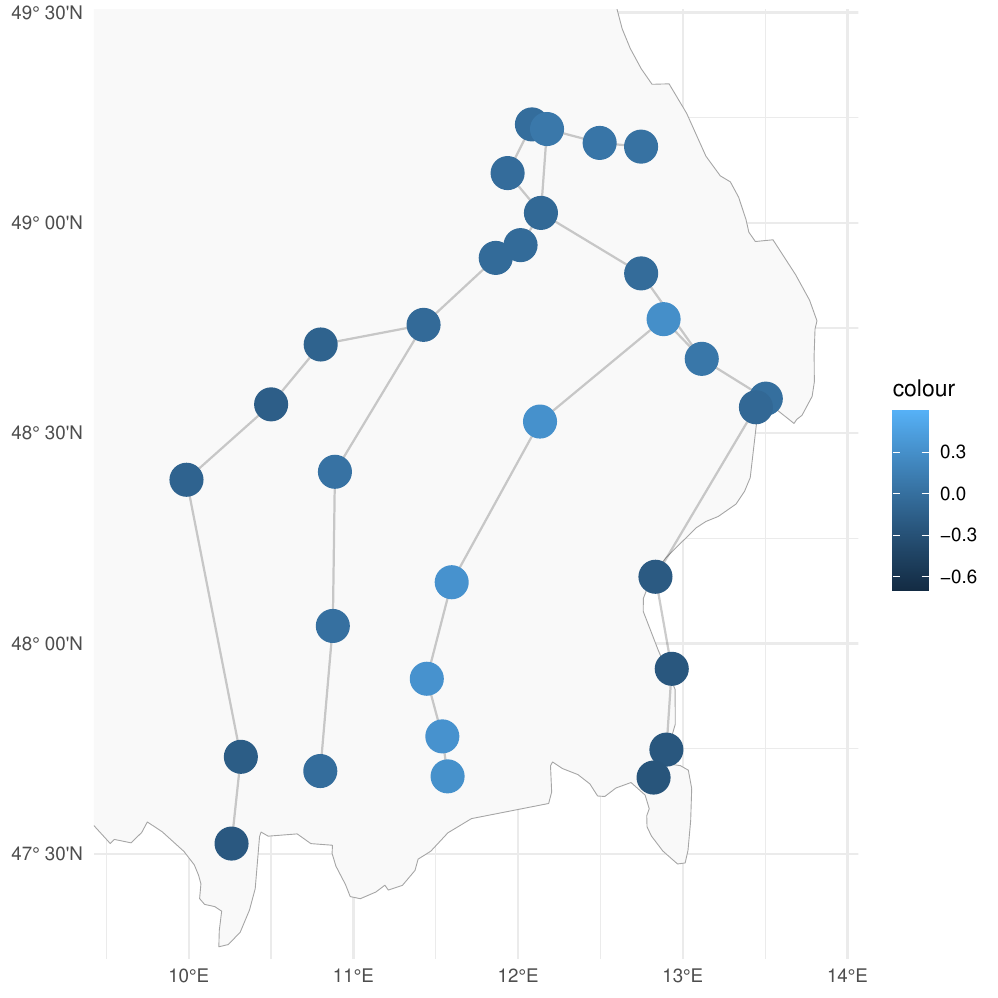}
    \end{subfigure}
    \hfill
    \begin{subfigure}{0.3\textwidth}
      \centering
      \includegraphics[width=\textwidth]{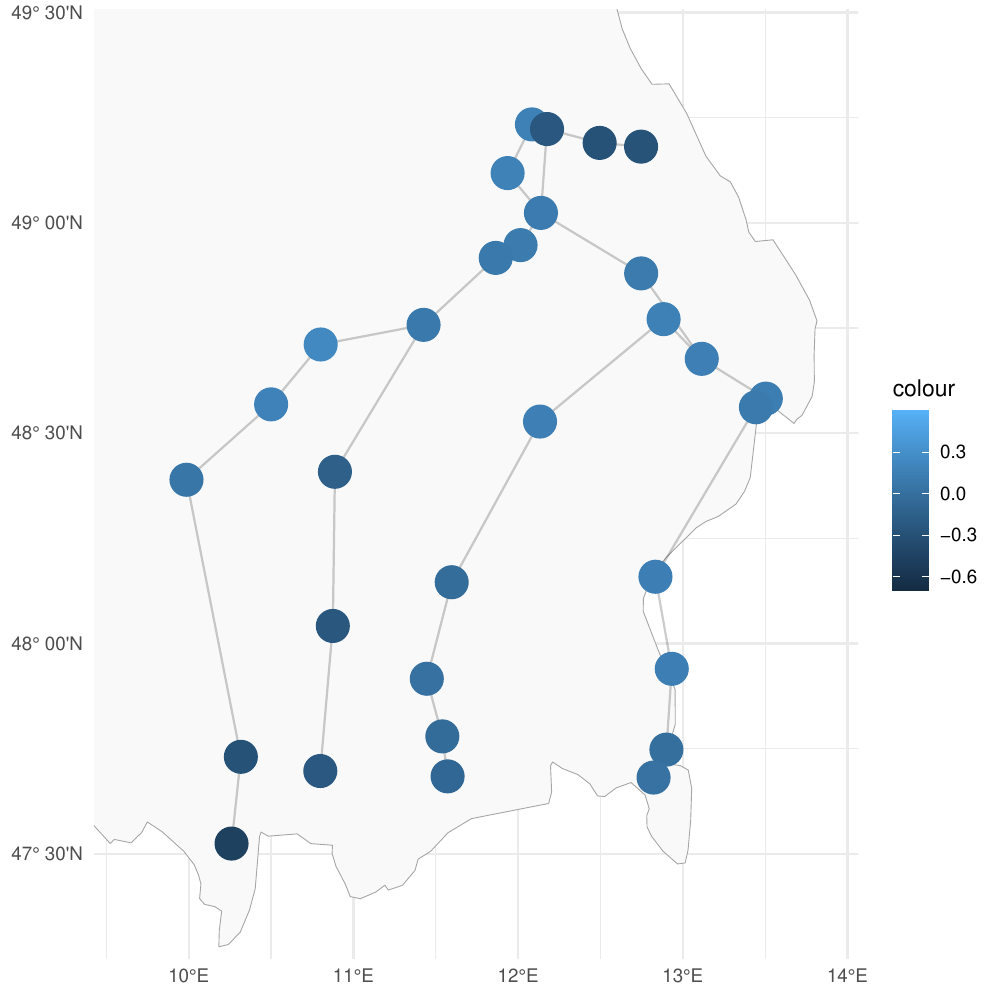}
    \end{subfigure}
    \hfill
    \begin{subfigure}{0.3\textwidth}
      \centering
      \includegraphics[width=\textwidth]{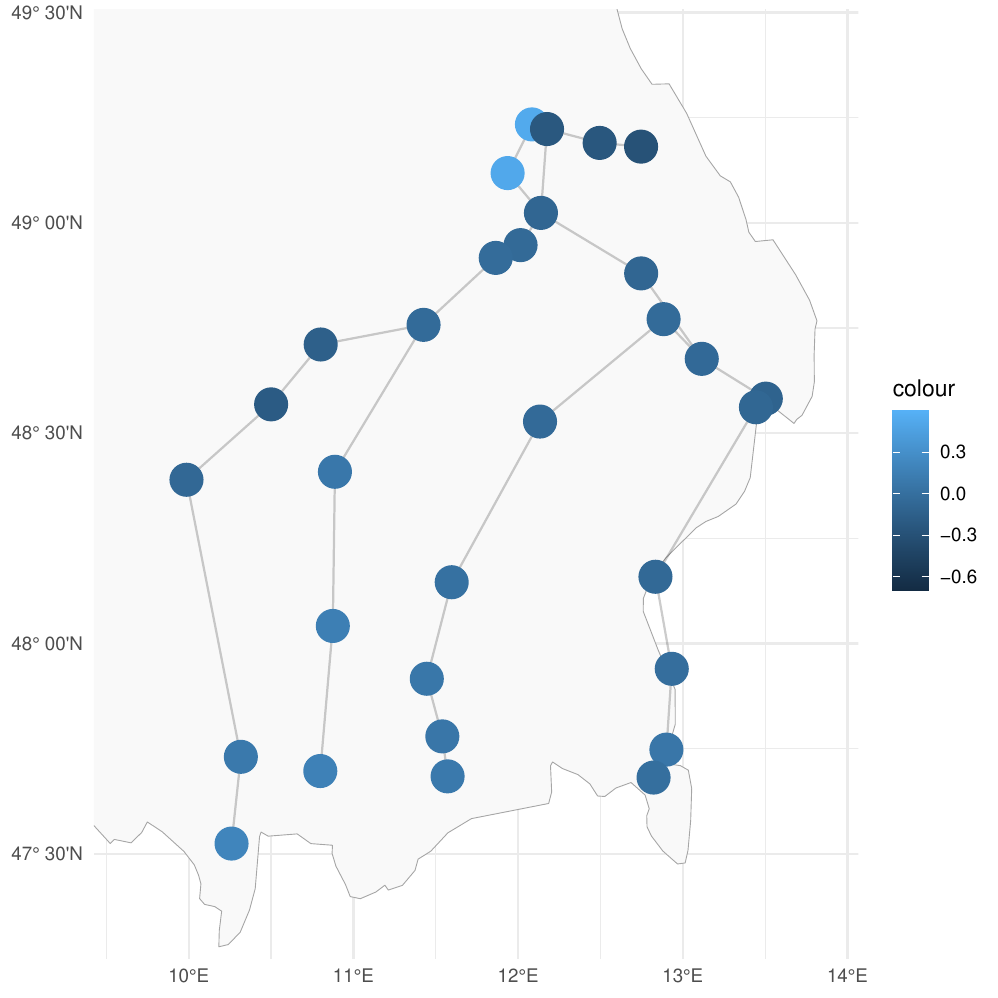}
    \end{subfigure}
    \caption{The first six principal components of the procedure by Drees and Sabourin, first three from left to right in the first row and fourth to sixth component in the second row, with each entry of the vector mapped to the corresponding measurement station and colored from dark blue for negative values to light blue for positive values. Note that the first principal component is almost constant and starting from the second component, it looks like the components concentrate on different river arms.}\label{fig:pcs_ds}
\end{figure}

\subsection{PCA for extremes as proposed by Cooley and Thibaud (\cite{CooleyThibaud_ddhde})}

An approach to PCA for extremes that, like ours, deals with the restricted support of max-stable distributions with Fréchet margins has been proposed in~\cite{CooleyThibaud_ddhde} and applied to data in ~\cite{JiangCooley_pcae}. We briefly summarize the procedure to highlight similarities and differences. Given i.i.d. data $X_1, \ldots, X_n$ from the max-domain of attraction of some multivariate extreme value distribution, transformed to $2$-Fréchet margins, the aim is to find a reconstruction of the data by determining a matrix $U \in \R^{d \times p}$ such that for an invertible map $t$ (for details see~\cite[Section 3]{CooleyThibaud_ddhde}), the reconstruction is given by 
\begin{equation}
    R = U t(U^T t^{-1}(X)), \quad t(\cdot) := \log(1 + \exp(\cdot)). 
\end{equation}
The matrix $U$ is obtained in practice by calculating the eigenvectors to the $p$ largest eigenvalues of the tail pairwise dependence matrix $\Sigma \in \R^{d \times d}$ given by 
\begin{equation}
    \label{eq:tpdm}
    \Sigma_{ij} := \int_{\S^{d-1}_+} a_i a_j \, \hat S_n(da), \quad i,j = 1, \ldots, d, 
\end{equation}
where $\hat S_n$ is an estimator of the spectral measure $S$. 

Similarly to the approach by~\cite{DreesSabourin_pcame} and our work, we can give elbow plots for the eigenvalues of the tail pairwise dependence matrix as a heuristic approach to find a reasonable number of principal components, see Figure~\ref{fig:elbow_danube_cy}. Again, for comparison we choose $p=6$ and plot the vectors of principal components in Figure~\ref{fig:pcs_cooley}. Due to strong similarities between the empirical versions of $\Sigma$ and $\Theta \Theta^T$ from the previous section, the principal components are very similar to those of Drees and Sabourin. Accordingly, a similar comment applies for the comparison to our method, in that the approach by Cooley and Thibaud captures a lower dimensional embedding of the spectral measure, whereas our approach aims at identifying a small number of driving factors of this spectral measure.

\begin{figure}[h!bt]
\centering
\begin{subfigure}{0.4\textwidth}
      \centering
      \includegraphics[width=\textwidth]{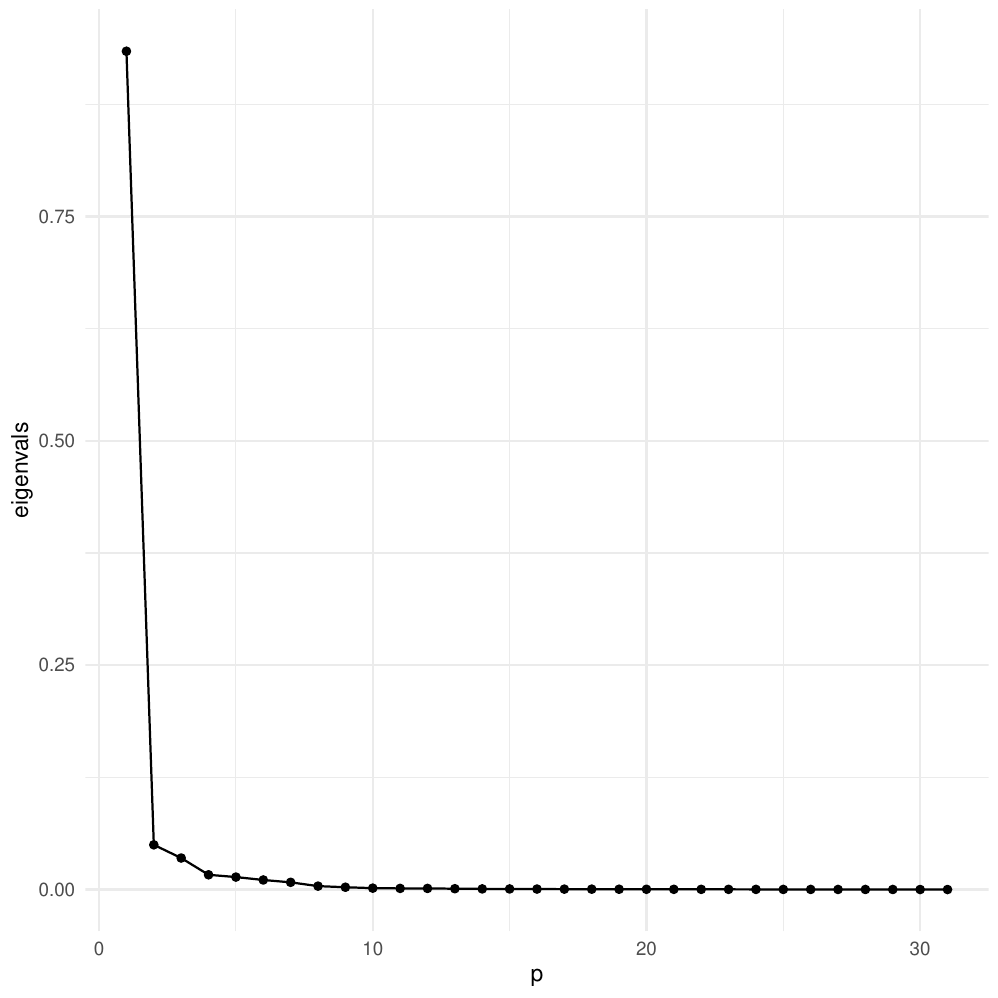}
    \end{subfigure}
    \caption{Elbow plots of the eigenvalues of the TPDM. The left plot would suggest to cut off at just one component or at $p = 4$ components. }\label{fig:elbow_danube_cy}
\end{figure}
\begin{figure}[h!bt]
\centering
\begin{subfigure}{0.3\textwidth}
      \centering
      \includegraphics[width=\textwidth]{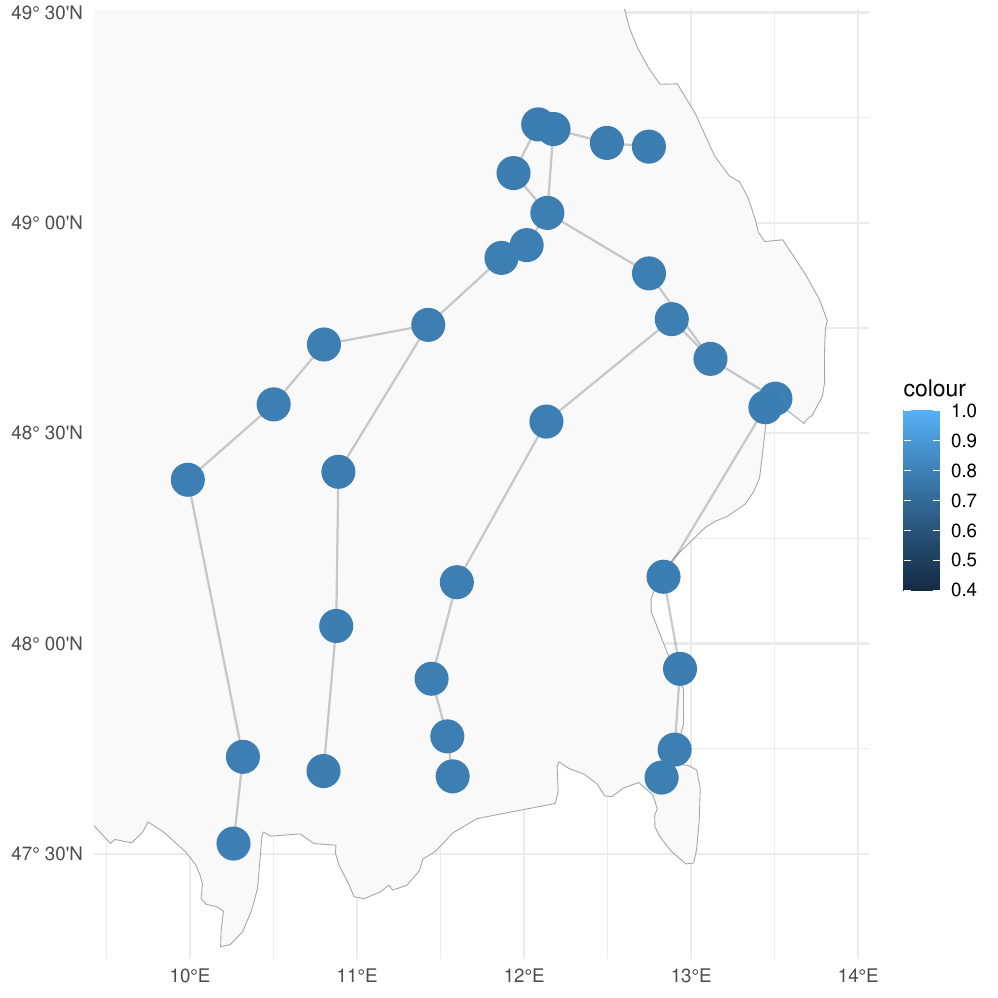}
    \end{subfigure}
    \hfill\begin{subfigure}{0.3\textwidth}
      \centering
      \includegraphics[width=\textwidth]{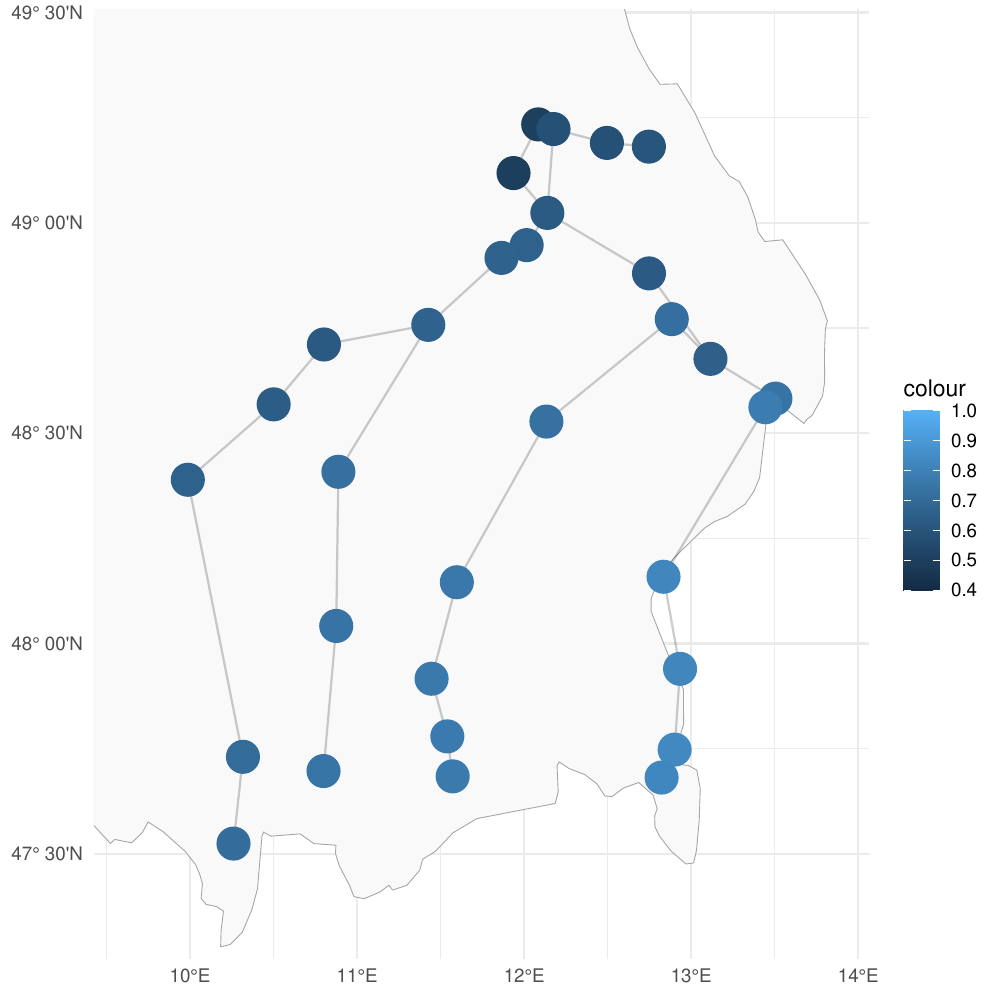}
    \end{subfigure}
    \hfill\begin{subfigure}{0.3\textwidth}
      \centering
      \includegraphics[width=\textwidth]{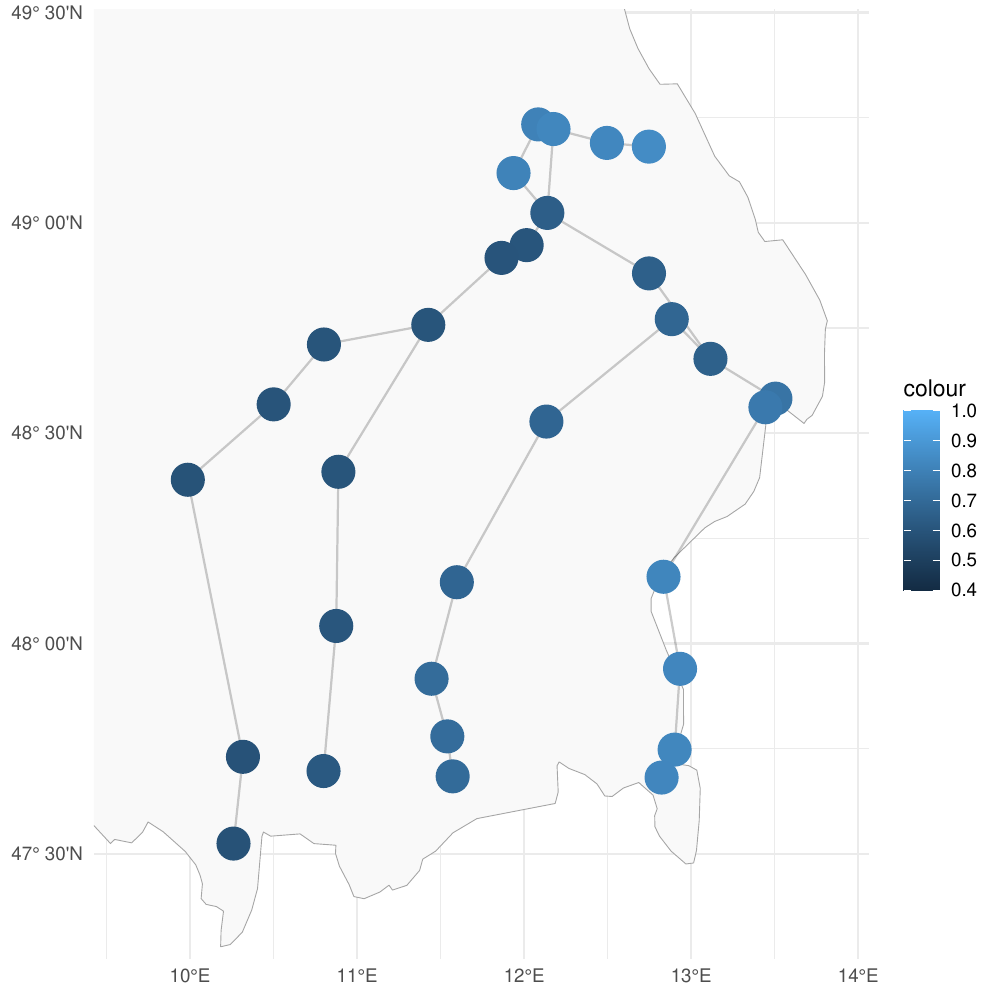}
    \end{subfigure}
    \hfill
    \begin{subfigure}{0.3\textwidth}
      \centering
      \includegraphics[width=\textwidth]{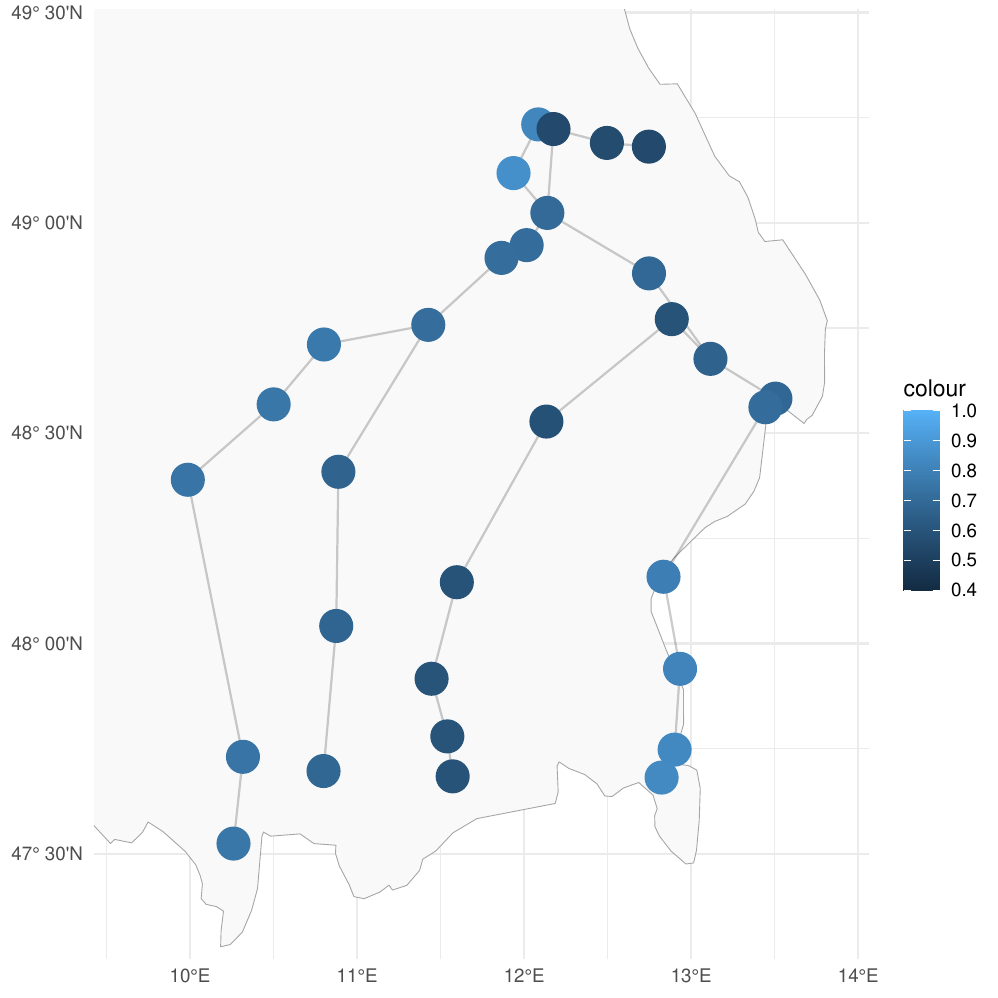}
    \end{subfigure}
    \hfill
    \begin{subfigure}{0.3\textwidth}
      \centering
      \includegraphics[width=\textwidth]{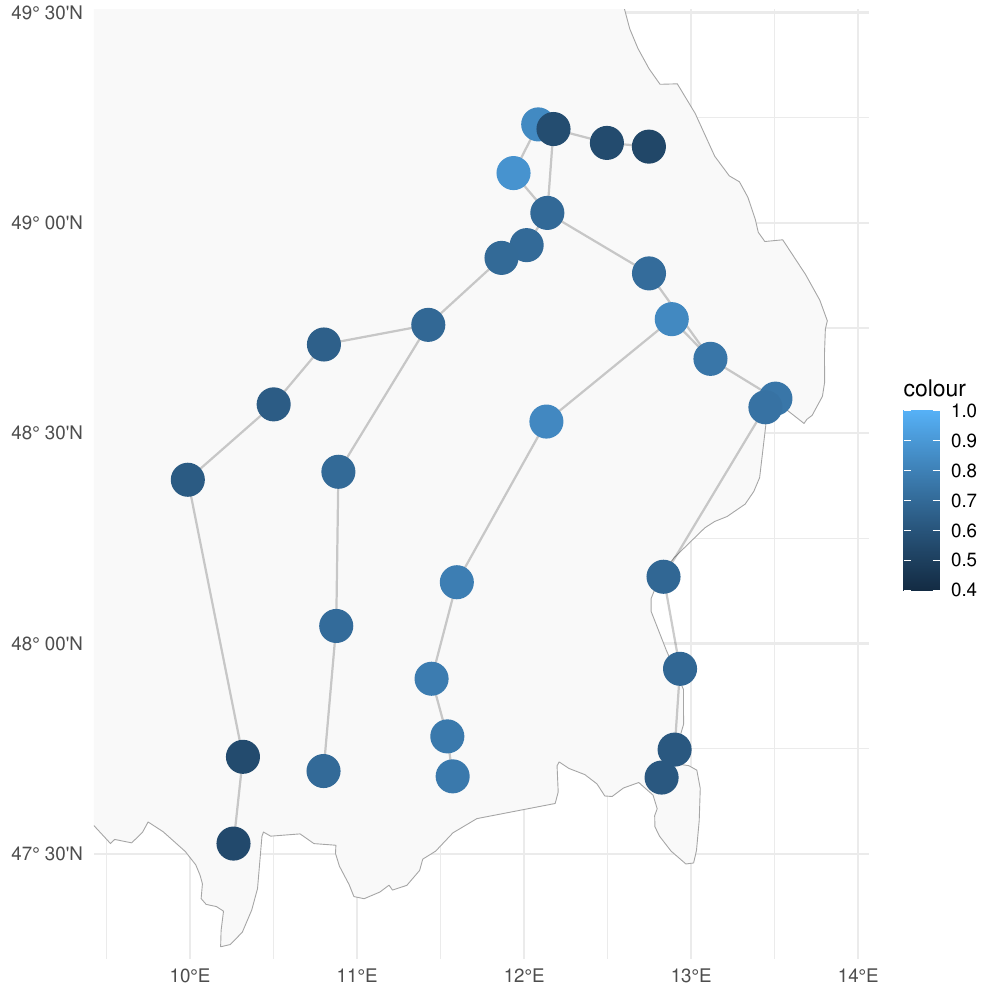}
    \end{subfigure}
    \hfill
    \begin{subfigure}{0.3\textwidth}
      \centering
      \includegraphics[width=\textwidth]{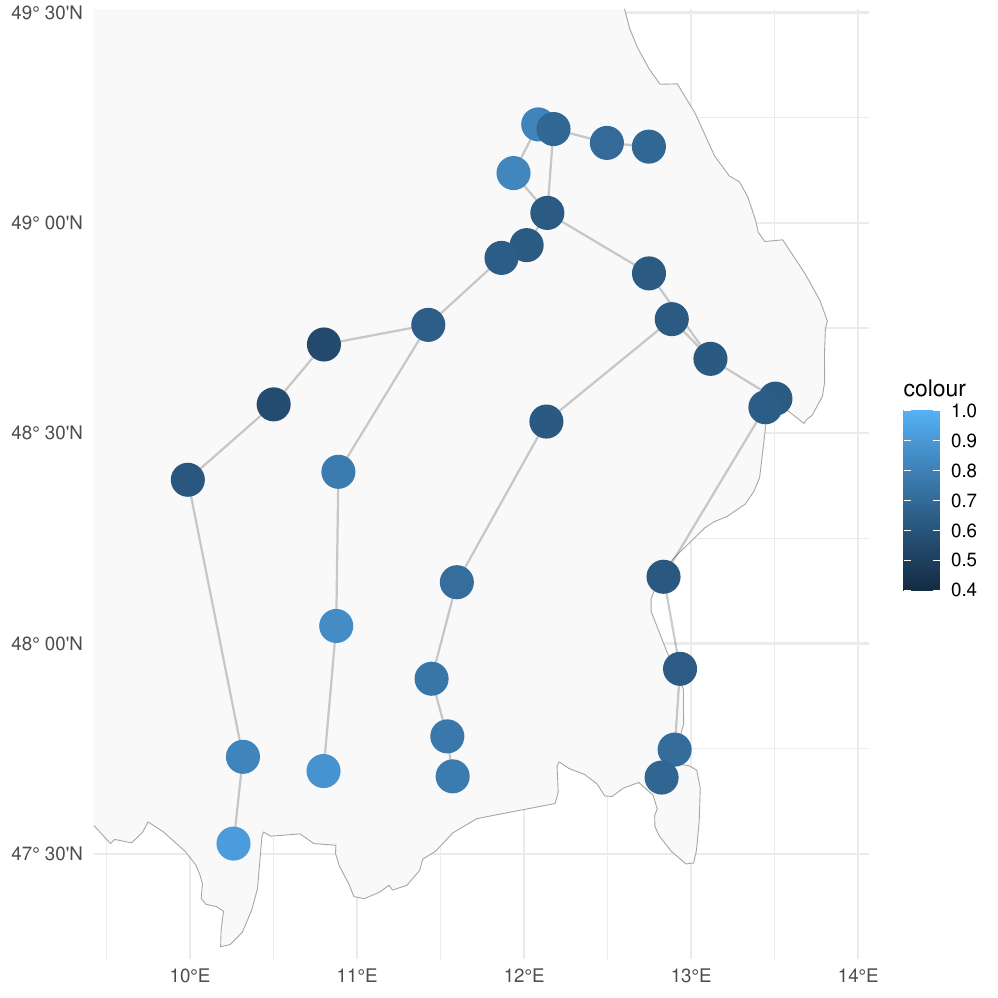}
    \end{subfigure}
    \caption{The first 6 principal components of the procedure by Cooley and Thibaud, plotted to their respective measurement station on the Danube river network. The principal components are transformed by the transformation map $t$ as proposed by the original work to ensure that they are positive. The first component again seems to be an average, this agrees with previous analyses using this version of PCA for extremes and the other components concentrate on different spatial domains.}\label{fig:pcs_cooley}
\end{figure}

\subsection{Spherical $k$-means clustering for extremes (\cite{JanssenWan_kmce})}

In some applications, instead of mapping the data to a lower dimensional subspace, it might be desirable to use a clustering approach and find $k$ prototypical extreme events that represent the data well. In~\cite{JanssenWan_kmce}, it was shown that spherical $k$-means clustering is a consistent procedure that can recover low dimensional max-linear models. We briefly outline the procedure and apply $k$-means clustering to the Danube dataset to discuss similarities and differences. The goal for spherical $k$-means clustering for extremes is to find $k$ prototypical points on the sphere that represent the angular component of extreme observations well with respect to a suitable metric $d$ on the sphere. To calculate these centroids, we first estimate the spectral measure $S$ from data that is assumed to be in the max-domain of attraction of some multivariate extreme value distribution and rescaled to standard Pareto margins. For the resulting estimator $\hat S_n$, one then minimizes 
\begin{equation}
    W(A, \hat S_n ) := \int_{\S^{d-1}_+} \min_{j=1,\ldots, k} d(A_{\cdot j}, a) \, \hat S_n(da), \quad A_{\cdot 1}, \ldots, A_{\cdot k} \in \S^{d-1}_+,
\end{equation}
over the matrix $A \in \R^{d \times k}$ of cluster centers $A_{\cdot 1}, \ldots, A_{\cdot k}$, where we employ the spherical distance measure $d$, for details see~\cite{JanssenWan_kmce}. 

We apply the spherical $K$-means procedure to the Danube dataset to visually inspect the elbow plot provided in Figure~\ref{fig:elbow_danube_kmeans} for the losses and a spatial plot of the vectors of cluster centers for $k=6$ in Figure~\ref{fig:kmeans_clustercenters}. The cluster centers allow to detect concentration of mass of the spectral measure, but the number of identified cluster centers is not related to the notion of dimension, and we cannot reconstruct the original observations from this approach. 
\begin{figure}[h!bt]
\centering
\begin{subfigure}{0.4\textwidth}
      \centering
      \includegraphics[width=\textwidth]{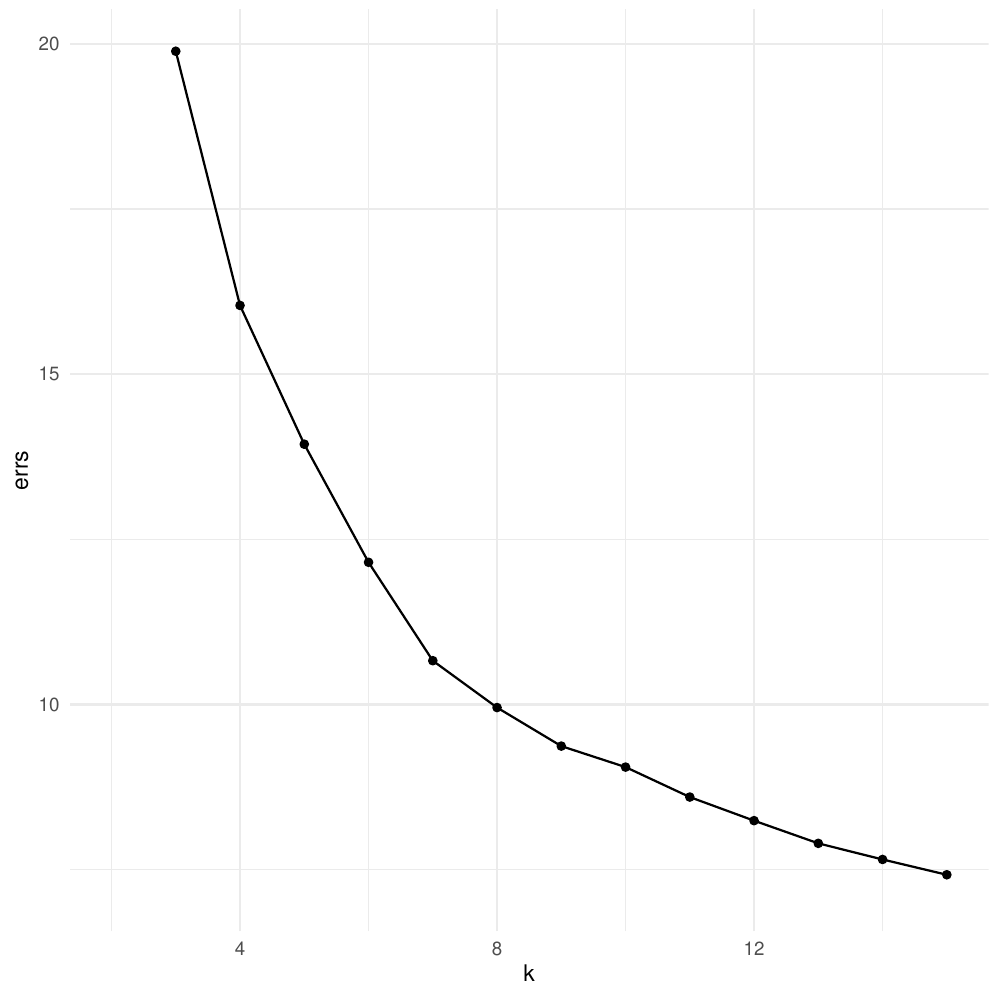}
    \end{subfigure}
    \caption{Elbow plot of the spherical $k$-means algorithm for different values of $k$. Note that there is no clear elbow visible in this plot. }\label{fig:elbow_danube_kmeans}
\end{figure}
\begin{figure}[h!bt]
\centering
\begin{subfigure}{0.3\textwidth}
      \centering
      \includegraphics[width=\textwidth]{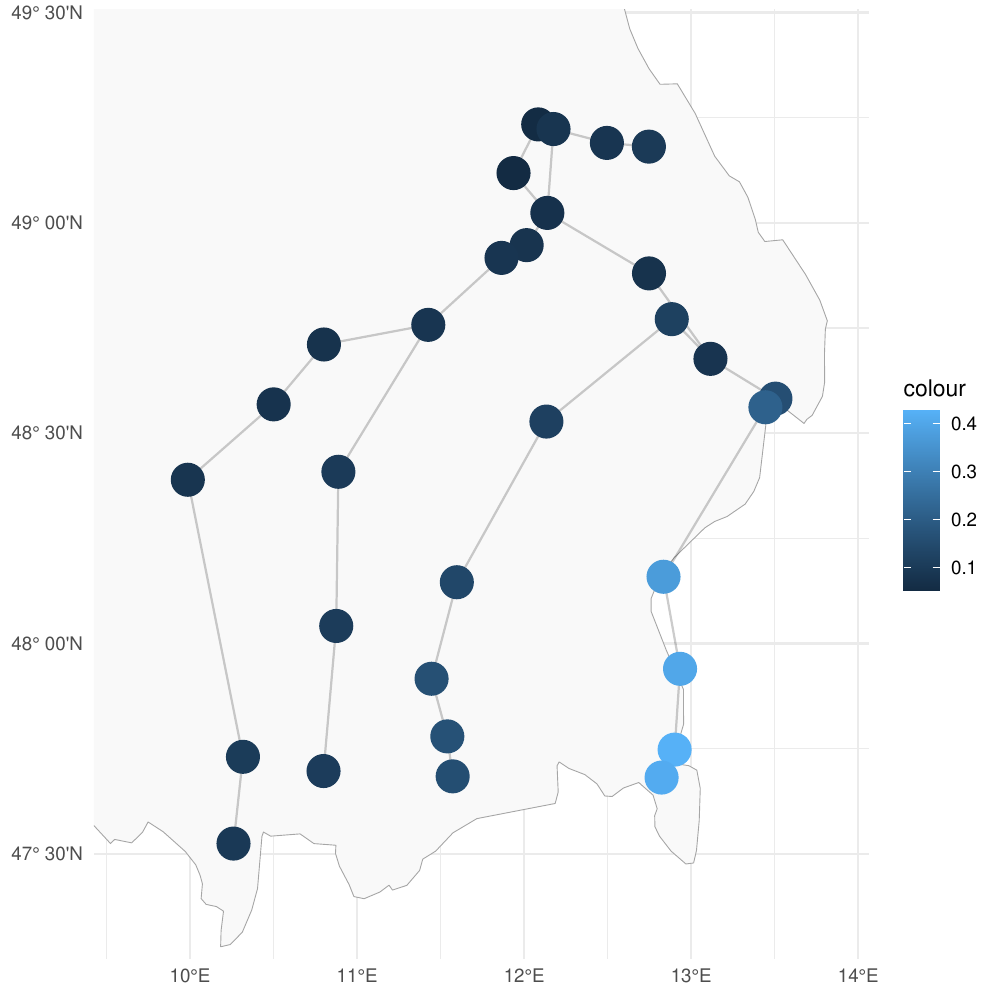}
    \end{subfigure}
    \hfill\begin{subfigure}{0.3\textwidth}
      \centering
      \includegraphics[width=\textwidth]{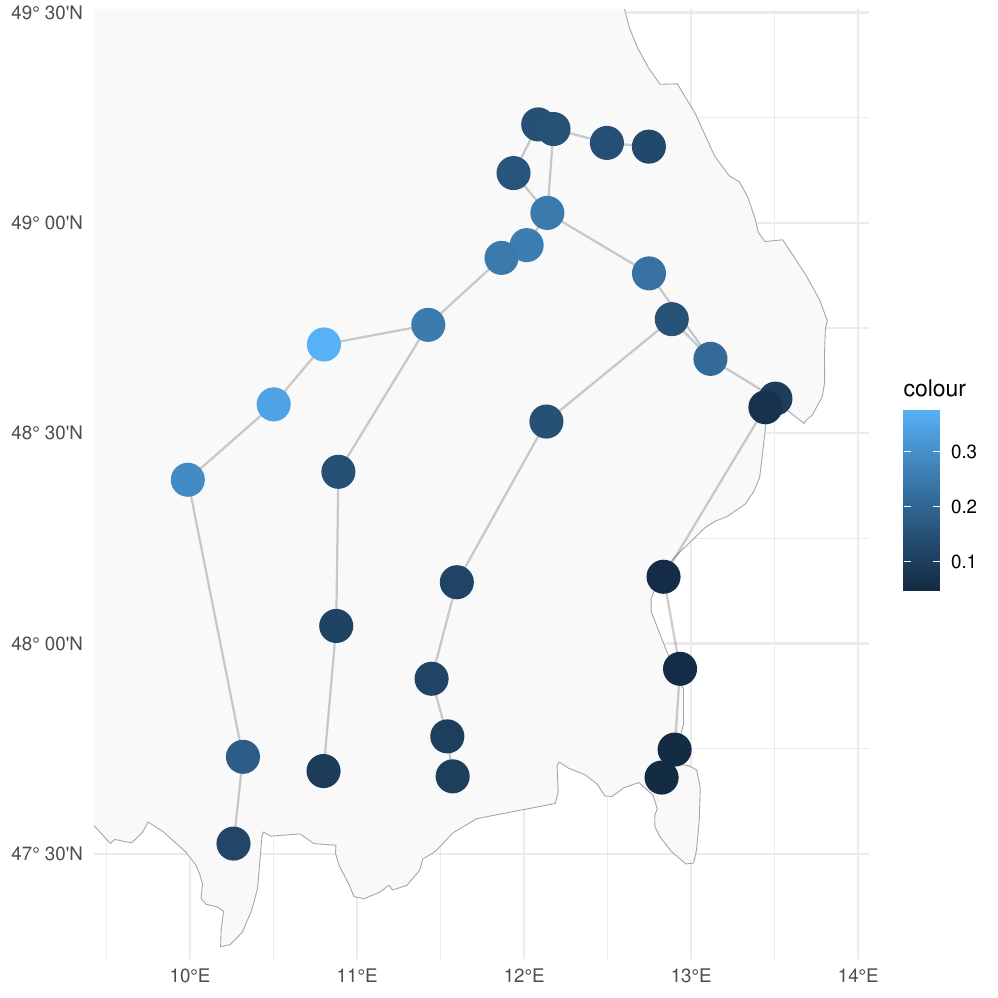}
    \end{subfigure}
    \hfill\begin{subfigure}{0.3\textwidth}
      \centering
      \includegraphics[width=\textwidth]{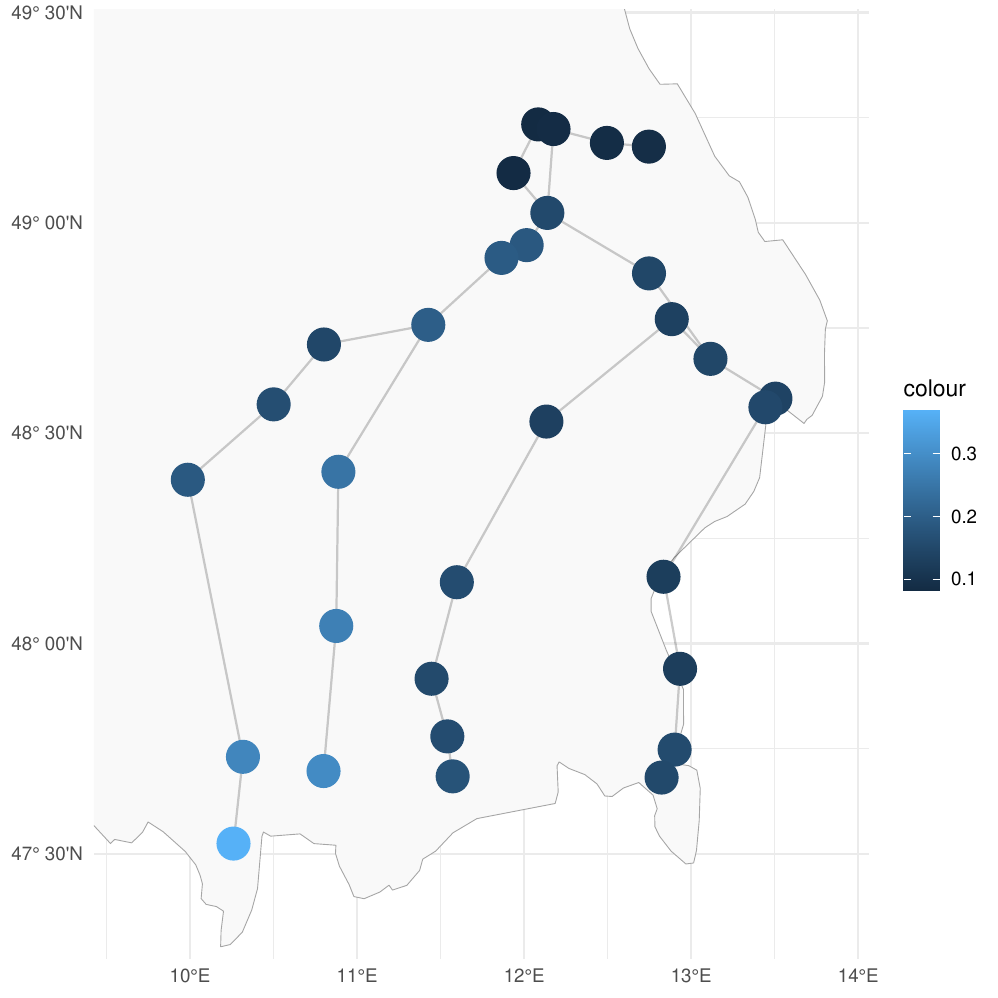}
    \end{subfigure}
    \hfill
    \begin{subfigure}{0.3\textwidth}
      \centering
      \includegraphics[width=\textwidth]{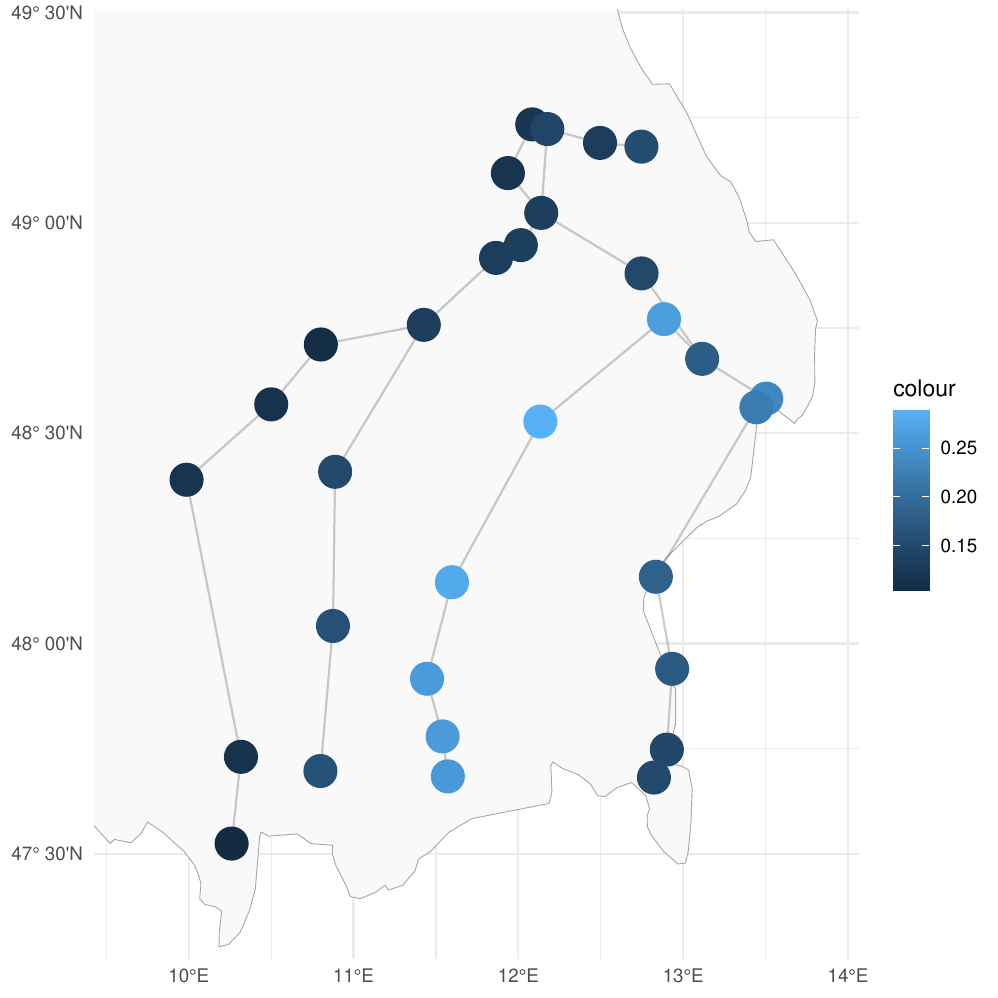}
    \end{subfigure}
    \hfill
    \begin{subfigure}{0.3\textwidth}
      \centering
      \includegraphics[width=\textwidth]{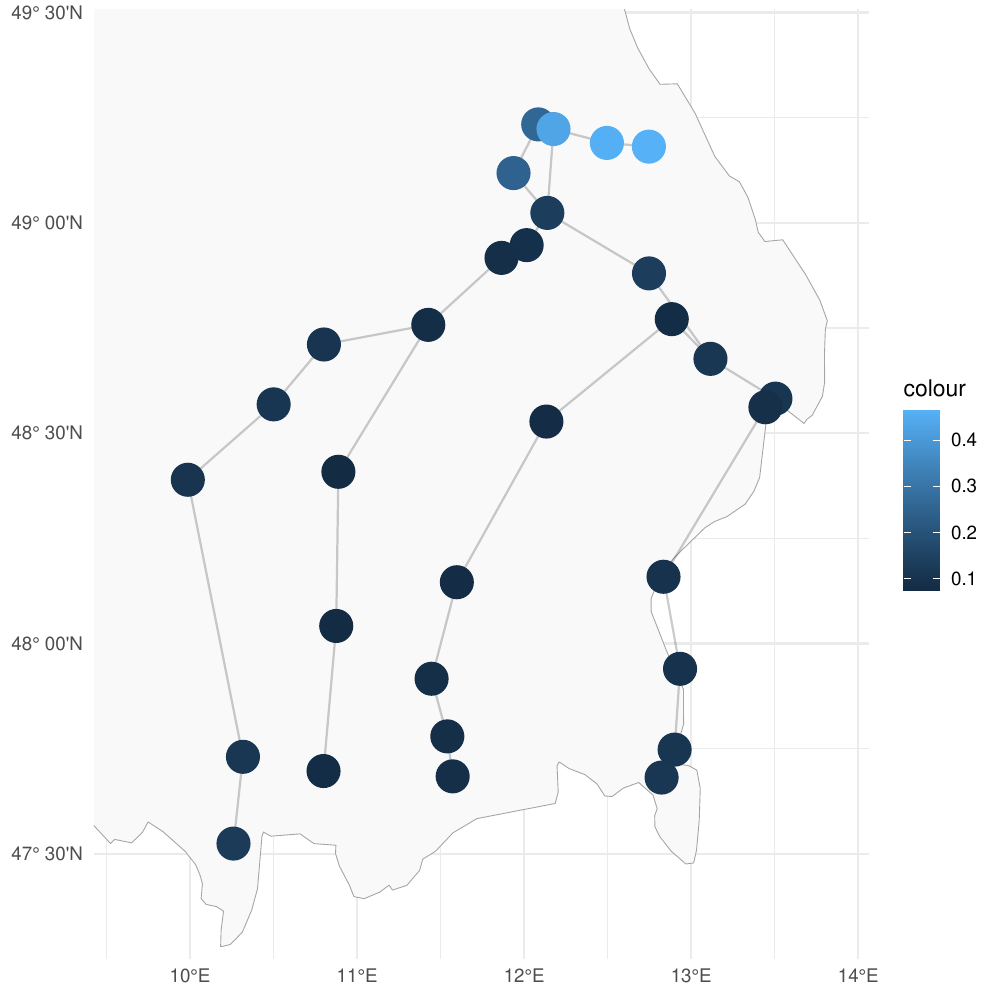}
    \end{subfigure}
    \hfill
    \begin{subfigure}{0.3\textwidth}
      \centering
      \includegraphics[width=\textwidth]{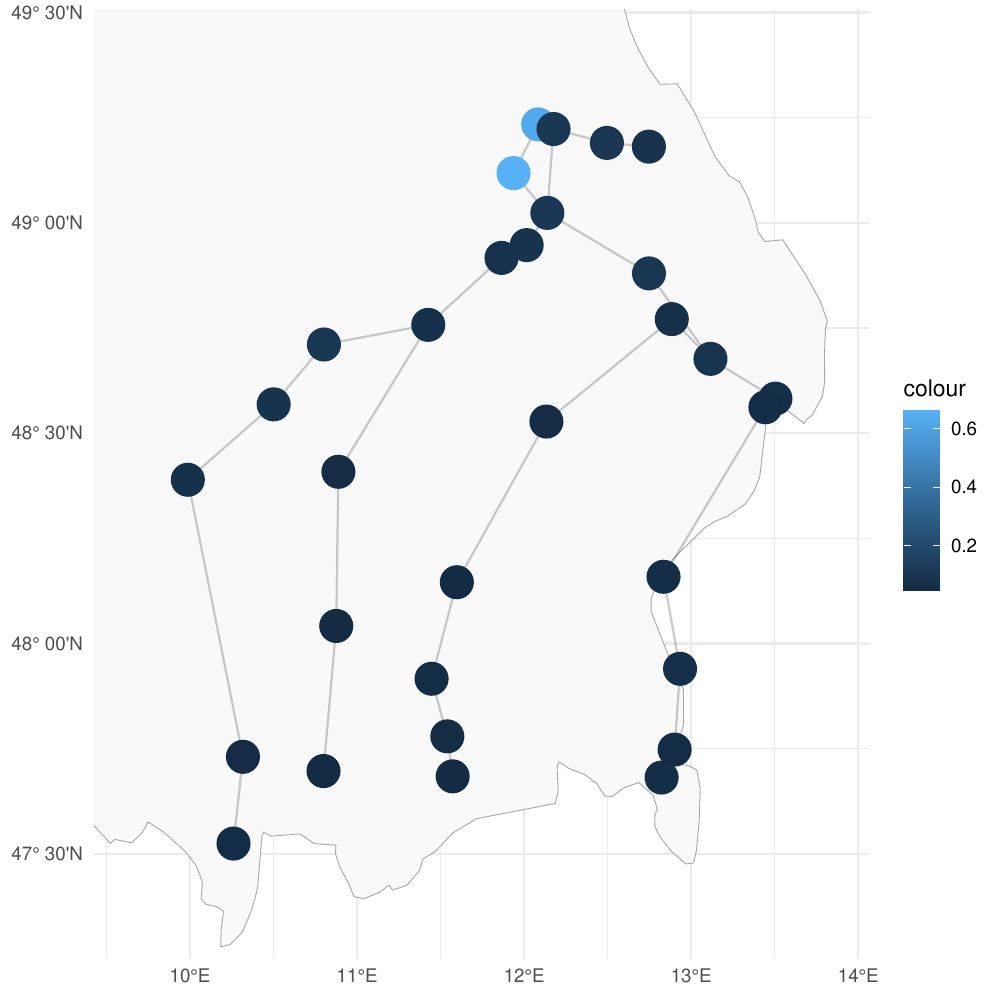}
    \end{subfigure}
    \caption{The cluster centers for $k = 6$ obtained by the spherical $k$-means clustering algorithm for extremes. The cluster centers seem to capture the structure of single river arms.}\label{fig:kmeans_clustercenters}
\end{figure}
}

\FloatBarrier

\bibliographystyle{plain}

\bibliography{./bibliography}
\end{document}